\documentclass[a4paper,UKenglish,cleveref,autoref,pdfa]{lipics-v2021}

\pdfoutput=1 

\hideLIPIcs

\title{Sure-almost-sure and Sure-limit-sure Window Mean Payoff in Markov Decision Processes}
\titlerunning{Sure-almost-sure and Sure-limit-sure Window Mean Payoff in MDPs}

\author{Pranshu Gaba}{Tata Institute of Fundamental Research, Mumbai, India}{pranshu.gaba@tifr.res.in}{https://orcid.org/0009-0000-8012-780X}{}
\author{Shibashis Guha}{Tata Institute of Fundamental Research, Mumbai, India}{shibashis@tifr.res.in}{https://orcid.org/0000-0002-9814-6651}{}
\authorrunning{P. Gaba and S. Guha}

\Copyright{Pranshu Gaba and Shibashis Guha}



\funding{%
  We acknowledge support of CEFIPRA project no.~7302-I and of the Department of Atomic Energy, Government of India, project no.~RTI4014.
}

\acknowledgements{We thank anonymous reviewers for suggesting various improvements and pointing out a subtle bug in the proof of \Cref{lem:sasf-fwmp-beta-empty}.}

\nolinenumbers

\usepackage{mypackages}
\usepackage{mymacros}

\begin{document}

\maketitle

\begin{abstract}
  Given rationals \(\GuaranteeThreshold\) and \(\AlmostSureThreshold\), the \emph{sure-almost-sure} problem for a threshold Boolean objective \(\Objective\) in a Markov decision process (MDP) asks if one can simultaneously ensure that
  all outcomes of the MDP have \(\Objective\)-value at least~\(\GuaranteeThreshold\) (i.e. \emph{sure} \(\GuaranteeThreshold\) satisfaction), and 
  with probability~\(1\) the outcome has \(\Objective\)-value at least \(\AlmostSureThreshold\) (i.e. \emph{almost-sure} \(\AlmostSureThreshold\) satisfaction).
  The \emph{sure-limit-sure} problem asks if for all \(\epsilon > 0\), one can simultaneously ensure that all outcomes have \(\Objective\)-value at least \(\GuaranteeThreshold\), and with probability at least \(1 - \epsilon\) the outcome has \(\Objective\)-value at least \(\LimitSureThreshold\).
  Moreover, if simultaneous satisfaction of objectives is possible, then one would also like to construct a strategy (for sure-almost-sure) or a family of strategies (for sure-limit-sure) that achieves this. 
  Even if both sure satisfaction and almost-sure (resp., limit-sure) satisfaction for an objective are known, combining the two is 
  often non-trivial and requires novel techniques and approaches. 

  In this paper, we solve the sure-almost-sure and sure-limit-sure problems for \emph{window mean-payoff objectives}.
  While it is known that almost-sure satisfaction and limit-sure satisfaction for window mean-payoff coincide in MDPs, we show that sure-almost-sure satisfaction is distinct from sure-limit-sure satisfaction.
  The window mean-payoff objective strengthens the standard mean-payoff objective by requiring that eventually, from every point in the infinite run, the average payoff becomes greater than a given threshold within a finite window length.
  We study two variants of window mean payoff: 
  in the \emph{fixed} variant, the window length \(\WindowLength\) is given, while in the \emph{bounded} variant, the length is not given but is required to be bounded throughout the run.
  We show that the sure-almost-sure problem and the sure-limit-sure problem are both in \(\PTime\) for the fixed variant (if \(\WindowLength\) is given in unary) and are both in \(\NP \cap \coNP\) for the bounded variant, matching the computational complexity of sure satisfaction and almost-sure satisfaction when considered separately for these objectives.
  We also give bounds for the memory requirement of winning strategies for all considered problems.
\end{abstract}

\section{Introduction}

\subparagraph*{Beyond worst-case synthesis.}
Classical two-player zero-sum games~\cite{AG11,GTW02} involve decision-making against a purely antagonistic environment where a threshold performance or a specific behaviour of the system needs to be ensured against \emph{all possible strategies of the environment}, that is even in the worst case. 
On the other hand, Markov decision processes (MDPs)~\cite{Puterman94} model uncertainty, and decision-making involves ensuring a \emph{specified behaviour with sufficient probability} or a \emph{high expected performance} against a stochastic environment.
Such a stochastic model of the environment, however, does not provide any guarantee on the worst-case performance, and a strategy that is adequate against an adversarial environment may provide a suboptimal performance against a stochastic environment.

In practice, both might be desired simultaneously: A system needs to provide a guarantee in the worst-case, as well as ensure some threshold performance with a high probability against a stochastic environment. 
The beyond worst-case (\(\BWC\)) framework was introduced in~\cite{BFRR17} to provide strict worst-case guarantee as well as good expected performance for quantitative specifications.
Subsequently, in~\cite{BRR17}, the beyond worst-case setting was studied for qualitative omega-regular objectives encoded as two parity objectives, where the first one needs to be satisfied surely, while the second one needs to be satisfied almost-surely or with a high threshold probability.

\subparagraph*{Window mean payoff.}
For objectives such as long-run limits of reward functions~\cite{EM79, ZP96}, a play may satisfy the objective while still exhibiting undesired behaviours over arbitrarily long finite segments~\cite{CHH09,CDRR15}. 
Finitary (or window) objectives enforce that undesired behaviour persist only within intervals of bounded length (windows) along the play. 
In this work, we focus on \emph{window mean-payoff objectives}~\cite{CDRR15}, which are finitary versions of the classical mean-payoff objective. 
Given a window length \(\WindowLength \geq 1\) and a threshold \(\Threshold \in \Rationals\), the fixed window mean-payoff objective \(\FWMP(\WindowLength,\Threshold)\) holds if, except for a finite prefix, from every position $i$ in the play there exists a window starting at $i$ of length at most \(\WindowLength\) whose mean payoff is at least \(\Threshold\). 
The bounded window mean-payoff objective \(\BWMP(\Threshold)\) holds for a play if there exists some \(\WindowLength \ge 1\) for which \(\FWMP(\WindowLength,\Threshold)\) holds for the play. 

\subparagraph*{Related work.}
The beyond worst-case framework under expectation and worst-case semantics was introduced in~\cite{BFRR17} for quantitative objectives. 
The problem where the players are restricted to finite-memory strategies was shown to be in \(\NP \cap \coNP\) for the mean-payoff objective.
The case of infinite memory strategy for the \(\BWC\) synthesis problem was left open in~\cite{BFRR17} and was solved in~\cite{CR15}. 
Further, in~\cite{CR15}, a natural relaxation of the \(\BWC\) problem, the beyond almost-sure synthesis (\(\BAS\)) problem was introduced that asks for a threshold performance almost surely, that is, with probability $1$, while maximizing expectation and was shown to be in \(\PTime\) for the mean-payoff objective. 
The beyond probability threshold synthesis problem that asks for a threshold performance with a given probability $p$ while maximizing expectation was studied for mean-payoff objective in~\cite{CKK17}. 
In~\cite{BRR17}, the beyond worst-case synthesis problems were studied for qualitative omega-regular objectives encoded as parity objectives, and the problems were shown to be in \(\NP \cap \coNP\); in~\cite{CP19,DG26}, these problems were studied in the context of stochastic games which are a generalization of MDPs where the environment is both stochastic and adversarial. 
In~\cite{BGR20}, Boolean combinations of omega-regular objectives that need to be enforced either surely, almost surely, existentially, or with non-zero probability were studied, and it was shown that both randomization and infinite memory may be required by an optimal strategy. 
In~\cite{BKW24}, a combination of parity objective and multiple reachability objectives along with threshold probabilities were considered where the parity objective needs to be satisfied surely and each reachability objective is to be satisfied with the corresponding threshold probability. 

Mean-payoff objectives were studied initially in two-player games on graphs without stochasticity~\cite{EM79,ZP96}, and finitary versions were introduced as window mean-payoff objectives~\cite{CDRR15} and have been studied extensively since the last decade. 
For window mean-payoff objectives, the satisfaction problem~\cite{BDOR20} and the expectation problem~\cite{BGR19} were studied in MDPs.
Both problems were shown to be in $\PTime$ for the \(\FWMPL\) objectives and in \(\NP \cap \coNP\) for the \(\BWMP\) objectives. 
The satisfaction problem and the expectation problem for window mean-payoff objectives have recently been studied in~\cite{DGG25} and~\cite{DGG25CONCUR} respectively for the more general setting that are stochastic games.
In~\cite{GG25}, \(\BWC\) synthesis with the expectation semantics has been studied for window mean-payoff objectives in MDPs.

\subparagraph*{Contributions.}
We study two problems for the fixed and bounded window mean-payoff objectives each.
Given an MDP \(\MDP\), thresholds \(\GuaranteeThreshold, \AlmostSureThreshold \in \Rationals\), and an initial vertex \(\Vertex\):
\begin{enumerate}
\item 
  (Sure-almost-sure (\(\SAS\)) synthesis) synthesize a strategy that ensures
  \begin{enumerate*}[label=(\emph{\roman*})]
  \item satisfaction of the window mean-payoff objective with threshold \(\GuaranteeThreshold\) surely, i.e. against all strategies of an adversarial environment, and 
  \item satisfaction of the window mean-payoff objective with threshold \(\AlmostSureThreshold\) almost surely, that is, with probability $1$, against a stochastic model of the environment.
  \end{enumerate*}

\item
  (Sure-limit-sure (\(\SLS\)) synthesis) for every $\epsilon > 0$,  synthesize a strategy that ensures
  \begin{enumerate*}[label=(\emph{\roman*})]
  \item satisfaction of the window mean-payoff objective with threshold \(\GuaranteeThreshold\) surely, and 
  \item satisfaction of the window mean-payoff objective with threshold \(\AlmostSureThreshold\) with probability $1 - \epsilon$.
  \end{enumerate*}
\end{enumerate}
While almost-sure satisfaction and limit-sure satisfaction are equivalent in MDPs~\cite{CDH04}, we show that sure-almost-sure satisfaction is strictly stronger than sure-limit-sure satisfaction (\Cref{exa:sas-sls-difference}).
We show that for both \(\FWMPL\) and \(\BWMP\), the \(\SAS\) and \(\SLS\) problems are no more complex than solving sure satisfaction~\cite{CDRR15} or almost-sure satisfaction~\cite{BDOR20} of the same objective separately.
We show in \Cref{thm:fwmp-sas-complexity,thm:fwmp-sls-complexity} that the above problems are in \(\PTime\) for \(\FWMPL\) if \(\WindowLength\) is given in unary, and in \Cref{thm:bwmp-sas-complexity,thm:bwmp-sls-complexity} that the problems are in \(\NP \cap \coNP\) for \(\BWMP\).
Deterministic strategies suffice, and memory requirement may be higher for sure-almost-sure satisfaction compared to satisfying the objective either surely or almost surely (\Cref{exa:fwmp-memory-lower-bound,exa:bwmp-memory-lower-bound}).
Our results are summarized in \Cref{table:contributions-summary}.

\begin{table}
  \centering
  \caption{Previous and our results (shaded) for window mean-payoff satisfaction in MDPs.
    In the memory bounds, we denote by \(\abs{V}\) the number of vertices in the MDP, by \(\WindowLength\) the window length (given in unary), by \(\ProbabilityFunction_{\min}\) the minimum non-zero probability on the edges in the MDP, and by \(\PayoffFunction_{\max}\) the maximum absolute payoff on the edges of the MDP.
  }
  \label{table:contributions-summary}
  \begin{tabular}{lllcc}
    \toprule
    & &      & \multicolumn{2}{c}{Memory bound}   \\ \cmidrule(l){4-5} \multicolumn{2}{l}{Winning condition} & Complexity& lower & upper \\ \midrule
    $\FWMPL$ & sure \cite{CDRR15} & $\PTime$ & $\WindowLength -1$ &$ \WindowLength$ \\
    & almost-sure \cite{BDOR20} & $\PTime$ & $\WindowLength -1$ & $ \WindowLength$ \\
    & \ours \(\SAS\) (\Cref{thm:fwmp-sas-complexity}) & \ours $\PTime$ & \ours \(\Omega(\max\{\abs{V}, \WindowLength\})\) & \ours $\bigO(\abs{V} \cdot \WindowLength)$ \\ 
    & \ours \(\SLS\) (\Cref{thm:fwmp-sls-complexity}) & \ours $\PTime$  & \ours \(\Omega\left(\frac{1 - \epsilon}{(\ProbabilityFunction_{\min})^{\abs{V}}} + \WindowLength\right)\) & \ours \(\bigO\left(\frac{\abs{V} \cdot \log(1/\epsilon)}{(\ProbabilityFunction_{\min})^{\abs{V}}} + \WindowLength\right)\) \\ \midrule
    $\BWMP$ & sure \cite{CDRR15} & $\NP \intersection \coNP$ & memoryless & memoryless \\ 
    & almost-sure \cite{BDOR20} & $\NP \intersection \coNP$ & memoryless & memoryless \\
    & \ours\(\SAS\) (\Cref{thm:bwmp-sas-complexity}) & \ours$\NP \cap \coNP$ &\ours $\Omega(\abs{V} \cdot \PayoffFunction_{\max})$ &\ours $\bigO(\abs{V}^{3} \cdot \abs{\PayoffFunction_{\max}})$ \\ 
    &\ours \(\SLS\) (\Cref{thm:bwmp-sls-complexity}) &\ours $\NP \cap \coNP$ &\ours  \(\Omega\left(\frac{1 - \epsilon}{(\ProbabilityFunction_{\min})^{\abs{V}}}\right)\) &\ours  \(\bigO\left(\frac{\abs{V} \cdot \log(1/\epsilon)}{(\ProbabilityFunction_{\min})^{\abs{V}}}\right)\) \\ 
    \bottomrule
  \end{tabular}
\end{table}

For the sure-almost-sure satisfaction of quantitative objectives such as mean payoff~\cite{BFRR17} and parity~\cite{BRR17}, the following observations are useful: there is a winning strategy for sure satisfaction (that may invalidate almost-sure satisfaction) and a winning strategy for almost-sure satisfaction (that may invalidate sure satisfaction), and in order to satisfy sure-almost-sure satisfaction, a strategy may need to switch between the two infinitely often, ensuring in the long run both sure and almost-sure satisfaction.
This however does not apply to sure-almost-sure synthesis for finitary objectives, as we have a bounded amount of steps to satisfy both objectives.
We instead require different techniques and a more careful analysis.

\subparagraph{Organization.}
\Cref{sec:preliminaries} defines technical preliminaries and \Cref{sec:window-mean-payoff-objectives} defines the window mean-payoff objectives.
Sections \ref{sec:bwc-satisfaction-boolean} and~\ref{sec:sls} study the decision problems and strategies for sure-almost-sure and sure-limit-sure satisfaction of window mean-payoff objectives respectively.
Finally, we conclude with a brief discussion in \Cref{sec:conc}.

\section{Preliminaries}%
\label{sec:preliminaries}

\subparagraph*{Probability distributions.}
A \emph{probability distribution} over a finite set \(A\) is a function \(\Prob \colon A \to [0,1]\) such that \(\sum_{a \in A} \Prob(a) = 1\).
We denote by \(\DistributionSet{A}\) the set of all probability distributions over \(A\).
The \emph{support} of the probability distribution \(\Prob\) on \(A\) is \(\Support{\Prob} = \{a \in A \suchthat \Prob(a) > 0\}\).
For algorithm and complexity reasons, we assume that probability distributions take rational values.
Probability distributions may also be defined over an infinite set \(A\), see e.g. \cite{Billingsley86}.

\subparagraph{Markov decision processes.}
A (finite) \emph{Markov decision process (MDP)} is a tuple \(\MDP = ((\Vertices, \Edges), (\VerticesMain, \VerticesRandom),  \ProbabilityFunction, \PayoffFunction)\), where:
\begin{itemize}
\item \((\Vertices, \Edges)\) is a directed graph with 
  a (finite) set \(\Vertices\) of vertices and a set  \(\Edges \subseteq \Vertices \times \Vertices\) of directed edges such that
  for all vertices \(\Vertex \in \Vertices\), 
  the set \(\OutNeighbours{\Vertex} = \{ \Vertex' \in \Vertices \suchthat (\Vertex, \Vertex') \in \Edges\}\) 
  of out-neighbours of \(\Vertex\) is non-empty, i.e. no deadlocks.
\item \((\VerticesMain, \VerticesRandom)\) is a partition of \(\Vertices\). The vertices in \(\VerticesMain\) belong to the system or Player~$\Main$ (denoted \(\PlayerMain\)), and the vertices in \(\VerticesRandom\) belong to the environment and are called \emph{probabilistic vertices}.
\item \(\ProbabilityFunction \colon \VerticesRandom \to \DistributionSet{\Vertices}\) is a \emph{probability function} that describes the behaviour of probabilistic vertices in the game. 
  It maps each probabilistic vertex \(\Vertex \in \VerticesRandom\) to a probability distribution \(\ProbabilityFunction(\Vertex)\) over the set \(\OutNeighbours{\Vertex}\) of out-neighbours of \(\Vertex\) such that \(\ProbabilityFunction(\Vertex)(\Vertex') > 0\) for all \(\Vertex' \in \OutNeighbours{\Vertex}\) (i.e., all out-neighbours have non-zero probability).
  We denote by $\ProbabilityFunction_{\min}$ the minimum non-zero probability in $\MDP$, that is, \(\min\{  \ProbabilityFunction(v)(v') \suchthat v, v' \in V, \ v' \in \OutNeighbours{v}\}\).
\item \(\PayoffFunction \colon \Edges \to \Rationals\) 
  is a \emph{payoff function} assigning a rational payoff to every edge in the game.
  We denote by $\wmax$ the maximum absolute payoff on the edges of $\MDP$, that is, \(\max_{e \in E}\abs{w(e)}\).
\end{itemize}
A play in an MDP \(\MDP\) begins by putting a token on a designated vertex $v_{0}$.
The token proceeds in rounds as follows for all $i \geq 0$.
At the beginning of round \(i\), if the token is on a vertex \(v_{i}\), and \(\Vertex[i] \in \VerticesMain\), then \(\PlayerMain\) chooses an out-neighbour \(\Vertex[i+1] \in \OutNeighbours{\Vertex}\);
otherwise \(\Vertex[i] \in \VerticesRandom\), and an out-neighbour \(\Vertex[i+1] \in \OutNeighbours{\Vertex}\) is chosen with probability \(\ProbabilityFunction(\Vertex[i])(\Vertex[i+1])\). 
Player \(\PlayerMain\) receives the amount \(\PayoffFunction(\Vertex[i], \Vertex[i+1])\) given by the payoff function and the token moves to the chosen \(v_{i+1}\) for the next round. 
This continues ad infinitum resulting in a \emph{play} $\Play = v_0 v_1 v_2 \dotsm \in \Vertices^\omega$ that is an infinite sequence of vertices such that \((\Vertex[i], \Vertex[i+1]) \in \Edges\) for all \(i \ge 0\).
For \(i < j\), we denote by \(\PlayInfix{i}{j}\) the \emph{infix} \(\Vertex[i] \Vertex[i+1] \dotsm \Vertex[j]\) of \(\Play\). 
Its length is \( \abs{\PlayInfix{i}{j}} = j - i\), the number of edges. 
We denote by \(\PlayPrefix{j}\) the finite \emph{prefix} \(\Vertex[0] \Vertex[1] \dotsm \Vertex[j]\) of \(\Play\), and by \(\PlaySuffix{i}\) the infinite \emph{suffix} \(\Vertex[i] \Vertex[i+1] \ldots \) of \(\Play\). 
We denote by \(\PlaySet{\MDP}\) and \(\PrefixSet{\MDP}\) the set of all plays and the set of all finite prefixes in \(\MDP\) respectively.
We denote by \(\Last{\Prefix}\) the last vertex of the prefix \(\Prefix \in \PrefixSet{\MDP}\). 
We denote by \(\PrefixSet[i]{\MDP}\) (\(i \in \{\Main, \Adversary\}\)) the set of all prefixes \(\Prefix\) such that \(\Last{\Prefix} \in \Vertices[i]\). 
\Cref{fig:sas-not-same-as-sls} shows an example of an MDP; vertices of \(\PlayerMain\) are shown as circles and probabilistic vertices as diamonds. 
The payoff for each edge is shown in {\color{red} red} and probability distributions over out-neighbours of probabilistic vertices are shown in {\color{blue} blue}.

\subparagraph{Objectives.}
An objective \(\Objective\) is a Borel-measurable subset of \(\PlaySet{\MDP}\)~\cite{Mar98}.
A play \(\Play\in \PlaySet{\MDP}\) \emph{satisfies} an objective \(\Objective\) if \(\Play \in \Objective\). 
For a target set \(T \subseteq \Vertices\), we define some objectives:
\begin{itemize}
\item \emph{reachability}, denoted \(\ReachObj(T) = \{\Play = v_{0}v_{1}v_{2}\cdots \suchthat \exists i \ge 0,\, v_{i} \in T \}\), is the set of all plays that visit a vertex in \(T\);
  we also define some variants of \(\ReachObj(T)\):
  \begin{itemize}
  \item For a subset \(F \subseteq E\) of edges, \emph{edge-restricted reachability},  denoted \(\ReachObj_{F}(T) = \{\Play = v_{0} v_{1} v_{2} \cdots \suchthat \exists i \ge 0, \, \forall j \ge 0, v_{i} \in T \wedge (j < i \implies (v_{j}, v_{j+1}) \in F) \}\), is the set of all plays that visit a vertex in \(T\) and use only edges in \(F\) until the first visit to \(T\);    
  \item For an integer \(k \ge 0\), \emph{bounded reachability}, denoted \(\ReachObj_{\le k}(T) = \{\Play = v_{0} v_{1} v_{2} \cdots \suchthat \exists i \ge 0, \, i \le k, \, v_{i} \in T \}\), is the set of all plays that visit a vertex in \(T\) in at most \(k\) steps from the start of the play;
  \end{itemize}
\item \emph{safety}, denoted \(\SafeObj(T) = \{\Play = v_{0}v_{1}v_{2} \cdots \suchthat \forall i \ge 0,\, v_{i} \in T \}\), is the set of all plays that never visit a vertex not in \(T\);
\item \emph{\Buchi}, denoted \(\BuchiObj(T) = \{\Play = v_{0}v_{1}v_{2}\cdots \suchthat \forall i \ge 0,\, \exists j \ge i,\, v_{j} \in T \}\), is the set of all plays that visit \(T\) infinitely often;
  and
\item \emph{co\Buchi}, denoted \(\CoBuchiObj(T) = \{\Play = v_{0}v_{1}v_{2} \suchthat \exists i \ge 0, \, \forall j \ge i, \, v_{j} \in T \}\), the set of all plays that eventually visit only vertices in \(T\).
\end{itemize}
The objectives \(\ReachObj(T)\) and \(\SafeObj(V \setminus T)\) are duals of each other, that is, a play satisfies \(\ReachObj(T)\) if and only if it does not satisfy \(\SafeObj(V \setminus T)\).
Similarly, the objectives \(\BuchiObj(T)\) and \(\CoBuchiObj(V \setminus T)\) are duals of each other.

For a rational threshold $\Threshold$, we define some more objectives:
\begin{itemize}
\item \emph{total payoff}, denoted \(\TP(\Threshold) = \{ \Play = v_{0} v_{1} v_{2}\cdots \suchthat \limsup_{n \to \infty} \sum_{i=0}^{n-1} \PayoffFunction(v_i, v_{i+1}) \ge \Threshold\}\), is the set of all plays with (supremum) total payoff at least \(\Threshold\); and
\item \emph{mean payoff}, denoted \(\MP(\Threshold) = \{ \Play = v_{0} v_{1} v_{2}\cdots \suchthat \limsup_{n \to \infty} \frac{1}{n}\sum_{i=0}^{n-1} \PayoffFunction(v_i, v_{i+1}) \ge \Threshold\}\), is the set of all plays with (supremum) limit-average payoff at least \(\Threshold\).
\end{itemize}
In this work, we consider finitary versions of the mean-payoff objective called \emph{window mean-payoff} objectives, which are defined in \Cref{sec:window-mean-payoff-objectives}.

\subparagraph{Prefix independence.}
An objective \(\Objective\) is \emph{closed under suffixes} if for all plays \(\Play\) satisfying \(\Objective\), all suffixes of \(\Play\) also satisfy \(\Objective\), that is, \(\PlaySuffix{j} \in \Objective\) for all \(j \ge 0\).
An objective \(\Objective\) is \emph{closed under prefixes} if for all plays \(\Play\) satisfying \(\Objective\), for all prefixes \(\Prefix\) such that the concatenation \(\Prefix \cdot \Play\) is a play in \(\MDP\) (that is \(\Prefix \cdot \Play \in \PlaySet{\MDP}\)), we have that \(\Prefix \cdot \Play\) satisfies \(\Objective\).
An objective \(\Objective\) is \emph{prefix-independent} if it is closed under both suffixes and prefixes.
In other words, an objective \(\Objective\) is prefix-independent if it is independent of finite prefixes of the play, that is, it depends only on the infinite suffixes of a play.
For a prefix-independent objective \(\Objective\) and for all plays \(\Play\) and \(\Play'\) with a common suffix (that is, \(\Play'\) can be obtained from \(\Play\) by removing and adding finite prefixes), we have that \(\Play \in \Objective\) if and only if \(\Play' \in \Objective\).
The objectives \(\BuchiObj(T)\), \(\CoBuchiObj(T)\), and \(\MP(\Threshold)\) are prefix-independent (closed under both suffixes and prefixes).
In general, the objective \(\SafeObj(T)\) is closed under suffixes but not prefixes, \(\ReachObj(T)\) is closed under prefixes but not suffixes, and 
\(\ReachObj_{F}(T)\), \(\ReachObj_{\le k}(T)\), and \(\TP(\Threshold)\) are closed under neither suffixes nor prefixes.

\subparagraph{Strategies.}
A (deterministic) \emph{strategy}\footnote{Strategies can be randomized in general. However, randomization is not useful for the winning conditions considered in this paper, and thus we restrict our attention to deterministic strategies only.} for \(\PlayerMain\) in an MDP \(\MDP\) is a function \(\Strategy \colon \PrefixSet[\Main]{\MDP} \to \Vertices\) that maps prefixes \(\Prefix\) with \(\Last{\Prefix} \in \VerticesMain\) to an out-neighbour of \(\Last{\Prefix}\). 
Strategies can be realized as the output of a (possibly infinite-state) \emph{Mealy machine}, that is a deterministic transition system with transitions labelled by input/output pairs. 
Intuitively, in each step, if the token is at a vertex~\(v\), then~\(v\) is given as input to the Mealy machine.
Then, the Mealy machine emits an output (an out-neighbour of~\(v\) if \(v \in \VerticesMain\), an empty string if \(v \in \VerticesRandom\)) and the state of the Mealy machine is also updated.

The \emph{memory size} of a strategy \(\Strategy\) is the smallest number of states a Mealy machine defining \(\Strategy\) can have. 
A strategy \(\Strategy\) is \emph{memoryless} if for all prefixes \(\Prefix, \Prefix' \in \PrefixSet[\Main]{\MDP}\) such that \(\Last{\Prefix} = \Last{\Prefix'}\), we have \(\Strategy(\Prefix) = \Strategy(\Prefix')\). 
Memoryless strategies can be defined by Mealy machines with only one state. 
A play \(\Play = v_0 v_1 \dotsm\) is an \emph{outcome} of a strategy \(\Strategy\) if for all \(j \geq 0\) with \(v_{j} \in \VerticesMain\), we have \(v_{j+1} = \Strategy(\PlayPrefix{j})\). 
Fixing a strategy \(\Strategy\) and an initial vertex \(v\) in an MDP \(\MDP\) yields a distribution \(\Pr{\Strategy}{\MDP, \Vertex}{\cdot}\) over plays in \(\MDP\) (see e.g. \cite{BK08}).
For an objective \(\Objective\), we have that \(\Pr{\Strategy}{\MDP, \Vertex}{\Objective}\) is the probability that an outcome of \(\Strategy\) from \(v\) in \(\MDP\) satisfies \(\Objective\).

\subparagraph{Winning conditions.}
A strategy \(\Strategy\) of \(\PlayerMain\) from \(v\) in \(\MDP\) is \emph{winning for positive-\(\Objective\)} if \(\Pr{\Strategy}{\MDP, v}{\Objective} > 0\), \emph{winning for probability\((p)\)}-\(\Objective\) if \(\Pr{\Strategy}{\MDP, \Vertex}{\Objective} \ge p\), \emph{winning for almost-sure-\(\Objective\)} if  \(\Pr{\Strategy}{\MDP, \Vertex}{\Objective} = 1\), and \emph{winning for sure-\(\Objective\)} if all outcomes of \(\Strategy\) from \(v\) in \(\MDP\) satisfy \(\Objective\) (that is, the edge probabilities are ignored and the environment is considered adversarial instead of probabilistic). 
If such a winning strategy \(\Strategy\) from \(v\) exists, then the vertex \(v\) is said to be winning for positive-\(\Objective\), probability\((p)\)-\(\Objective\), almost-sure-\(\Objective\), or sure-\(\Objective\) respectively in \(\MDP\).
In addition, we also say that a vertex \(v\) is winning for \emph{limit-sure-\(\Objective\)} in \(\MDP\) if for all \(\epsilon > 0\), we have that \(v\) is winning for probability\((1 - \epsilon)\)-\(\Objective\).
Similar to winning vertices, an edge \(e\) is winning for positive-\(\Objective\) (resp., probability\((p)\)-\(\Objective\), limit-sure-\(\Objective\), almost-sure-\(\Objective\), sure-\(\Objective\)) if starting from \(e\), there exists a strategy \(\Strategy\) that is winning for positive-\(\Objective\) (resp., probability\((p)\)-\(\Objective\), limit-sure-\(\Objective\), almost-sure-\(\Objective\), sure-\(\Objective\)). 
The \emph{winning region} for positive-\(\Objective\) in \(\MDP\) is the set of all vertices in \(\MDP\) that are winning for positive-\(\Objective\).
The winning regions for probability\((p)\)-\(\Objective\), limit-sure-\(\Objective\), almost-sure-\(\Objective\), and sure-\(\Objective\) are defined similarly.
The winning region for sure-\(\ReachObj(T)\) of a set \(T\subseteq V\) is used often and thus has a special name, the \emph{sure-attractor} of \(T\), and is denoted by \(\SureAttr{}{T}\).

\subparagraph{Two-player games.}
When computing the winning region for sure-\(\Objective\) in an MDP \(\MDP= ((\Vertices, \Edges), (\VerticesMain, \VerticesRandom),  \ProbabilityFunction, \PayoffFunction)\), it is useful to view \(\MDP\) as a \emph{two-player game} played by \(\PlayerMain\) against an adversary \(\PlayerAdversary\) by forgetting the probability function $\ProbabilityFunction$ and treating vertices in $\Vertex \in \VerticesRandom$ as ones belonging \(\PlayerAdversary\) and the outgoing transitions of $\Vertex$ as the choices of \(\PlayerAdversary\).

\subparagraph*{SubMDPs.}
Given an MDP \(\MDP = \tuple{(\Vertices, \Edges), (\VerticesMain, \VerticesRandom),  \ProbabilityFunction, \PayoffFunction}\), a subset \(V' \subseteq V\) of vertices induces a \emph{subMDP} of \(\MDP\) if for all vertices \(v' \in V'\), player \(\PlayerMain\) has a strategy from \(v'\) to ensure that the token always remains in \(V'\) for all outcomes, that is, the set \(V'\) is a winning region for sure-\(\SafeObj(V')\) in \(\MDP\).
In graph-theoretic terms, the subset \(V'\) induces a subMDP if
\begin{enumerate*}[label=(\emph{\roman*})]
\item every vertex  \(\Vertex' \in \Vertices'\) has an out-neighbour in $\Vertices'$, that is    \(\OutNeighbours{\Vertex'} \cap \Vertices' \neq \emptyset\), and
\item every probabilistic vertex \(\Vertex' \in \VerticesRandom \intersection \Vertices'\) has all its out-neighbours in~$\Vertices'$, that is \(\OutNeighbours{\Vertex'} \subseteq \Vertices'\).
\end{enumerate*}
The \emph{induced subMDP} is denoted by  \(\subMDP{\MDP}{ \Vertices'}\) and is obtained by restricting \(\MDP\) to \(V'\).
Formally, \(\subMDP{\MDP}{V'} = ((\Vertices',  \Edges'),
(\VerticesMain \intersection \Vertices', \VerticesRandom \intersection \Vertices'), \ProbabilityFunction',
\PayoffFunction')\), where \(\Edges' = \Edges \intersection (\Vertices' \times \Vertices')\), and \(\ProbabilityFunction'\) and \(\PayoffFunction'\) are restrictions of \(\ProbabilityFunction\) and \(\PayoffFunction\) to \((\Vertices', \Edges')\).
\begin{proposition}\label{prop:winning-region-sure-objective-closed-under-suffixes-induces-submdp}
  If \(\ObjectiveSuf\) is an objective that is closed under suffixes, then the winning region for sure-\(\ObjectiveSuf\) in \(\MDP\) induces a subMDP of \(\MDP\).
\end{proposition}
\begin{proof}
  We show that the winning region for sure-\(\ObjectiveSuf\) in \(\MDP\) induces a subMDP in \(\MDP\).
  Let \(W_{\Sure}\) denote this winning region.
  We show that both Conditions (\emph{i}) and (\emph{ii}) hold for \(W_{\Sure}\).
\begin{itemize}
\item \emph{Condition (i):}
  Suppose towards a contradiction that Condition (\emph{i}) does not hold for \(W_{\Sure}\), that is, there exists a vertex \(v \in W_{\Sure}\) such that \(v\) does not have an out-neighbour in \(W_{\Sure}\).

  Since \(v\) is winning for sure-\(\ObjectiveSuf\), there exists a strategy \(\Strategy\) from \(v\) that is winning for sure-\(\ObjectiveSuf\).
  In particular, if \(v \in \VerticesMain\), then \(\PlayerMain\) moves the token from \(v\) to an out-neighbour \(u\) of \(v\) that is \(\Strategy(v)\), and if \(v \in \VerticesRandom\), then the token may move to any out-neighbour \(u\) of \(v\).
  In both cases, player \(\PlayerMain\) can continue following \(\Strategy\) from \(u\) and ensure that all outcomes of \(\Strategy\) starting with the prefix \(vu\) satisfy \(\ObjectiveSuf\), and thus, we have that \(vu\) is a winning prefix for sure-\(\ObjectiveSuf\).
  Since \(\ObjectiveSuf\) is closed under suffixes, we have that for every play \(\Play\) starting with \(vu\) that satisfies \(\ObjectiveSuf\), the play \(\PlaySuffix{1}\) obtained by dropping the initial vertex \(v\) from \(\Play\) also satisfies \(\ObjectiveSuf\).
  Thus, if the play begins from \(u\) and  \(\PlayerMain\) plays as if the play started from \(vu\), that is, if \(\PlayerMain\) follows the strategy \(\Strategy[u]\) given by \(\Strategy[u](u \cdot \Prefix) = \Strategy(v \cdot u \cdot \Prefix)\) for all prefixes \(\Prefix\), then we have that \(\Strategy[u]\) is a winning strategy from \(u\) for sure-\(\ObjectiveSuf\).
  We get that \(u\) is winning for sure-\(\ObjectiveSuf\) and thus belongs to \(W_{\Sure}\).
  This is a contradiction since we had assumed that none of the out-neighbours of \(v\) belong to \(W_{\Sure}\).
\item \emph{Condition (ii):}
  Suppose towards a contradiction that there exists a probabilistic vertex \(v \in W_{\Sure} \intersection \VerticesRandom\) such that \(v\) has an out-neighbour \(u\) that does not belong to \(W_{\Sure}\).
  Using similar arguments as the previous case, we have that \(vu\) is a winning prefix for sure-\(\ObjectiveSuf\), and since \(\ObjectiveSuf\) is closed under suffixes, we have that \(u\) is winning for sure-\(\ObjectiveSuf\).
  Thus, the vertex \(u\) belongs to \(W_{\Sure}\), which is a contradiction.\qedhere
\end{itemize}
\end{proof}

A similar result holds for the winning region for almost-sure-\(\ObjectiveSuf\) in \(\MDP\).
\begin{proposition}\label{prop:winning-region-almost-sure-objective-closed-under-suffixes-induces-submdp}
    If \(\ObjectiveSuf\) is an objective that is closed under suffixes, then the winning region for almost-sure-\(\ObjectiveSuf\) in \(\MDP\) induces a subMDP of \(\MDP\).
\end{proposition}

\subparagraph*{SubMDP closure.}
Consider a subset \(V' \subseteq V\) that satisfies Condition (\emph{i}) but not necessarily Condition (\emph{ii}), that is, every vertex  \(v' \in V'\) has an out-neighbour in \(V'\) but there may exist probabilistic vertices in \(V'\) with some out-neighbours not in \(V'\).
Such a subset \(V'\) induces a \emph{subMDP closure} of \(\MDP\), denoted by \(\subMDPClosure{\MDP}{V'}\).
The subMDP closure is itself an MDP and is obtained by first restricting \(\MDP\) to \(V'\), and then for each probabilistic vertex \(v' \in V'\), we scale up probabilities over the out-neighbours of \(v'\) to obtain a probability distribution over the out-neighbours of \(v'\) that belong to \(V'\).
Formally, 
\(\subMDPClosure{\MDP}{V'} = ((\Vertices',  \Edges'),
(\VerticesMain \intersection \Vertices', \VerticesRandom \intersection \Vertices'), \ProbabilityFunction', \PayoffFunction')\), where \(\Edges' = \Edges \intersection (\Vertices' \times \Vertices')\), and \(\ProbabilityFunction'\) is given by \(\ProbabilityFunction'(v')(u') = \ProbabilityFunction(v')(u')/\sum_{x' \in \OutNeighbours{v'} \intersection V'} \ProbabilityFunction(v')(x')\) for all \(v' \in V' \intersection \VerticesRandom\), \(u' \in \OutNeighbours{v'}  \intersection V'\), and \(\PayoffFunction'\) is a restriction of \(\PayoffFunction\) to \((\Vertices', \Edges')\).
It can be seen that the complement of a sure-attractor induces a subMDP closure.

\subparagraph*{Maximal end components.}
An \emph{end component} (EC) in an MDP \(\MDP\) is a subset \(T \subseteq V\) of vertices such that \(T\) is strongly connected and \(T\) induces a subMDP \(\subMDP{\MDP}{T}\). 
Thus, for every pair of vertices \(v, v'\) in the \(T\), player \(\PlayerMain\) has a strategy to reach \(v'\) from \(v\) with probability~\(1\) and \(\PlayerMain\) can ensure with probability~\(1\) that the token never leaves \(T\). 
A \emph{maximal end component} (MEC) is an EC that is not contained in any other EC. 
The MECs in an MDP are disjoint, and thus, each vertex in an MDP belongs to either no MEC or exactly one MEC.
The MEC decomposition can computed in polynomial time~\cite{CH14}.

\section{Window mean-payoff objectives}%
\label{sec:window-mean-payoff-objectives}
In this section, we recall window mean-payoff objectives introduced in~\cite{CDRR15} in the context of two-player games.
We study two types of window mean-payoff objectives: 
\emph{fixed}, in which a window length \(\WindowLength \geq 1\) is given,  and 
\emph{bounded}, in which for every play we need a bound on the window length. 
First, we recall the notion of windows in a play.

\subparagraph{Windows.}
We extend the definition of \emph{\(\Threshold\)-windows} as given in~\cite{CDRR15} for arbitrary thresholds \(\Threshold \in \Rationals\).
Given a play \(\Play = \Vertex[0] \Vertex[1] \cdots\), for all \(i \ge 0\), a \(\Threshold\)-window \emph{starts} at position \(i\) of the play.
The \(\Threshold\)-window starting at position \(i\) \emph{closes} at a position \(j\) if \(j\) is the first position after \(i\) such that the mean payoff of the infix \(\Play(i, j)\) is at least \(\Threshold\). 
Equivalently, if for all \(i < k \le j\) it is the case that the mean payoff of \(\Play(i,k)\) is less than \(\Threshold\), then the \(\Threshold\)-window starting at \(i\) is \emph{open} at \(j\).
Given \(j > 0\), we say a \(\Threshold\)-window is open at \(j\) if there exists an open \(\Threshold\)-window \(\Play(i, j)\) starting at \(i\) for some \(i < j\). 
When there is no confusion, we also informally say that the \(\Threshold\)-window opens at vertex \(v_{i}\) and closes at vertex \(v_{j}\) to mean that the window opens at position \(i\) and closes at position \(j\).
An interesting and useful observation is the \emph{inductive property of \(\Threshold\)-windows}: if the \(\Threshold\)-window starting at position \(i\) closes at position \(j\), then for all \(i < k < j\), the \(\Threshold\)-window starting at \(k\) is closed at or before \(j\).

\subparagraph{Fixed window length.}
Let \(\WindowLength \geq 1\) be a given window length and \(\Threshold \in \Rationals\) be a given threshold. 
A play \(\Play = \Vertex[0] \Vertex[1] \dotsm \)  in \(\MDP\) satisfies the \emph{good window} objective \(\GW_{\MDP}(\WindowLength, \Threshold)\) if the \(\Threshold\)-window that opens at \(v_0\) in \(\Play\) closes in at most \(\WindowLength\) steps.
\[\GW_{\MDP}(\WindowLength, \Threshold) = \{\Play \in \PlaySet{\MDP} \mid \exists j \in \PositiveSet{\WindowLength}, \ \frac{1}{j} \sum_{k=0}^{j-1} \PayoffFunction(\Vertex[k], \Vertex[k+1]) \geq \Threshold\} \]
A play \(\Play \)  in \(\MDP\) satisfies the \emph{direct fixed window mean-payoff} objective \(\DirFWMP_{\MDP}(\WindowLength, \Threshold)\) if every \(\Threshold\)-window in \(\Play\) closes in at most \(\WindowLength\) steps.
\begin{align*}
  \DirFWMP_{\MDP}(\WindowLength, \Threshold) = \{ \Play \in \PlaySet{\MDP} \suchthat  \forall i \ge 0, \ \PlaySuffix{i} \in \GW_{\MDP}(\WindowLength, \Threshold)\}
\end{align*} 
The prefix-independent version of this objective is the \emph{fixed window mean-payoff} objective \(\FWMP_{\MDP}(\WindowLength, \Threshold)\).
A play \(\Play\) satisfies \(\FWMP_{\MDP}(\WindowLength, \Threshold)\) 
if there exists a suffix of \(\Play\) that satisfies \(\DirFWMP_{\MDP}(\WindowLength, \Threshold)\).
  In other words, the play \(\Play\) satisfies \(\FWMP_{\MDP}(\WindowLength, \Threshold)\) 
if from some point on in the play, all \(\Threshold\)-windows close in at most \(\WindowLength\) steps.
\begin{align*}
  \FWMP_{\MDP}(\WindowLength, \Threshold) 
  &= \{ \Play \in \PlaySet{\MDP} \suchthat  \exists k \geq 0, \  \PlaySuffix{k} \in \DirFWMP_{\MDP}(\WindowLength, \Threshold) \}
\end{align*} 
We omit the subscript \(\MDP\) when it is clear from the context. 
Note that \(\FWMP(\WindowLength, \Threshold) \subseteq  \FWMP(\WindowLength', \Threshold) \) for \(\WindowLength \le \WindowLength'\) as a smaller window length is a stronger constraint.
We also have that \(\FWMP(\WindowLength, \Threshold) \supseteq \FWMP(\WindowLength, \Threshold')\) if \(\Threshold \le \Threshold'\), since closing of \(\Threshold'\)-windows implies closing of \(\Threshold\)-windows.

\subparagraph{Bounded window length.}
We also consider bounded versions of \(\DirFWMP(\WindowLength, \Threshold)\) and \(\FWMP(\WindowLength, \Threshold)\).
A play \(\Play\) satisfies the \emph{direct bounded window mean-payoff} objective if there exists a window length \(\WindowLength \ge 1\) such that \(\Play\) satisfies \(\DirFWMP_{\MDP}(\WindowLength, \Threshold)\).
\begin{align*}
  \DirBWMP_{\MDP}(\Threshold) = \{ \Play \in \PlaySet{\MDP} \suchthat \exists \WindowLength \ge 1, \ \Play \in \DirFWMP_{\MDP}(\WindowLength, \Threshold)\}
\end{align*}
Similarly, a play \(\Play\) satisfies the bounded window mean-payoff objective \(\BWMP(\Threshold)\) if there exists a window length \(\WindowLength \ge 1\) such that \(\Play\) satisfies \(\FWMP_{\MDP}(\WindowLength, \Threshold)\).
The \(\BWMP(\Threshold)\) objective is also the prefix-independent version of the \(\DirBWMP(\Threshold)\) objective.
\begin{align*}
  \BWMP_{\MDP}(\Threshold)
  &= \{ \Play \in \PlaySet{\MDP} \suchthat \exists \WindowLength \ge 1, \ \Play \in \FWMP_{\MDP}(\WindowLength, \Threshold)\}\\
  &= \{ \Play \in \PlaySet{\MDP} \suchthat  \exists k \geq 0, \  \PlaySuffix{k} \in \DirBWMP_{\MDP}(\Threshold) \}
\end{align*}
Note that \(\DirFWMP(\WindowLength, \Threshold)\) and \(\DirBWMP(\Threshold)\) are closed under suffixes, and \(\FWMP(\WindowLength, \Threshold)\) and \(\BWMP(\Threshold)\) are prefix-independent objectives.
We also note that all of these objectives strengthen the standard mean-payoff objective \(\MP(\Threshold)\)~\cite{EM79}.
In particular, we have the following relation between the objectives
\[\DirFWMP(\WindowLength, \Threshold) \subseteq \FWMP(\WindowLength, \Threshold) \subseteq \BWMP(\Threshold) \subseteq \MP(\Threshold),\]
\[\DirFWMP(\WindowLength, \Threshold) \subseteq \DirBWMP(\Threshold) \subseteq \BWMP(\Threshold) \subseteq \MP(\Threshold),\]
and all of these inclusions are strict in general. 

\subparagraph{Decision and synthesis problems.}
We consider the following problems that are a combination of sure and almost-sure (and limit-sure) satisfaction of \(\FWMPL\) objectives.
\begin{enumerate}
\item \emph{sure-almost-sure} satisfaction: 
  Given an MDP \(\MDP\), a vertex \(v\), a sure threshold \(\GuaranteeThreshold\) and an almost-sure threshold \(\AlmostSureThreshold\), does there exist a strategy \(\Strategy\) of \(\PlayerMain\) such that 
  \begin{enumerate*}[label=(\emph{\roman*})]
  \item all outcomes of \(\Strategy\) from \(v\) satisfy \(\FWMP(\WindowLength, \GuaranteeThreshold)\), and 
  \item with probability~\(1\), the outcome of \(\Strategy\) from \(v\) satisfies \(\FWMP(\WindowLength, \AlmostSureThreshold)\) from \(v\) (i.e., \(\Pr{\Strategy}{\MDP, \Vertex}{\FWMP(\WindowLength, \AlmostSureThreshold)} = 1\))?
    If the answer is yes, then we would also like to synthesize such a strategy.
  \end{enumerate*}
\item \emph{sure-limit-sure} satisfaction:
  Given an MDP \(\MDP\), a vertex \(v\), a sure threshold \(\GuaranteeThreshold\) and a limit-sure threshold \(\AlmostSureThreshold\),
  for all \(0 < \epsilon < 1\), does there exist a strategy \(\Strategy[\epsilon]\) of \(\PlayerMain\) such that
  \begin{enumerate*}[label=(\emph{\roman*})]
  \item all outcomes of \(\Strategy[\epsilon]\) from \(v\) satisfy \(\FWMP(\WindowLength, \GuaranteeThreshold)\), and 
  \item with probability at least~\(1 - \epsilon\), the outcome of \(\Strategy[\epsilon]\) from \(v\) satisfies \(\FWMP(\WindowLength, \AlmostSureThreshold)\) from \(v\) (i.e., \(\Pr{\Strategy[\epsilon]}{\MDP, \Vertex}{\FWMP(\WindowLength, \AlmostSureThreshold)} \ge 1 - \epsilon\))?
    If the answer is yes, then we would also like to synthesize such a family of strategies.
  \end{enumerate*}
\end{enumerate}
We assume without loss of generality that \(\GuaranteeThreshold < \AlmostSureThreshold\), because otherwise the problems reduce to checking satisfaction of sure-\(\FWMP(\WindowLength, \GuaranteeThreshold)\).
We also study analogous problems for \(\BWMP\).
\begin{remark}[SAS and SLS satisfaction of direct objectives]
  We only consider sure-almost-sure and sure-limit-sure satisfaction of the \emph{prefix-independent} versions of the window mean-payoff objectives.
  It is shown in \cite[Remark 4.10]{BDOR20} that almost-sure satisfaction of \(\DirFWMP(\WindowLength, \GuaranteeThreshold)\) (and thus also limit-sure satisfaction of \(\DirFWMP(\WindowLength, \GuaranteeThreshold)\)) is equivalent to the sure satisfaction of  \(\DirFWMP(\WindowLength, \GuaranteeThreshold)\). 
  Thus, sure-almost-sure satisfaction (and sure-limit-sure satisfaction) of \(\DirFWMP(\WindowLength)\) reduces to just the sure satisfaction of \(\DirFWMP(\WindowLength, \AlmostSureThreshold)\).
  It is also shown in \cite[Example 7.1]{BDOR20} that almost-sure satisfaction  of \(\DirBWMP\) is not well-behaved. 
  Thus, sure-almost-sure and sure-limit-sure satisfaction of \(\DirBWMP\) are also not well-behaved and we do not study these problems in this paper.
\end{remark}

\begin{example}\label{exa:sas-sls-difference}
  While almost-sure satisfaction and limit-sure satisfaction of \(\FWMP(\WindowLength, \GuaranteeThreshold)\) are the same in MDPs~\cite{CDH04}, we illustrate using the MDP in \Cref{fig:sas-not-same-as-sls} with values \(\WindowLength = 2\), \(\GuaranteeThreshold = 1\), and \(\AlmostSureThreshold = 5\) that sure-almost-sure satisfaction is strictly stronger than sure-limit-sure satisfaction for \(\FWMP(\WindowLength)\).
  \begin{figure}[t]
    \centering
    \begin{tikzpicture}[node distance=1.5cm]
      \node[state, initial above] (v1) {\(v_{1}\)};
      \node[random, draw, right of=v1] (v2) {\(v_{2}\)};
      \node[state, right of=v2] (v3) {\(v_{3}\)};

      \draw 
      (v1) edge[loop left] node[auto]{\small \(\EdgeValues{1}{}\)} (v1)
      ;
      \draw 
      (v1) edge[bend left] node[auto]{\small \(\EdgeValues{-1}{}\)} (v2)
      (v2) edge[bend left] node[auto]{\small \(\EdgeValues{-1}{.5}\)} (v1)
      (v2) edge[] node[auto]{\small \(\EdgeValues{1}{.5}\)} (v3)
      ;
      \draw 
      (v3) edge[loop right] node[auto]{\small \(\EdgeValues{5}{}\)} (v3)
      ;
    \end{tikzpicture}

    \caption{%
      All vertices are winning for sure-\(\FWMP(2, 1)\), almost-sure-\(\FWMP(2, 5)\), and sure-\(\FWMP(2, 1)\)-limit-sure-\(\FWMP(2, 5)\).
      However, vertices \(v_{1}\) and \(v_{2}\) are not winning for sure-\(\FWMP(2, 1)\)-almost-sure-\(\FWMP(2, 5)\).
    } 
    \label{fig:sas-not-same-as-sls}
  \end{figure}
  We have an MDP where the vertex \(v_{1}\) is winning for sure-\(\FWMP(2, 1)\)-limit-sure-\(\FWMP(2, 5)\) but not for sure-\(\FWMP(2, 1)\)-almost-sure-\(\FWMP(2, 5)\).

  To see that sure-\(\FWMP(2, 1)\) and almost-sure-\(\FWMP(2, 5)\) cannot be satisfied simultaneously from \(v_{1}\): in order to satisfy  almost-sure-\(\FWMP(2, 5)\), the token must reach \(v_{3}\) with probability \(1\), and that requires that as long as the token has not reached \(v_{3}\), the token must be forwarded to \(v_{2}\) from \(v_{1}\) repeatedly.
  However, consider the outcome of this strategy where the token always moves from \(v_{2}\) to \(v_{1}\).
  In this outcome, the sequence \(v_{1} v_{2} v_{1}\) appears infinitely often, and each such occurrence \(v_{1} v_{2} v_{1}\) is an open \(1\)-window of length \(2\). Thus, this outcome does not satisfy \(\FWMP(2, 1)\), and we have that sure-\(\FWMP(2, 1)\) is violated.

  On the other hand, the vertex \(v_{1}\) is winning for sure-limit-sure in the following way: given \(0 < \epsilon < 1\), visit \(v_{2}\) a large integer \(N\) number of times such that the probability of reaching \(v_{3}\) is greater than \(1 - \epsilon\). 
  If the vertex \(v_{2}\) has been visited \(N\) times without reaching \(v_{3}\), then loop on \(v_{1}\) for the rest of the play.
  This strategy simultaneously satisfies both sure-\(\FWMP(2, 1)\) and limit-sure-\(\FWMP(2, 5)\).
  \lipicsEnd
\end{example}

\begin{example}
  The winning region for sure-\(\FWMP(\WindowLength, \GuaranteeThreshold)\)-almost-sure-\(\FWMP(\WindowLength, \AlmostSureThreshold)\) is not merely the intersection or the composition of the winning regions for sure-\(\FWMP(\WindowLength, \GuaranteeThreshold)\) and for almost-sure-\(\FWMP(\WindowLength, \AlmostSureThreshold)\): in the MDP in \Cref{fig:sas-not-same-as-sls}, we have that every vertex is winning for sure-\(\FWMP(2, 1)\) and for almost-sure-\(\FWMP(2, 5)\) separately, and yet there exist vertices \(v_{1}\) and \(v_{2}\) that are not winning for sure-\(\FWMP(2, 1)\)-almost-sure-\(\FWMP(2, 5)\).
  Indeed, we need to make good use of prefix-independence and the finitary nature of window mean-payoff objectives to find the winning region for the combination of winning conditions.
  \lipicsEnd
\end{example}

\begin{example}
  We show using the MDP in \Cref{fig:sls-not-intersection-of-s-ls} that the winning region for sure-\(\FWMP(\WindowLength, \GuaranteeThreshold)\)-limit-sure-\(\FWMP(\WindowLength, \LimitSureThreshold)\) is in general not equal to the intersection of winning region for sure-\(\FWMP(\WindowLength, \GuaranteeThreshold)\) and the winning region for limit-sure-\(\FWMP(\WindowLength, \LimitSureThreshold)\).
  \begin{figure}[t]
    \centering
    \begin{tikzpicture}[node distance=1.5cm]
      \node[state, initial above] (v1) {\(v_{1}\)};
      \node[random, draw, right of=v1] (v2) {\(v_{2}\)};
      \node[state, right of=v2] (v3) {\(v_{3}\)};

      \draw 
      (v1) edge[loop left] node[auto]{\small \(\EdgeValues{1}{}\)} (v1)
      ;
      \draw 
      (v1) edge node[auto]{\small \(\EdgeValues{0}{}\)} (v2)
      (v2) edge[loop above] node[auto]{\small \(\EdgeValues{-1}{.5}\)} (v2)
      (v2) edge[] node[auto]{\small \(\EdgeValues{0}{.5}\)} (v3)
      ;
      \draw 
      (v3) edge[loop right] node[auto]{\small \(\EdgeValues{5}{}\)} (v3)
      ;
    \end{tikzpicture}

    \caption{%
      The winning region for sure-\(\FWMP(2, 1)\) is \(\{v_{1}, v_{3}\}\) and the winning region for limit-sure-\(\FWMP(2, 5)\) is \(\{v_{1}, v_{2}, v_{3}\}\).
      However, the winning region for sure-\(\FWMP(2, 1)\)-limit-sure-\(\FWMP(2, 5)\) is \(\{v_{3}\}\), a strict subset of the intersection \(\{v_{1}, v_{3}\}\) of the winning regions for sure-\(\FWMP(2, 1)\) and for limit-sure-\(\FWMP(2, 5)\).
    } 
    \label{fig:sls-not-intersection-of-s-ls}
  \end{figure}

  Observe that the winning region for sure-\(\FWMP(2, 1)\) is equal to \(\{v_{1}, v_{3}\}\), since the strategy that always takes the self-loops on \(v_{1}\) and \(v_{3}\) is winning for sure-\(\FWMP(2, 1)\) from \(v_{1}\) and \(v_{3}\), and there exists an outcome starting from \(v_{2}\) that loops on \(v_{2}\) forever giving a play that is losing for \(\FWMP(2, 1)\).
  On the other hand, all vertices in this MDP are winning for almost-sure-\(\FWMP(2, 5)\) (and thus also for limit-sure-\(\FWMP(2, 5)\)) as witnessed by the strategy that moves the token from \(v_{1}\) to \(v_{2}\) and loops on \(v_{3}\).
  However, only the vertex \(v_{3}\) is winning for sure-\(\FWMP(2, 1)\)-limit-sure-\(\FWMP(2, 5)\).
  There does not exist a strategy that simultaneously satisfies both sure-\(\FWMP(2, 1)\) and limit-sure-\(\FWMP(2, 5)\) from \(v_{1}\) as the former cannot be satisfied by moving to \(v_{2}\) and the latter cannot be satisfied by staying at~\(v_{1}\).

  However, we show in \Cref{sec:sls} that the winning region for sure-\(\FWMP(\WindowLength, \GuaranteeThreshold)\)-limit-sure-\(\FWMP(\WindowLength, \LimitSureThreshold)\) is the composition of the winning region for limit-sure-\(\FWMP(\WindowLength, \LimitSureThreshold)\) (equivalently almost-sure-\(\FWMP(\WindowLength, \LimitSureThreshold)\)) and for sure-\(\FWMP(\WindowLength, \GuaranteeThreshold)\).
  That is, the winning region for sure-\(\FWMP(\WindowLength, \GuaranteeThreshold)\)-limit-sure-\(\FWMP(\WindowLength, \LimitSureThreshold)\) is equal to the winning region for limit-sure-\(\FWMP(\WindowLength, \LimitSureThreshold)\) (equivalently almost-sure-\(\FWMP(\WindowLength, \LimitSureThreshold)\)) in the subMDP induced by the winning region for sure-\(\FWMP(\WindowLength, \GuaranteeThreshold)\) in the original MDP.
  \lipicsEnd
\end{example}

\section{Sure-almost-sure satisfaction}
\label{sec:bwc-satisfaction-boolean}

\subsection{Sure-almost-sure satisfaction of FWMP}
In this section, we present \Cref{alg:sure-fwmp-almostsure-fwmp} (\(\SASFWMPAlgo\)) that returns the set of all vertices \(v\) from which 
\(\PlayerMain\) can satisfy sure-\(\FWMP(\WindowLength, \GuaranteeThreshold)\)-almost-sure-\(\FWMP(\WindowLength, \AlmostSureThreshold)\).
We note that for sure-almost-sure satisfaction of window mean-payoff objectives, the exact probabilities out of probabilistic vertices do not matter, and thus we may omit them in figures.

\paragraph*{\Cref{alg:sure-fwmp-almostsure-fwmp}: Sure-almost-sure \(\FWMP\)}

\begin{algorithm}[t]
  \caption{\(\SASFWMPAlgo(\MDP, \WindowLength, \GuaranteeThreshold, \AlmostSureThreshold)\)}%
  \label{alg:sure-fwmp-almostsure-fwmp}
  \begin{algorithmic}[1]
    \Require An MDP \(\MDP = ((\Vertices, \Edges), (\VerticesMain, \VerticesRandom), \ProbabilityFunction, \PayoffFunction)\), window length \(\WindowLength\), sure threshold \(\GuaranteeThreshold\), and almost-sure threshold \(\AlmostSureThreshold\)
    \Ensure The winning region for sure-\(\FWMP(\WindowLength, \GuaranteeThreshold)\)-almost-sure-\(\FWMP(\WindowLength, \AlmostSureThreshold)\) in \(\MDP\)
    \State \(W_{\Sure\GuaranteeThreshold} \assign \SureFWMPAlgo(\MDP, \WindowLength, \GuaranteeThreshold)\)\label{alg-line:sasf-sure-alpha-win}
    \State \(W_{\Sure\AlmostSureThreshold} \assign \SureFWMPAlgo(\subMDP{\MDP}{W_{\Sure\GuaranteeThreshold}}, \WindowLength, \AlmostSureThreshold)\)\label{alg-line:sasf-sure-beta-win}

    \State \(W \assign W_{\Sure\AlmostSureThreshold}\)\label{alg-line:sasf-initialize-w-to-wbeta}
    \Repeat \label{alg-line:repeat}
    \State \(P \assign \PosCPreAlgo(\subMDP{\MDP}{W_{\Sure\GuaranteeThreshold}}, W)\)\label{alg-line:sasf-poscpre}
    \State \(W_{\sdab} \assign \SureDirFWMPASBuchiAlgo(\subMDPClosure{\MDP}{W_{\Sure\GuaranteeThreshold} \setminus W}, \WindowLength, \GuaranteeThreshold, P)\)\label{alg-line:sasf-sadb}
    \State \(W \assign W \union W_{\sdab}\)\label{alg-line:sasf-include-wsdab}
    \State \(W \assign \SureAttrAlgo(\subMDP{\MDP}{ W_{\Sure\GuaranteeThreshold}}, W)\)\label{alg-line:sasf-sure-attr}

    \Until{\(W_{\sdab} = \emptyset\)}\label{alg-line:sasf-repeat-block-end}
    \State \Return \(W\)\label{alg-line:sasf-return-w}
  \end{algorithmic}
\end{algorithm}

\subparagraph{Description of \Cref{alg:sure-fwmp-almostsure-fwmp}.}
This algorithm computes and returns a set \(W\) (Line~\ref{alg-line:sasf-return-w}) of vertices that are winning for sure-\(\FWMP(\WindowLength, \GuaranteeThreshold)\)-almost-sure-\(\FWMP(\WindowLength, \AlmostSureThreshold)\) in \(\MDP\).

We begin in Line~\ref{alg-line:sasf-sure-alpha-win} by computing the set \(W_{\Sure\GuaranteeThreshold}\) that is the winning region for \(\PlayerMain\) for sure-\(\FWMP(\WindowLength, \GuaranteeThreshold)\) in \(\MDP\)
by viewing the probabilistic vertices of \(\MDP\) as adversarial, and then using \cite[Algorithm~1]{CDRR15} that finds the winning region for \(\PlayerMain\) for \(\FWMP(\WindowLength, \GuaranteeThreshold)\) in two-player games.
If \(\PlayerMain\) cannot satisfy sure-\(\FWMP(\WindowLength, \GuaranteeThreshold)\) from a vertex \(v\) in \(\MDP\), then
\(\PlayerMain\) also cannot satisfy sure-\(\FWMP(\WindowLength, \GuaranteeThreshold)\)-almost-sure-\(\FWMP(\WindowLength, \AlmostSureThreshold)\) from \(v\) in \(\MDP\), and we do not include in \(W\) any vertex belonging to \(\Vertices \setminus W_{\Sure\GuaranteeThreshold}\).
Moreover, since \(\FWMP(\WindowLength, \GuaranteeThreshold)\) is prefix-independent (and thus also closed under suffixes), it follows from \Cref{prop:winning-region-sure-objective-closed-under-suffixes-induces-submdp} that \(W_{\Sure\GuaranteeThreshold}\) induces a subMDP \(\subMDP{\MDP}{W_{\Sure\GuaranteeThreshold}}\) in which \(\PlayerMain\) has a winning strategy for sure-\(\FWMP(\WindowLength, \GuaranteeThreshold)\) from every vertex, and we subsequently only work with this subMDP.
In Line~\ref{alg-line:sasf-sure-beta-win}, we compute the set \(W_{\Sure\AlmostSureThreshold}\) that is the winning region for sure-\(\FWMP(\WindowLength, \AlmostSureThreshold)\) in \(\subMDP{\MDP}{ W_{\Sure\GuaranteeThreshold}}\), and in Line~\ref{alg-line:sasf-initialize-w-to-wbeta}, we initialize \(W\) to \(W_{\Sure\AlmostSureThreshold}\).
We then enter a \textbf{repeat} block (Lines~\ref{alg-line:sasf-poscpre} to~\ref{alg-line:sasf-sure-attr}) in which some vertices may be added to \(W\) in the following way.
\begin{itemize}
\item In Line~\ref{alg-line:sasf-poscpre}, we find the set \(P\) of probabilistic vertices in \(\subMDP{\MDP}{ W_{\Sure\GuaranteeThreshold}}\) from which the token reaches \(W\) in one step with positive probability but not surely.
  In terms of the underlying graph, \(P\) is the set of all vertices in \((W_{\Sure\GuaranteeThreshold} \intersection \VerticesRandom) \setminus W\) with at least one out-neighbour in \(W\) and at least one out-neighbour not in \(W\).
  We denote this by \(\PosCPreAlgo\) in the algorithm.
\item
  We note that every vertex \(v\) in \(W_{\Sure\GuaranteeThreshold} \setminus W\) has at least one out-neighbour in \(W_{\Sure\GuaranteeThreshold} \setminus W\) because otherwise \(\PlayerMain\) could ensure that the token reaches \(W\) surely from \(v\), in which case \(v\) should be included in \(W\).
  Thus, the set \(W_{\Sure\GuaranteeThreshold} \setminus W\) induces a subMDP closure of \(\subMDP{\MDP}{W_{\Sure\GuaranteeThreshold}}\), and we recall that we denote this by \(\subMDPClosure{\MDP}{W_{\Sure\GuaranteeThreshold} \setminus W}\).
  Suppose there exists a vertex \(v\) in \(\subMDPClosure{\MDP}{W_{\Sure\GuaranteeThreshold} \setminus W}\) from which \(\PlayerMain\) can simultaneously ensure with probability~\(1\) that \(P\) is visited infinitely often (that is, almost-sure-\(\BuchiObj(P)\)) and surely that all \(\GuaranteeThreshold\)-windows in the outcome always close in at most \(\WindowLength\) steps (that is, sure-\(\DirFWMP(\WindowLength, \GuaranteeThreshold)\)).
  Then starting from \(v\) in \(\subMDP{\MDP}{ W_{\Sure\GuaranteeThreshold}}\), we have that the token almost-surely eventually reaches \(W\) (from where sure-\(\FWMP(\WindowLength, \GuaranteeThreshold)\)-almost-sure-\(\FWMP(\WindowLength, \AlmostSureThreshold)\) is satisfied), and even if the token does not reach \(W\) (which happens with probability~\(0\)), the outcome still surely satisfies \(\DirFWMP(\WindowLength, \GuaranteeThreshold)\) (and thus also \(\FWMP(\WindowLength, \GuaranteeThreshold)\)).
  Thus, in Line~\ref{alg-line:sasf-sadb}, we call \Cref{alg:sure-dirfwmp-almostsure-buchi} (\SureDirFWMPASBuchiAlgo)  to compute the set \(W_{\sdab}\) of vertices that is the winning region for sure-\(\DirFWMP(\WindowLength, \GuaranteeThreshold)\)-almost-sure-\(\BuchiObj(P)\) in  \(\subMDPClosure{\MDP}{W_{\Sure\GuaranteeThreshold} \setminus W}\), and in Line~\ref{alg-line:sasf-include-wsdab}, we add to \(W\) all vertices in \(W_{\sdab}\). 
(The subscript \(\sdab\) is short for sure-\(\DirFWMP(\WindowLength, \GuaranteeThreshold)\)-almost-sure-\(\BuchiObj(P)\).)
\item  Finally, we compute \(\SureAttr{}{W}\), the sure-attractor of \(W\), that is the set of all vertices in \(W_{\Sure\GuaranteeThreshold}\) from which \(\PlayerMain\) can surely reach \(W\).
  We note that if \(\PlayerMain\) can satisfy sure-\(\FWMP(\WindowLength, \GuaranteeThreshold)\)-almost-sure-\(\FWMP(\WindowLength, \AlmostSureThreshold)\) from all vertices in \(W\), then \(\PlayerMain\) can also satisfy the same from all vertices in \(\SureAttr{}{W}\).
  This holds because \(\FWMP(\WindowLength, \GuaranteeThreshold)\) and \(\FWMP(\WindowLength, \AlmostSureThreshold)\) are prefix-independent objectives; starting from a vertex \(v\) in \(\SureAttr{}{W} \setminus W\), player \(\PlayerMain\) can ensure that every outcome from \(v\) reaches a vertex in \(W\), following which \(\PlayerMain\) can forget the entire history before reaching \(W\) (due to prefix-independence) and switch to a strategy that is winning for sure-\(\FWMP(\WindowLength, \GuaranteeThreshold)\)-almost-sure-\(\FWMP(\WindowLength, \AlmostSureThreshold)\) from this vertex in \(W\).
  Thus, we update \(W\) to include all vertices in \(\SureAttr{}{W}\) (Line~\ref{alg-line:sasf-sure-attr}).
\end{itemize}
We repeat Lines~\ref{alg-line:sasf-poscpre} to~\ref{alg-line:sasf-sure-attr} until we reach a fixed point for \(W\), that is when the set \(W_{\sdab}\) is empty (Line~\ref{alg-line:sasf-repeat-block-end}) and thus no vertices are added to \(W\).
At this point, we exit the \textbf{repeat} block and return \(W\) (Line~\ref{alg-line:sasf-return-w}).

We define some notation for the sets computed by \Cref{alg:sure-fwmp-almostsure-fwmp}.
Suppose the \textbf{repeat} block runs \(k\) times for some integer \(k \ge 1\).
For \(1 \le i \le k\), let \(P^{i}\), \(W_{\sdab}^{i}\), and \(W^{i}\) denote the sets \(P\), \(W_{\sdab}\), and \(W\) computed in the \(i^{\text{th}}\) iteration of the \textbf{repeat} block.
We also denote the set \(W_{\Sure\AlmostSureThreshold}\) by \(W^{0}\) for convenience.
We note some relations between the various sets.
\begin{itemize}
\item For \(1 \le i \le k\), the sets \(W^{i-1}\) and \(W_{\sdab}^{i}\) are disjoint, since \(W_{\sdab}^{i}\) is the winning region for sure-\(\DirFWMP(\WindowLength, \GuaranteeThreshold)\)-almost-sure-\(\BuchiObj(P^{i})\) in  \(\subMDPClosure{\MDP}{W_{\Sure\GuaranteeThreshold} \setminus W^{i-1}}\).
\item For \(1 \le i \le k\), if \(W_{\sdab}^{i}\) is non-empty, then \(P^{i}\) and \(W_{\sdab}^{i}\) have a non-empty intersection since a vertex in \(P^{i}\) that is visited infinitely often to satisfy sure-\(\DirFWMP(\WindowLength, \GuaranteeThreshold)\)-almost-sure-\(\BuchiObj(P^{i})\) from a vertex in \(W_{\sdab}^{i}\) also belongs to \(W_{\sdab}^{i}\). 
\item For \(1 \le i \le k\), we have that \(W^{i-1} \union W_{\sdab}^{i} \subseteq W^{i}\) since \(W^{i}\) is obtained by computing the sure-attractor of \(W^{i-1} \union W_{\sdab}^{i}\).
  Let us denote by \(A^{i}\) the set \(W^{i} \setminus (W^{i-1} \union W_{\sdab}^{i})\), that is the set of all vertices in \(W^{i}\) but not in \(W^{i-1} \union W_{\sdab}^{i}\) from which \(\PlayerMain\) has a sure-attractor strategy to reach \(W^{i-1} \union W_{\sdab}^{i}\).
  Thus, we have for all \(1 \le i \le k\) that \(W^{i-1}\), \(W_{\sdab}^{i}\), and \(A^{i}\) form a partition of \(W^{i}\).
  Moreover, the entire set \(W\) can be divided into pairwise disjoint sets \(W^{0}\), \(W_{\sdab}^{1}\), \(A^{1}\), \(W_{\sdab}^{2}\), \(A^{2}\), \(\ldots,\) \(W_{\sdab}^{k}\), and \(A^{k}\).
\item Since the algorithm terminates in the \(k^{\text{th}}\) iteration, we have that \(W_{\sdab}^{k} = \emptyset\) and thus also \(A^{k} = \emptyset\).
  We get that \(W^{k} = W^{k - 1}\).
\item For \(1 \le i < j \le k\), the sets \(P^{i}\) and \(P^{j}\) are (in general) incomparable, that is, there may exist vertices \(u, v\) such that \(u \in P^{i}\) and \(u \not\in P^{j}\), and \(v \not\in P^{i}\) and \(v \in P^{j}\).
  We may have \(u \in P^{i}\) and \(u \not\in P^{j}\) if \(u\) is a vertex that has out-neighbours in \(W^{i-1}\) and in \(W^{j-1} \setminus W^{i-1}\) and nowhere else.
  We may have \(v \not\in P^{i}\) and \(v \in P^{j}\) if \(v\) is a vertex that has out-neighbours in \(W^{j-1} \setminus W^{i-1}\) and in \(\Vertices \setminus W^{j-1}\) and nowhere else.
\end{itemize}
\begin{example}
\label{exa:sasf}
  \begin{figure}[t]
    \centering
    \scalebox{0.9}{
      \begin{tikzpicture}[node distance=1.75cm,
        inner sep=10pt]

        \tikzset{
          every state/.append style={minimum size=1.5pt, inner sep=1.75pt},
          random/.append style={diamond, thick, fill=gray!20, minimum size=1.5pt, inner sep=1pt},
          notineg/.append style={thick, dotted},
          inediamond/.append style={dashed},
        }
        \draw[rounded corners=6pt, fill opacity=0.05]
        (-0.7, 0) rectangle ++(13.2, 6);

        \draw[fill=black, fill opacity=0.02]
        (12.5, 5.6) {[rounded corners=6pt] -- ++(-11.9, 0) -- ++(0, -5.6) -- ++(+11.9, 0)} -- cycle;
        
        \draw[fill=orange, fill opacity=0.02]
        (12.5, 2) {[rounded corners=6pt] -- ++(-4, 0)  -- ++(0, -2) -- ++(+4, 0)} -- cycle;

        \draw[fill=blue, fill opacity=0.02]
        (12.3, 3.5) {[rounded corners=6pt] -- ++(-3.8, 0)  -- ++(0,-1.25) -- ++(+3.8, 0) --  cycle};
        \draw[fill=orange, fill opacity=0.02]
        (12.5, 5.2) {[rounded corners=6pt] -- ++(-4.5,0) -- ++(0,-5.2) -- ++(+4.5,0)} -- cycle;

        \draw[fill=blue, fill opacity=0.02]
        (7.8, 3.5) {[rounded corners=6pt] -- ++(-1.9, 0)  -- ++(0,-3.35) -- ++(+1.9, 0) --  cycle};
        \draw[fill=orange, fill opacity=0.02]
        (12.5, 5.2) {[rounded corners=6pt] -- ++(-7,0) -- ++(0,-5.2) -- ++(+7,0)} -- cycle;

        \draw[fill=orange, fill opacity=0.02]
        (12.5, 5.2) {[rounded corners=6pt] -- ++(-8,0) -- ++(0, -5.2) -- ++(+8, 0)} -- cycle;

        \draw[fill=blue, fill opacity=0.02]
        (4.3, 3.5) {[rounded corners=6pt] -- ++(-1.9, 0)  -- ++(0, -3.35) -- ++(+1.9
          , 0) --  cycle};
        \draw[fill=orange, fill opacity=0.02]
        (12.5, 5.2) {[rounded corners=6pt] -- ++(-10.5, 0) -- ++(0, -5.2) -- ++(+10.5, 0)} -- cycle;

        \node (s1) at (  10, 1)     [random, draw] {\(s_1\)};
        \node (s2) at (11.5, 1.5)   [state] {\(s_2\)};
        \node (s3) at (11.5, 0.5)   [state] {\(s_3\)};

        \node (s4) at (  10, 2.8)  [random, draw] {\(s_4\)};
        \node (s5) at (11.5, 2.8)  [state] {\(s_5\)};

        \node (s6) at ( 9.5, 4.3)   [random, draw] {\(s_6\)};
        \node (s7) at (11,   4.3)   [state] {\(s_7\)};

        \node (t1) at (7.3, 0.6)    [random, draw] {\(t_1\)};
        \node (t2) at (7.2, 1.75)   [random, draw] {\(t_2\)};
        \node (t3) at (7.2, 2.9)   [state] {\(t_3\)};

        \node (t4) at (6.7, 4.7)   [state] {\(t_4\)};
        \node (t5) at (7.5, 4.2)   [random, draw] {\(t_5\)};

        \node (u1) at (3.5,  .6)   [random, draw] {\(u_1\)};
        \node (u2) at (3.5, 1.7)   [state] {\(u_2\)};
        \node (u3) at (3.85, 3)    [random, draw] {\(u_3\)};

        \node (v1) at (1.3, 1.2)   [state] {\(v_1\)};
        \node (v2) at (1.3, 3)   [random, draw] {\(v_2\)};

                \node (x1) at (-0.1, 1.2)   [state] {\(x_1\)};
        \node (x2) at (-0.1, 3)   [random, draw] {\(x_2\)};

        \node (V) at (-0.4, 5.7)    {\(V\)}; 
        
        \node (Walpha) at (1.03, 5.25)    {\(W_{\Sure\GuaranteeThreshold}\)};
        
        \node (Ak) at (2.55, 4.9)    {\(W^{k-1}\)};
        \node (A2) at (5.85, 4.9)    {\(W^{2}\)};
        \node (A1) at (8.35, 4.9)    {\(W^{1}\)};

        \node (wsdabk) at (2.95, 3.12)    {\(W^{k-1}_{\sdab}\)};
        \node (dots) at (5, 2.3)    {\(\cdots\)};

        \node (wsdab2) at (6.4, 3.12)    {\(W^{2}_{\sdab}\)};
        \node (wsdab1) at (9.03, 3.12)    {\(W^{1}_{\sdab}\)};

        \node (wbeta) at (8.95, 1.6)    {\(W_{\Sure\AlmostSureThreshold}\)};

        \draw
        (s1) edge[bend right=10] (s2)
        (s1) edge[bend left=10] (s3)

        (s2) edge[loop right] node[below, inner sep=5pt]{\small \(\EdgeValues{5}{}\)} (s2)

        (s3) edge[bend right=10] (s2)
        
        (s4) edge[bend right=10] (s2)
        (s4) edge[bend left=15] node[above, inner sep=2pt]{\small \(\EdgeValues{0}{}\)} (s5)

        (s5) edge[bend left=15] node[below, inner sep=2pt]{\small \(\EdgeValues{0}{}\)} (s4)

        (s6) edge[bend right=10] (s1)
        (s6) edge[out=-60, in=130] (s5)    
        (s6) edge[bend left=15] node[above, inner sep=2pt]{\small \(\EdgeValues{- 1}{}\)} (s7)

        (s7) edge[bend left=15] node[below, inner sep=2pt]{\small \(\EdgeValues{- 1}{}\)} (s6)    
        (s7) edge[bend left=10] node[right, inner sep=2pt, pos=0.3]{\small \(\EdgeValues{- 1}{}\)} (s5)

        (t1) edge[bend right=20]node[below right, inner sep=2pt]{\small \(\EdgeValues{}{}\)} (s1)
        (t1) edge node[right, inner sep=2pt]{\small \(\EdgeValues{- 1}{}\)}  (t2)
        
        (t2) edge node[right, inner sep=2.5pt]{\small \(\EdgeValues{2}{}\)}  (t3)
        
        (t2) edge[out=30, in=190] (s6)

        (t3) edge[bend right=40] node[left, inner sep=2pt]{\small \(\EdgeValues{- 1}{}\)} (t1)
        
        (t4) edge (t5)
        (t4) edge ++(-1.75, -0.75)

        (t5) edge node[left, inner sep=2pt, pos=0.2]{\small \(\EdgeValues{- 1}{}\)}  (t3)
        (t5) edge[bend left=20] node[left, inner sep=2pt, pos=0.2]{\small \(\EdgeValues{}{}\)}  (s6)

        (u1) edge[loop right] node[above, inner sep=4pt]{\small \(\EdgeValues{1}{}\)}  (u1)
        (u1) edge[bend left=10] node[left, inner sep=4pt]{\small \(\EdgeValues{2}{}\)}  (u2)
        
        (u2) edge[loop right] (u2)
        (u2) edge[bend left=15] node[left, inner sep=2pt]{\small \(\EdgeValues{0}{}\)}  (u3)
        (u2) edge (v1)

        (u3) edge[bend left=15] node[right, inner sep=2pt]{\small \(\EdgeValues{0}{}\)}  (u2)
        (u3) edge ++(0.9, 0.25)

        (v1) edge (x1)
        (v1) edge[bend left=15] node[left, inner sep=2pt, pos=0.3]{\small \(\EdgeValues{0}{}\)}  (v2)
        (v1) edge[loop below] node[below right, inner sep=1.5pt]{\small \(\EdgeValues{0}{}\)} (v1)

        (v2) edge[bend left=15] node[right, inner sep=1.5pt, pos=0.3] {\small \(\EdgeValues{ - 1}{}\)}  (v1)
        (v2) edge[loop above] node[above left, inner sep=1.5pt]{\small \(\EdgeValues{2}{}\)} (v2)
        (v2) edge[bend left=15] (u2)

        (x1) edge[bend left=15] node[left, inner sep=2pt, pos=0.3]{\small \(\EdgeValues{0}{}\)}  (x2)

        (x2) edge (v2)
        (x2) edge[bend left=15] node[right, inner sep=1.5pt, pos=0.3] {\small \(\EdgeValues{ - 1}{}\)}  (x1)
        (x2) edge[loop above] (x2)
        ;
      \end{tikzpicture}
    }
    \caption{%
      An illustration of various features of \Cref{alg:sure-fwmp-almostsure-fwmp} to compute the winning region for sure-\(\FWMP(\WindowLength = 3, \GuaranteeThreshold = 0)\)-almost-sure-\(\FWMP(\WindowLength = 3, \AlmostSureThreshold = 5)\).
      Only relevant edge payoffs are shown.
    }
    \label{fig:sasf-algo-sets}
  \end{figure}
  We show some features of \Cref{alg:sure-fwmp-almostsure-fwmp} in \Cref{fig:sasf-algo-sets} for window length \(\WindowLength = 3\) and for thresholds \(\GuaranteeThreshold = 0\) and \(\AlmostSureThreshold = 5\).
  Vertices \(\{x_{1}, x_{2}\}\) are not winning for sure-\(\FWMP(\WindowLength, \GuaranteeThreshold)\) and are removed.
  The set \(W_{\Sure\AlmostSureThreshold}\) (also referred to as \(W^{0}\)) consists of \(\{s_{1}, s_{2}, s_{3}\}\) since from all three vertices the token reaches \(s_{2}\) surely where all \(\AlmostSureThreshold\)-windows close in at most \(\WindowLength\) steps.
  Next, we compute the set \(P^{1}\) to be \(\{s_4, s_6, t_1\}\), the set of all probabilistic vertices with at least one out-edge in \(W^0\) and at least one out-edge not in \(W^0\).
  In the subMDP closure induced by \(W_{\Sure\GuaranteeThreshold} \setminus W^{0}\), we compute $W_{\sdab}^{1}$, the winning region for sure-\(\DirFWMP(\WindowLength, \GuaranteeThreshold)\)-almost-sure-\(\BuchiObj(P^{1})\), to be $\{s_4,s_5\}$.
  The set \(A^{1}\) of all vertices in \(W_{\Sure\GuaranteeThreshold} \setminus (W_{\sdab}^1 \union W^{0})\) from which \(\PlayerMain\) can ensure that the token reaches \(W_{\sdab}^{1} \union W^0\) surely is \(\{s_6, s_7\}\).
  We come to the end of the first iteration of the \textbf{repeat} loop and we have that \(W^1 = A^1 \union W_{\sdab}^{1} \union W^0 = \{s_1, s_2, s_3, s_4, s_5, s_6, s_7\} \). 
  Since \(W_{\sdab}^1\) is non-empty, we have another iteration of the \textbf{repeat} loop.

  We compute \(P^2 = \{t_1, t_2, t_5\}\) for \(W^1\).
  The winning region for sure-\(\DirFWMP(\WindowLength, \GuaranteeThreshold)\)-almost-sure-\(\BuchiObj(P^{2})\) in \(\subMDPClosure{\MDP}{W_{\Sure\GuaranteeThreshold} \setminus W^{1}}\) is \(W_{\sdab}^2 = \{t_1, t_2, t_3\}\), and the set \(A^2\) is \(\{t_4, t_5\}\).
  We get that \(W^2 = \{t_1, t_2, t_3, t_4, t_5\} \union W^1\).
  We repeat this until we reach \(W^{k - 1}\).
  Finally, we have that \(P^{k} = \{v_{2}\}\) and \(W_{\sdab}^{k} = \emptyset\).
  Thus, the algorithm terminates here and returns \(W^{k}\).
  \lipicsEnd
\end{example}

\subparagraph{Proof of correctness of \Cref{alg:sure-fwmp-almostsure-fwmp}.}
We prove some intermediate lemmas before we prove the correctness of \Cref{alg:sure-fwmp-almostsure-fwmp} in \Cref{thm:sasf-correctness}.
We have given intuition on why vertices in \(W^{k}\) are winning for sure-\(\FWMP(\WindowLength, \GuaranteeThreshold)\)-almost-sure-\(\FWMP(\WindowLength, \AlmostSureThreshold)\) in \(\MDP\).
For the converse, since vertices in \(V \setminus W_{\Sure\GuaranteeThreshold}\) are not winning for sure-\(\FWMP(\WindowLength, \GuaranteeThreshold)\), they are also not winning for sure-\(\FWMP(\WindowLength, \GuaranteeThreshold)\)-almost-sure-\(\FWMP(\WindowLength, \AlmostSureThreshold)\).
For the remaining set \(W_{\Sure\GuaranteeThreshold}\setminus W^{k}\) of vertices, we show in \Cref{lem:sasf-sure-dirfwmp-almost-sure-reach} and \Cref{lem:sasf-sure-fwmp-almost-sure-reach} that starting from \(W_{\Sure\GuaranteeThreshold}\setminus W^{k}\), player \(\PlayerMain\) cannot satisfy sure-\(\DirFWMP(\WindowLength, \GuaranteeThreshold)\)-almost-sure-\(\ReachObj(W^{k})\) or sure-\(\FWMP(\WindowLength, \GuaranteeThreshold)\)-almost-sure-\(\ReachObj(W^{k})\) respectively.
Thus, starting from \(W_{\Sure\GuaranteeThreshold} \setminus W^{k}\), if \(\PlayerMain\) wants to satisfy sure-\(\DirFWMP(\WindowLength, \GuaranteeThreshold)\) or sure-\(\FWMP(\WindowLength, \GuaranteeThreshold)\), then the token remains in \(W_{\Sure\GuaranteeThreshold} \setminus W^{k}\) with some positive probability.
Then, we show in \Cref{lem:sasf-fwmp-beta-empty} that the winning region for sure-\(\FWMP(\WindowLength, \AlmostSureThreshold)\) in \(\subMDPClosure{\MDP}{W_{\Sure\GuaranteeThreshold} \setminus W^{k}}\) is empty, and in \Cref{lem:sasf-w-empty-fwmp-empty} that the winning region for sure-\(\FWMP(\WindowLength, \GuaranteeThreshold)\)-almost-sure-\(\FWMP(\WindowLength, \AlmostSureThreshold)\) in \(\subMDPClosure{\MDP}{W_{\Sure\GuaranteeThreshold} \setminus W^{k}}\) is also empty.
Then, these lemmas together imply that if the token remains in \(W_{\Sure\GuaranteeThreshold} \setminus W^{k}\), then sure-\(\FWMP(\WindowLength, \GuaranteeThreshold)\)-almost-sure-\(\FWMP(\WindowLength, \AlmostSureThreshold)\) cannot be satisfied.

\begin{lemma}\label{lem:sasf-sure-dirfwmp-almost-sure-reach}
  For all vertices \(v \in W_{\Sure\GuaranteeThreshold} \setminus W^{k}\) in the MDP \(\MDP\), the vertex \(v\) is not winning for sure-\(\DirFWMP(\WindowLength, \GuaranteeThreshold)\)-almost-sure-\(\ReachObj(W)\) in \(\MDP\).
\end{lemma}
\begin{proof}
  For each vertex \(v \in (W_{\Sure\GuaranteeThreshold} \setminus W^{k}) \intersection \VerticesMain\) belonging to \(\PlayerMain\), there is no out-neighbour of \(v\) in \(W^{k}\), because otherwise \(v\) would belong to the sure-attractor of \(W^{k}\), and hence, also to \(W^{k}\) (Line~\ref{alg-line:sasf-sure-attr}), which is a contradiction since we started with \(v \in W_{\Sure\GuaranteeThreshold} \setminus W^{k}\).
  A similar argument gives us that for all probabilistic vertices \(v \in (W_{\Sure\GuaranteeThreshold} \setminus W^{k}) \intersection \VerticesRandom\), it is not the case that all out-neighbours of \(v\) belong to \(W^{k}\), because that would imply that \(v\) belongs to the sure-attractor of \(W^{k}\), and hence, also to \(W^{k}\), which is a contradiction.
  Thus, if \(v\) is a vertex in \(W_{\Sure\GuaranteeThreshold} \setminus W^{k}\) with an out-neighbour in \(W^{k}\), then it follows that \(v\) is a probabilistic vertex and that \(v\) has at least one out-neighbour in \(W_{\Sure\GuaranteeThreshold} \setminus W^{k}\).
  Note that the set of all such vertices is precisely the set \(P^{k}\) computed on Line~\ref{alg-line:sasf-poscpre}.

  Thus, starting from a vertex \(v \in W_{\Sure\GuaranteeThreshold} \setminus W^{k}\), if the token reaches \(W^{k}\), then token reaches \(W^{k}\) through a vertex in \(P^{k}\).
  Moreover, starting from \(v\), the token can almost-surely reach \(W^{k}\) in \(\subMDP{\MDP}{ W_{\Sure\GuaranteeThreshold}}\) if and only if starting from \(v\), the token can almost-surely visit \(P^{k}\) infinitely often in \(\subMDPClosure{\MDP}{W_{\Sure\GuaranteeThreshold} \setminus W^{k}}\).
  Further, the token can almost-surely reach \(W^{k}\) from \(v\) while also satisfying sure-\(\DirFWMP(\WindowLength, \GuaranteeThreshold)\) in \(\subMDP{\MDP}{ W_{\Sure\GuaranteeThreshold}}\) if and only if starting from \(v\), the token can almost-surely visit \(P^{k}\) infinitely often while also satisfying sure-\(\DirFWMP(\WindowLength, \GuaranteeThreshold)\) in \(\subMDPClosure{\MDP}{W_{\Sure\GuaranteeThreshold} \setminus W^{k}}\).
  Since at the end of \Cref{alg:sure-fwmp-almostsure-fwmp}, the set \(W_{\sdab}^{k}\) that is the winning region for sure-\(\DirFWMP(\WindowLength, \GuaranteeThreshold)\)-almost-sure-\(\BuchiObj(P^{k})\) is empty, we get that no vertex in \(W_{\Sure\GuaranteeThreshold} \setminus W^{k}\) is winning for sure-\(\DirFWMP(\WindowLength, \GuaranteeThreshold)\)-almost-sure-\(\ReachObj(W^{k})\) in \(\subMDP{\MDP}{ W_{\Sure\GuaranteeThreshold}}\).
\end{proof}

\begin{lemma}\label{lem:sasf-sure-fwmp-almost-sure-reach}
  For all vertices \(v \in W_{\Sure\GuaranteeThreshold} \setminus W^{k}\) in the MDP \(\MDP\), the vertex \(v\) is not winning for sure-\(\FWMP(\WindowLength, \GuaranteeThreshold)\)-almost-sure-\(\ReachObj(W^{k})\) in \(\MDP\).
\end{lemma}
\begin{proof}
  Observe that a vertex \(v \in W_{\Sure\GuaranteeThreshold} \setminus W^{k}\) is winning for sure-\(\FWMP(\WindowLength, \GuaranteeThreshold)\)-almost-sure-\(\ReachObj(W^{k})\) in \(\subMDP{\MDP}{ W_{\Sure\GuaranteeThreshold}}\) if and only if \(v\) is winning for sure-\(\FWMP(\WindowLength, \GuaranteeThreshold)\)-almost-sure-\(\BuchiObj(P^{k})\) in \(\subMDPClosure{\MDP}{W_{\Sure\GuaranteeThreshold} \setminus W^{k}}\).
  From \Cref{lem:sasf-sure-dirfwmp-almost-sure-reach}, we have that no vertex is winning for sure-\(\DirFWMP(\WindowLength, \GuaranteeThreshold)\)-almost-sure-\(\BuchiObj(P^{k})\) in \(\subMDPClosure{\MDP}{W_{\Sure\GuaranteeThreshold} \setminus W^{k}}\).
  That is, for every vertex \(v \in W_{\Sure\GuaranteeThreshold} \setminus W^{k}\), for every strategy \(\Strategy\) of \(\PlayerMain\) from \(v\) that is winning for almost-sure-\(\BuchiObj(P^{k})\), there exists an outcome of this strategy that contains an open \(\GuaranteeThreshold\)-window of length \(\WindowLength\).
  Since this is true for every vertex in \(W_{\Sure\GuaranteeThreshold} \setminus W^{k}\), it is also the case that there exists an outcome of \(\Strategy\) with infinitely many open \(\GuaranteeThreshold\)-windows of length \(\WindowLength\).
  Thus, no vertex in \(W_{\Sure\GuaranteeThreshold} \setminus W^{k}\) is winning for \(\PlayerMain\) for sure-\(\FWMP(\WindowLength, \GuaranteeThreshold)\)-almost-sure-\(\ReachObj(W^{k})\) in \(\subMDP{\MDP}{ W_{\Sure\GuaranteeThreshold}}\).
\end{proof}

\begin{lemma}\label{lem:sasf-fwmp-beta-empty}
  The winning region for sure-\(\FWMP(\WindowLength, \AlmostSureThreshold)\) in \(\subMDPClosure{\MDP}{W_{\Sure\GuaranteeThreshold} \setminus W^{k}}\) is empty.
\end{lemma}
\begin{proof}
  Suppose towards a contradiction that the winning region for sure-\(\FWMP(\WindowLength, \AlmostSureThreshold)\) in \(\subMDPClosure{\MDP}{W_{\Sure\GuaranteeThreshold} \setminus W^{k}}\) is non-empty.
  Then, from~\cite{CDRR15}, we have that the winning region for sure-\(\DirFWMP(\WindowLength, \AlmostSureThreshold)\) in \(\subMDPClosure{\MDP}{W_{\Sure\GuaranteeThreshold} \setminus W^{k}}\) is also non-empty.
  Now one of the following two cases must occur.
  We show that both cases lead to contradictions.
  \begin{itemize}
  \item There exists a vertex \(v\) in this region that is winning for sure-\(\DirFWMP(\WindowLength, \AlmostSureThreshold)\)-sure-\(\SafeObj(W_{\Sure\GuaranteeThreshold} \setminus (P^{k} \union W^{k}))\) in \(\subMDPClosure{\MDP}{W_{\Sure\GuaranteeThreshold} \setminus W^{k}}\), and there exists such a winning strategy \(\Strategy\).
    All outcomes of \(\Strategy\) from \(v\) in \(\subMDPClosure{\MDP}{W_{\Sure\GuaranteeThreshold} \setminus W^{k}}\) satisfy \(\DirFWMP(\WindowLength, \AlmostSureThreshold)\) and stay in \(W_{\Sure\GuaranteeThreshold} \setminus (P^{k} \union W^{k})\) forever.
    Then, the same strategy \(\Strategy\) is also winning for sure-\(\DirFWMP(\WindowLength, \AlmostSureThreshold)\)-sure-\(\SafeObj(W_{\Sure\GuaranteeThreshold} \setminus (P^{k} \union W^{k}))\) in the larger MDP \(\subMDP{\MDP}{ W_{\Sure\GuaranteeThreshold}}\).
    In particular, the vertex \(v\) is winning for sure-\(\FWMP(\WindowLength, \AlmostSureThreshold)\) in \(\subMDP{\MDP}{ W_{\Sure\GuaranteeThreshold}}\), and thus \(v\) belongs to \(W_{\Sure\AlmostSureThreshold}\), and also to \(W^{k}\).
  This is a contradiction since we started with \(v \in W_{\Sure\GuaranteeThreshold} \setminus W^{k}\).    
  \item No vertex in this region is winning for sure-\(\DirFWMP(\WindowLength, \AlmostSureThreshold)\)-sure-\(\SafeObj(W_{\Sure\GuaranteeThreshold} \setminus (P^{k} \union W^{k}))\) in \(\subMDPClosure{\MDP}{W_{\Sure\GuaranteeThreshold} \setminus W^{k}}\).
    That is, for all vertices \(v\) in the region, for all strategies \(\Strategy\) of \(\PlayerMain\) from \(v\) that are winning for sure-\(\DirFWMP(\WindowLength, \AlmostSureThreshold)\), there exists an outcome of \(\Strategy\) in \(\subMDPClosure{\MDP}{W_{\Sure\GuaranteeThreshold} \setminus W^{k}}\) from \(v\) that reaches \(P^{k}\) in the next \(\abs{V} \cdot \WindowLength\) steps (recall that the memory required by \(\PlayerMain\) for sure-\(\DirFWMP(\WindowLength, \AlmostSureThreshold)\) is \(\WindowLength\)~\cite{DGG25}).
    Thus, if \(\PlayerMain\) follows this strategy \(\Strategy\) from a vertex \(v\) in \(\subMDPClosure{\MDP}{W_{\Sure\GuaranteeThreshold} \setminus W^{k}}\), then from each point of the play, the probability of reaching \(P^{k}\) in the next \(\abs{V} \cdot \WindowLength\) steps is positive (at least \(\ProbabilityFunction_{\min}^{\abs{V} \cdot \WindowLength}\)).
    Hence, we have that with probability~\(1\), an outcome of \(\Strategy\) from \(v\) visits \(P^{k}\) infinitely often.
    This means that \(v\) is winning for    sure-\(\DirFWMP(\WindowLength, \AlmostSureThreshold)\)-almost-sure-\(\BuchiObj(P^{k})\), and thus also sure-\(\DirFWMP(\WindowLength, \GuaranteeThreshold)\)-almost-sure-\(\BuchiObj(P^{k})\) in \(\subMDPClosure{\MDP}{W_{\Sure\GuaranteeThreshold} \setminus W^{k}}\).
    This implies that \(v\) must belong to \(W_{\sdab}^{k}\), which is a contradiction since \(W_{\sdab}^{k}\) is empty.
    \qedhere
  \end{itemize}
\end{proof}

\begin{lemma}\label{lem:sasf-w-empty-fwmp-empty}
  The winning region for sure-\(\FWMP(\WindowLength, \GuaranteeThreshold)\)-almost-sure-\(\FWMP(\WindowLength, \AlmostSureThreshold)\) in \(\subMDPClosure{\MDP}{W_{\Sure\GuaranteeThreshold} \setminus W^{k}}\) is empty.
\end{lemma}
\begin{proof}
  From \Cref{lem:sasf-fwmp-beta-empty}, we have that the winning region for sure-\(\FWMP(\WindowLength, \AlmostSureThreshold)\) in \(\subMDPClosure{\MDP}{W_{\Sure\GuaranteeThreshold} \setminus W^{k}}\) is empty.
  From \cite[Lemma~6.1]{DGG25}, we have that if the winning region for sure-\(\FWMP(\WindowLength, \AlmostSureThreshold)\) in an MDP is empty, then the winning region for almost-sure-\(\FWMP(\WindowLength, \AlmostSureThreshold)\) in the MDP is also empty.
  Thus, we have that the winning region for almost-sure-\(\FWMP(\WindowLength, \AlmostSureThreshold)\) in \(\subMDPClosure{\MDP}{W_{\Sure\GuaranteeThreshold} \setminus W^{k}}\) is empty.
  It follows that \(\PlayerMain\) cannot satisfy sure-\(\FWMP(\WindowLength, \GuaranteeThreshold)\)-almost-sure-\(\FWMP(\WindowLength, \AlmostSureThreshold)\) from any vertex in  \(\subMDPClosure{\MDP}{W_{\Sure\GuaranteeThreshold} \setminus W^{k}}\).
\end{proof}

The following theorem shows the correctness of \Cref{alg:sure-fwmp-almostsure-fwmp}.
\begin{theorem}\label{thm:sasf-correctness}
  \Cref{alg:sure-fwmp-almostsure-fwmp} computes the winning region for sure-\(\FWMP(\WindowLength, \GuaranteeThreshold)\)-almost-sure-\(\FWMP(\WindowLength, \AlmostSureThreshold)\) in \(\MDP\).
\end{theorem}
\begin{proof}
  (\(\Rightarrow\))
  We show that starting from an arbitrary vertex \(v \in W^{k}\), player \(\PlayerMain\) has a winning strategy for sure-\(\FWMP(\WindowLength, \GuaranteeThreshold)\)-almost-sure-\(\FWMP(\WindowLength, \AlmostSureThreshold)\).
  We consider the cases \(v \in W^{0}\) and \(v \in W^{k} \setminus W^{0}\) separately:
  \begin{itemize}
  \item Suppose the token is on a vertex \(v\) in \(W^{0}\) (that is, \(W_{\Sure\AlmostSureThreshold}\), the winning region for \(\PlayerMain\) for sure-\(\FWMP(\WindowLength, \AlmostSureThreshold)\)).
    Since \(\FWMP(\WindowLength, \AlmostSureThreshold)\) is a prefix-independent objective, we have from \Cref{prop:winning-region-sure-objective-closed-under-suffixes-induces-submdp} that \(W^{0}\) induces a subMDP of \(\MDP\).
    Thus, starting from \(v\), player \(\PlayerMain\) has a strategy that is winning for sure-\(\FWMP(\WindowLength, \AlmostSureThreshold)\) and also keeps the token in \(W^{0}\) for all outcomes of the strategy.
    Since \(\GuaranteeThreshold < \AlmostSureThreshold\), we have that satisfaction of sure-\(\FWMP(\WindowLength, \AlmostSureThreshold)\) implies satisfaction of  sure-\(\FWMP(\WindowLength, \GuaranteeThreshold)\)-almost-sure-\(\FWMP(\WindowLength, \AlmostSureThreshold)\), and thus \(v\) is winning for sure-\(\FWMP(\WindowLength, \GuaranteeThreshold)\)-almost-sure-\(\FWMP(\WindowLength, \AlmostSureThreshold)\).
  \item Otherwise, suppose the token is on a vertex \(v\) in \(W^{k} \setminus W^{0}\).
    In particular, there exists \(1 \le i < k\) such that \(v \in W^{i} \setminus W^{i-1}\).
    Since each vertex in \(W^{i} \setminus W^{i-1}\) belongs to either \(A^{i}\) and \(W_{\sdab}^{i}\), we have the following cases:
    \begin{itemize}
    \item The token is on a vertex \(v\) in \(A^{i}\).
      Then \(\PlayerMain\) follows a sure-attractor strategy from \(v\) to surely reach a vertex in \(W_{\sdab}^{i}\) or a vertex \(v'\) in \(W^{i-1}\).
    \item The token is on a vertex \(v\) in \(W_{\sdab}^{i}\).
      Then, player \(\PlayerMain\) follows a strategy from \(v\) that is winning for sure-\(\DirFWMP(\WindowLength, \GuaranteeThreshold)\)-almost-sure-\(\BuchiObj(P^{i})\) to keep visiting vertices in \(P^{i}\) while surely always closing all \(\GuaranteeThreshold\)-windows in at most \(\WindowLength\) steps.
      The token either eventually reaches a vertex \(v'\) in \(W^{i-1}\) (which happens with probability~\(1\)) or remains in \(W_{\sdab}^{i}\) forever (which happens with probability~\(0\)).
    \end{itemize}
    If the token reaches a vertex \(v'\) in \(W^{i-1}\), then the vertex \(v'\) is either in \(W^{0}\) or in \(W^{i'} \setminus W^{i'-1}\) for some \(1 \le i' \le i - 1\), and we can repeat the above argument once again.
  \end{itemize}
  Repeating this sufficiently many times (at most \(k\) times), we get that from every vertex in \(W^{k}\), player \(\PlayerMain\) has a strategy such that with probability~\(1\) the token eventually reaches \(W^{0}\) where the \(\FWMP(\WindowLength, \AlmostSureThreshold)\) objective is surely satisfied, and with probability~\(0\) the token is stuck in \(W_{\sdab}^{i}\) for some \(1 \le i \le k\) where surely all \(\GuaranteeThreshold\)-windows close in at most \(\WindowLength\) steps.
  Thus, from all vertices in \(W^{k}\), player \(\PlayerMain\) has a strategy to satisfy sure-\(\FWMP(\WindowLength, \GuaranteeThreshold)\)-almost-sure-\(\FWMP(\WindowLength, \AlmostSureThreshold)\).

  (\(\Leftarrow\))
  Now, we show the converse, that is, if a vertex \(v\) does not belong to \(W^{k}\), then \(v\) is not winning for \(\PlayerMain\) for sure-\(\FWMP(\WindowLength, \GuaranteeThreshold)\)-almost-sure-\(\FWMP(\WindowLength, \AlmostSureThreshold)\).
  If \(v \not \in W^{k}\), then we have two cases: either \(v \in V \setminus W_{\Sure\GuaranteeThreshold}\) or \(v \in W_{\Sure\GuaranteeThreshold} \setminus W^{k}\).
  \begin{itemize}
  \item If \(v \in V \setminus W_{\Sure\GuaranteeThreshold}\),  then \(v\) is not winning for sure-\(\FWMP(\WindowLength, \GuaranteeThreshold)\), and thus \(v\) is also not winning for sure-\(\FWMP(\WindowLength, \GuaranteeThreshold)\)-almost-sure-\(\FWMP(\WindowLength, \AlmostSureThreshold)\).
  \item Otherwise, we have that \(v \in W_{\Sure\GuaranteeThreshold} \setminus W^{k}\).
    From \Cref{lem:sasf-sure-fwmp-almost-sure-reach}, we have that \(v\) is not winning for sure-\(\FWMP(\WindowLength, \GuaranteeThreshold)\)-almost-sure-\(\ReachObj(W^{k})\).
    Thus, for all strategies \(\Strategy\) of \(\PlayerMain\) that are winning for sure-\(\FWMP(\WindowLength, \GuaranteeThreshold)\), there exists \(0 < p < 1\) such that with probability at least~\(p\), the outcome of \(\Strategy\) from \(v\) remains in \(W_{\Sure\GuaranteeThreshold} \setminus W^{k}\) forever.
    Further, from \Cref{lem:sasf-w-empty-fwmp-empty}, it follows that with some positive probability, the outcome does not satisfy \(\FWMP(\WindowLength, \AlmostSureThreshold)\).
    Thus, sure-\(\FWMP(\WindowLength, \GuaranteeThreshold)\)-almost-sure-\(\FWMP(\WindowLength, \AlmostSureThreshold)\) cannot be satisfied by \(\PlayerMain\) from \(v\).
  \end{itemize}
  We have shown the converse direction as well, that is, if \(v \not\in W^{k}\), then \(v\) is not winning for \(\PlayerMain\) for  sure-\(\FWMP(\WindowLength, \GuaranteeThreshold)\)-almost-sure-\(\FWMP(\WindowLength, \AlmostSureThreshold)\).
\end{proof}

\paragraph*{\Cref{alg:sure-dirfwmp-almostsure-buchi}: Sure-\(\DirFWMP(\WindowLength, \GuaranteeThreshold)\)-almost-sure-\(\BuchiObj(T)\)}
Recall that \Cref{alg:sure-fwmp-almostsure-fwmp} uses \Cref{alg:sure-dirfwmp-almostsure-buchi} (\(\SureDirFWMPASBuchiAlgo\)) as a subroutine, and we describe this algorithm now.
\Cref{alg:sure-dirfwmp-almostsure-buchi} returns the set of all vertices in an MDP that are winning for sure-\(\DirFWMP(\WindowLength, \GuaranteeThreshold)\)-almost-sure-\(\BuchiObj(T)\) for a target set \(T \subseteq V\).
We note that it may be the case that every vertex in an MDP is winning for sure-\(\DirFWMP(\WindowLength, \GuaranteeThreshold)\) and almost-sure-\(\BuchiObj(T)\) separately, and yet no vertex in the MDP is winning for the combination sure-\(\DirFWMP(\WindowLength, \GuaranteeThreshold)\)-almost-sure-\(\BuchiObj(T) \).
We show such an MDP in \Cref{fig:sdab-simultaneous-satisfaction-not-possible}.
\begin{figure}[t]
  \centering

  \begin{tikzpicture}[node distance=1.5cm]
    \node[state, initial above, accepting] (v1) {\(v_{1}\)};
    \node[state, right of=v1] (v2) {\(v_{2}\)};

    \draw 
    (v1) edge[loop left] node[auto]{\small \(\EdgeValues{-1}{}\)} (v1)
    ;
    \draw 
    (v1) edge[bend left] node[auto]{\small \(\EdgeValues{0}{}\)} (v2)
    (v2) edge[bend left] node[auto]{\small \(\EdgeValues{-4}{}\)} (v1)
    ;
    \draw 
    (v2) edge[loop right] node[auto]{\small \(\EdgeValues{2}{}\)} (v2)
    ;
  \end{tikzpicture}

  \caption{%
    Both vertices are winning for sure-\(\DirFWMP(2, 0)\) and for almost-sure-\(\BuchiObj(\{v_{1}\}) \) (and even sure-\(\BuchiObj(\{v_{1}\}) \)), but neither vertex is winning for sure-\(\DirFWMP(2, 0)\)-almost-sure-\(\BuchiObj( \{v_{1}\}) \).
  }
  \label{fig:sdab-simultaneous-satisfaction-not-possible}
\end{figure}

\begin{algorithm}[t]
  \caption{\(\SureDirFWMPASBuchiAlgo(\MDP, \WindowLength, \GuaranteeThreshold, T)\)}%
  \label{alg:sure-dirfwmp-almostsure-buchi}
  \begin{algorithmic}[1]
    \Require An MDP \(\MDP = ((\Vertices, \Edges), (\VerticesMain, \VerticesRandom), \ProbabilityFunction, \PayoffFunction)\), window length \(\WindowLength\), threshold \(\GuaranteeThreshold\), and target set \(T\)
    \Ensure The winning region for sure-\(\DirFWMP(\WindowLength, \GuaranteeThreshold)\)-almost-sure-\(\BuchiObj(T)\) in \(\MDP\)
    \State \(W_{\sdpr} \assign \SureDirFWMPPosReachAlgo(\MDP, \WindowLength, \GuaranteeThreshold, T)\)\label{alg-line:sdfab-sdpr}
    \If{\(V = W_{\sdpr}\)}\label{alg-line:sdfab-until}
    \State \Return \(V\)\label{alg-line:sdfab-return-winning-region}
    \Else
    \State \(S \assign \SureSafeAlgo(\MDP, W_{\sdpr})\)\label{alg-line:sdfab-almost-sure-attr}
    \State \Return \(\SureDirFWMPASBuchiAlgo(\subMDP{\MDP}{S}, \WindowLength, \GuaranteeThreshold, T \intersection S)\)\label{alg-line:sdfab-restrict-to-almost-sure-attr}
    \EndIf
  \end{algorithmic}
\end{algorithm}

\subparagraph{Description of \Cref{alg:sure-dirfwmp-almostsure-buchi}.}
We begin by calling \Cref{alg:sure-dirfwmp-positive-reach} (\SureDirFWMPPosReachAlgo) in Line~\ref{alg-line:sdfab-sdpr} to compute \(W_{\sdpr}\), the winning region for sure-\(\DirFWMP(\WindowLength, \GuaranteeThreshold)\)-positive-\(\ReachObj(T)\) in \(\MDP\).
(The subscript \(\sdpr\) is short for sure-\(\DirFWMP(\WindowLength, \GuaranteeThreshold)\)-positive-\(\ReachObj(T)\).)
If every vertex in \(\MDP\) belongs to \(W_{\sdpr}\), then we go to Line~\ref{alg-line:sdfab-return-winning-region} and return \(V\), that is all vertices in \(\MDP\).
Otherwise, we have that \(W_{\sdpr}\) is a strict subset of \(V\) and we go to Line~\ref{alg-line:sdfab-almost-sure-attr} to compute \(S\), the winning region for sure-\(\SafeObj(W_{\sdpr})\), that is the set of all vertices in \(\MDP\) from which \(\PlayerMain\) can ensure for all outcomes that the token visits only vertices in \(W_{\sdpr}\).
The set \(S\) induces a subMDP of \(\MDP\), and in Line~\ref{alg-line:sdfab-restrict-to-almost-sure-attr} we recursively call \Cref{alg:sure-dirfwmp-almostsure-buchi} (\(\SureDirFWMPASBuchiAlgo\)) on the subMDP \(\subMDP{\MDP}{S}\).

\subparagraph{Proof of correctness of \Cref{alg:sure-dirfwmp-almostsure-buchi}.}
We prove the correctness of \Cref{alg:sure-dirfwmp-almostsure-buchi} in \Cref{thm:sdfab-correctness}.
Before that, we prove an intermediate result in \Cref{lem:sdfab-v-l-bound}.

\begin{lemma}\label{lem:sdfab-v-l-bound}
  For all vertices \(v\) in an MDP \(\MDP\), if \(\PlayerMain\) has a strategy to reach a target \(T \subseteq \Vertices\) from \(v\) with positive probability while simultaneously satisfying  sure-\(\DirFWMP(\WindowLength, \GuaranteeThreshold)\), then \(\PlayerMain\) also has a strategy to reach \(T\) from \(v\) with positive probability in at most \(\abs{\Vertices} \cdot \WindowLength\) steps while simultaneously satisfying  sure-\(\DirFWMP(\WindowLength, \GuaranteeThreshold)\).
\end{lemma}
\begin{proof}
  Since the \(\DirFWMP(\WindowLength, \GuaranteeThreshold)\) objective is closed under suffixes, we have from \Cref{prop:winning-region-sure-objective-closed-under-suffixes-induces-submdp} that  \(\PlayerMain\) can ensure that the token never leaves the winning region of sure-\(\DirFWMP(\WindowLength, \GuaranteeThreshold)\), and moreover, it is also necessary for \(\PlayerMain\) to keep the token within this winning region in order to satisfy  sure-\(\DirFWMP(\WindowLength, \GuaranteeThreshold)\).
  Thus, without loss of generality, let us suppose that all vertices in the MDP are winning for sure-\(\DirFWMP(\WindowLength, \GuaranteeThreshold)\), and we only consider strategies of \(\PlayerMain\) that are winning for sure-\(\DirFWMP(\WindowLength, \GuaranteeThreshold)\).

  Now, suppose for a vertex \(v\) in \(\MDP\), there exists a winning strategy for sure-\(\DirFWMP(\WindowLength, \GuaranteeThreshold)\) that is also winning for positive-\(\ReachObj(T)\).
  Then, we want to show that there exists a winning strategy for sure-\(\DirFWMP(\WindowLength, \GuaranteeThreshold)\) following which \(\PlayerMain\) can, starting from \(v\), with positive probability, reach \(T\) from \(v\) in at most \(\abs{\Vertices} \cdot \WindowLength\) steps.

  Suppose that this is not the case, that is, suppose that for all strategies of \(\PlayerMain\) that are winning for sure-\(\DirFWMP(\WindowLength, \GuaranteeThreshold)\) and for positive-\(\ReachObj(T)\) from \(v\), the probability that the token reaches \(T\) in at most \(\abs{\Vertices} \cdot \WindowLength\) steps starting from \(v\) is zero.
  Let \(x > \abs{\Vertices} \cdot \WindowLength\) be the minimum number of steps that \(\PlayerMain\) can ensure that the token reaches \(T\) from \(v\) in with positive probability.
  Consider one such outcome where the token reaches \(T\) from \(v\) in \(x\) steps
  We also have that this outcome satisfies \(\DirFWMP(\WindowLength, \GuaranteeThreshold)\), and thus, all \(\GuaranteeThreshold\)-windows in this outcome close in at most \(\WindowLength\) steps.
  Moreover, from the inductive property of windows, we have that there exist infinitely many positions in this outcome where there are no open \(\GuaranteeThreshold\)-windows pending to be closed.

  Consider the prefix \(v_{0}v_{1}\cdots v_{x}\) of this outcome of length \(x\), that is, the prefix that ends with the first visit to a vertex in \(T\).
  We are interested in the positions in this prefix where there are no open \(\GuaranteeThreshold\)-windows.
  At the start of the play (that is, at \(v_{0}\)) there are no open \(\GuaranteeThreshold\)-windows.
  The \(\GuaranteeThreshold\)-window starting at \(v_{0}\) closes in at most \(\WindowLength\) steps, say at \(v_{i_{1}}\), and when this \(\GuaranteeThreshold\)-window closes, there are once again no open \(\GuaranteeThreshold\)-windows at \(v_{i_{1}}\) by the inductive property of windows.
  When the \(\GuaranteeThreshold\)-window starting at \(v_{i_{1}}\) closes, say at \(v_{i_{2}}\), this is once again a position where there are no open \(\GuaranteeThreshold\)-windows.
  This way, the distance between every two consecutive positions \(v_{i_{j}}\) and \(v_{i_{j+1}}\) with no open \(\GuaranteeThreshold\)-windows is at most \(\WindowLength\).
  Since the prefix \(v_{0}v_{1}\cdots v_{x}\) has at least \(\abs{\Vertices} \cdot \WindowLength + 1\) vertices and \(v_{0}\) is a position with no open \(\GuaranteeThreshold\)-windows, there exist at least \(\abs{\Vertices} + 1\) positions in the prefix where there are no open \(\GuaranteeThreshold\)-windows.
  Since there are only \(\abs{\Vertices}\) vertices in the game, by the pigeonhole principle, two of these positions (say \(v_{j}\) and \(v_{k}\)) have the same vertex \(u\) of the game.
  The following two cases are possible:
  \begin{itemize}
  \item If \(u \in \VerticesMain\), then \(\PlayerMain\) could skip the cycle \(v_{j} \cdots v_{k}\), resulting in the sequence of vertices \(v_{0} v_{1} \cdots v_{j} v_{k+1} \cdots v_{x}\).
    Since \(\PlayerMain\) can ensure with positive probability that the token reaches \(v_{x}\) from \(v_{k}\) in at most \(x - k\) steps, player \(\PlayerMain\) can also ensure with positive probability that the token reaches \(v_{x}\) from \(v_{j}\) in at most \(x - k\) steps as both \(v_{j}\) and \(v_{k}\) are the same vertex \(u\) and at both positions, all \(\GuaranteeThreshold\)-windows have been closed.
    This gives a way for \(\PlayerMain\) to ensure that the token reaches \(v_{x}\) from \(v_{0}\) in less than \(x\) steps, which is a contradiction, since \(x\) was assumed to be the minimum number of steps in which the token reaches \(T\).
  \item If \(u \in \VerticesRandom\), then again, consider the prefix  \(v_{0} v_{1} \cdots v_{j} v_{k+1} \cdots v_{x}\) obtained by skipping the cycle \(v_{j} \cdots v_{k}\).
    With positive probability, the sequence \(v_{k+1}\cdots v_{x}\) is seen immediately after \(v_{j}\).
    Thus, using the same argument as in the previous case, player \(\PlayerMain\) can ensure with positive probability that the token reaches \(v_{x}\) from \(v_{0}\) in less than \(x\) steps, which is a contradiction.
  \end{itemize}
  In both cases we get a contradiction, and thus, we have that \(\PlayerMain\) can reach \(T\) from \(v\) with positive probability in at most \(\abs{\Vertices} \cdot \WindowLength\) steps while also satisfying sure-\(\DirFWMP(\WindowLength, \GuaranteeThreshold)\).
\end{proof}

\begin{theorem}\label{thm:sdfab-correctness}
  \Cref{alg:sure-dirfwmp-almostsure-buchi} computes the winning region for sure-\(\DirFWMP(\WindowLength, \GuaranteeThreshold)\)-almost-sure-\(\BuchiObj(T)\).
\end{theorem}
\begin{proof}
  \((\Rightarrow)\)
  First, we show that from every vertex \(v\) returned by~\Cref{alg:sure-dirfwmp-almostsure-buchi}, player \(\PlayerMain\) has a strategy that is winning for sure-\(\DirFWMP(\WindowLength, \GuaranteeThreshold)\)-almost-sure-\(\BuchiObj(T)\).
  Recall that the algorithm returns all vertices in \(\MDP\) if all vertices in \(\MDP\) are winning for sure-\(\DirFWMP(\WindowLength, \GuaranteeThreshold)\)-positive-\(\ReachObj(T)\), and otherwise the algorithm is called recursively on the subMDP \(\subMDP{\MDP}{S}\) where \(S\) is the winning region for sure-\(\SafeObj(W_{\sdpr})\) in \(\MDP\), and \(W_{\sdpr}\) is the winning region for sure-\(\DirFWMP(\WindowLength, \GuaranteeThreshold)\)-positive-\(\ReachObj(T)\) in \(\MDP\).
  Thus, if the algorithm eventually returns a non-empty set, then in the final recursive call which is called on the subMDP say \(\MDP'\) of \(\MDP\), we have that all vertices of \(\MDP'\) belong to \(W_{\sdpr}' = \SureDirFWMPPosReachAlgo(\MDP', \WindowLength, \GuaranteeThreshold, T)\).
  For every vertex \(v\) in \(W_{\sdpr}'\), since \(v\) is winning for sure-\(\DirFWMP(\WindowLength, \GuaranteeThreshold)\)-positive-\(\ReachObj(T)\) in \(\MDP'\), from \Cref{lem:sdfab-v-l-bound}, there exists a strategy \(\Strategy[\sdpr]^{v}\) of \(\PlayerMain\) that is winning for sure-\(\DirFWMP(\WindowLength, \GuaranteeThreshold)\)-positive-\(\ReachObj_{\le \abs{\Vertices}\cdot \WindowLength}(T)\) from \(v\) in \(\MDP'\).
  Using these strategies, we construct a strategy \(\Strategy[\sdab]\) that is winning for sure-\(\DirFWMP(\WindowLength, \GuaranteeThreshold)\)-almost-sure-\(\BuchiObj(T)\) from all vertices in \(\MDP'\).

  \emph{Description of strategy \(\Strategy[\sdab]\).}
  If the play begins from a vertex \(v\) in \(\MDP'\), then \(\Strategy[\sdab]\) follows \(\Strategy[\sdpr]^{v}\) for the first \(\abs{\Vertices} \cdot \WindowLength\) steps and  further continues to follow \(\Strategy[\sdpr]^{v}\) until all open \(\GuaranteeThreshold\)-windows are closed for the first time after the \(\abs{\Vertices} \cdot \WindowLength\) steps.
  At this point, if the token is at a vertex \(v'\), then the strategy \(\Strategy[\sdab]\) resets to \(\Strategy[\sdpr]^{v'}\), that is, \(\Strategy[\sdab]\) forgets the history and follows \(\Strategy[\sdpr]^{v'}\) from \(v'\) for \(\abs{\Vertices} \cdot \WindowLength\) steps before resetting again when all \(\GuaranteeThreshold\)-windows are closed the next time.

  We claim that this strategy \(\Strategy[\sdab]\) is winning for  sure-\(\DirFWMP(\WindowLength, \GuaranteeThreshold)\)-almost-sure-\(\BuchiObj(T)\) in \(\MDP\) from all vertices in \(\MDP'\), and we prove this claim now.
  Since, for every \(v\) in \(\MDP'\), the strategy \(\Strategy[\sdpr]^{v}\) is winning for sure-\(\DirFWMP(\WindowLength, \GuaranteeThreshold)\), we have that for every outcome of \(\Strategy[\sdpr]^{v}\), every \(\GuaranteeThreshold\)-window closes in at most \(\WindowLength\) steps.
  Moreover, from the inductive property of windows,  for every outcome of \(\Strategy[\sdpr]^{v}\), the number of steps between every two consecutive points of the outcome where all \(\GuaranteeThreshold\)-windows are closed is at most \(\WindowLength\).
  It follows that for every outcome of  \(\Strategy[\sdab]\), it takes at most \(\WindowLength\) steps for all \(\GuaranteeThreshold\)-windows to be closed after the first \(\abs{\Vertices} \cdot \WindowLength\) steps, at which point \(\Strategy[\sdab]\) resets.
  Therefore, we have that the number of steps between consecutive resets of \(\Strategy[\sdab]\) is greater than \(\abs{\Vertices} \cdot \WindowLength\) and is at most \(\abs{\Vertices} \cdot \WindowLength + \WindowLength\).
  Further, since \(\Strategy[\sdpr]^{v}\) is winning for positive-\(\ReachObj_{\le \abs{\Vertices}\cdot \WindowLength}(T)\) in \(\MDP\) for all \(v\) in \(\MDP'\), we have that after every reset of \(\Strategy[\sdab]\), with positive probability (at least \((\ProbabilityFunction_{\min})^{\abs{\Vertices}\cdot \WindowLength}\)) the token reaches \(T\) before the strategy resets again.
  Since the strategy resets at most every \(\abs{\Vertices} \cdot \WindowLength + \WindowLength\) steps and there is at least a fixed positive probability of reaching \(T\) after every reset, the probability that the token reaches \(T\) infinitely often is \(1\).
  Since \(\Strategy[\sdab]\) only resets when all \(\GuaranteeThreshold\)-windows are closed, and each \(\Strategy[\sdpr]^{v}\) is winning for sure-\(\DirFWMP(\WindowLength, \GuaranteeThreshold)\), we have that surely all \(\GuaranteeThreshold\)-windows are always closed in at most \(\WindowLength\) steps.
  Thus, \(\Strategy[\sdab]\) is a winning strategy for sure-\(\DirFWMP(\WindowLength, \GuaranteeThreshold)\)-almost-sure-\(\BuchiObj(T)\) in \(\MDP\) from every vertex \(v\) in \(\MDP'\).

  \((\Leftarrow)\)
  Now, we prove the converse direction, that is, if a vertex \(v\) is not returned by the algorithm, then \(\PlayerMain\) cannot simultaneously satisfy sure-\(\DirFWMP(\WindowLength, \GuaranteeThreshold)\) and almost-sure-\(\BuchiObj(T)\) from \(v\).
  A vertex \(v\) is not included in the region returned by the algorithm if in some recursive call of the algorithm, the vertex \(v\) does not belong to the winning region \(S\) for sure-\(\SafeObj(W_{\sdpr})\) (Line~\ref{alg-line:sdfab-restrict-to-almost-sure-attr}).
  If \(v\) does not belong to \(S\), then for all strategies of \(\PlayerMain\), starting from \(v\), with positive probability the token eventually visits a vertex not in \(W_{\sdpr}\).
  From a vertex not in \(W_{\sdpr}\), for all strategies of \(\PlayerMain\) that are winning for sure-\(\DirFWMP(\WindowLength, \GuaranteeThreshold)\), the token does not reach the target set \(T\) with any positive probability, and thus, we also have that the token does not visit \(T\) infinitely often with any positive probability.
  Hence, from every vertex \(v\) not in \(S\), if \(\PlayerMain\) follows a strategy that is winning for sure-\(\DirFWMP(\WindowLength, \GuaranteeThreshold)\), then almost-sure-\(\BuchiObj(T)\) is not satisfied, and otherwise if \(\PlayerMain\) follows a strategy that is not winning for sure-\(\DirFWMP(\WindowLength, \GuaranteeThreshold)\), then sure-\(\DirFWMP(\WindowLength, \GuaranteeThreshold)\) is not satisfied.
  Thus, we have that every vertex \(v\) that is not returned by the algorithm is not winning for sure-\(\DirFWMP(\WindowLength, \GuaranteeThreshold)\)-almost-sure-\(\BuchiObj(T)\) in \(\MDP\). 
\end{proof}

\paragraph*{\Cref{alg:sure-dirfwmp-positive-reach}: Sure-\(\DirFWMP(\WindowLength, \GuaranteeThreshold)\)-positive-\(\ReachObj(T)\)}

Finally, we describe \Cref{alg:sure-dirfwmp-positive-reach} (\(\SureDirFWMPPosReachAlgo\))
that returns the set of all vertices in an MDP that are winning for sure-\(\DirFWMP(\WindowLength, \GuaranteeThreshold)\)-positive-\(\ReachObj(T)\) for a target set \(T \subseteq V\).
This algorithm is used as a subroutine in \Cref{alg:sure-dirfwmp-almostsure-buchi} (\(\SureDirFWMPASBuchiAlgo\)).
The MDP shown in \Cref{fig:sdab-simultaneous-satisfaction-not-possible} shows that every vertex in an MDP may be winning for sure-\(\DirFWMP(\WindowLength, \GuaranteeThreshold)\) and for positive-\(\ReachObj(T)\) separately, and yet it may be the case that no vertex in the MDP is winning for the combination sure-\(\DirFWMP(\WindowLength, \GuaranteeThreshold)\)-positive-\(\ReachObj(T) \).

\begin{algorithm}[t]
  \caption{\(\SureDirFWMPPosReachAlgo(\MDP, \WindowLength, \GuaranteeThreshold, T)\)}%
  \label{alg:sure-dirfwmp-positive-reach}
  \begin{algorithmic}[1]
    \Require An MDP \(\MDP = ((\Vertices, \Edges), (\VerticesMain, \VerticesRandom), \ProbabilityFunction, \PayoffFunction)\), window length \(\WindowLength\), threshold \(\GuaranteeThreshold\), and target set \(T\)
    \Ensure The winning region for sure-\(\DirFWMP(\WindowLength, \GuaranteeThreshold)\)-positive-\(\ReachObj(T)\) in \(\MDP\)
    \State \(W_{\sd} \assign \SureDirFWMPAlgo(\MDP, \WindowLength, \GuaranteeThreshold)\)\label{alg-line:sdfpr-sure-dirfwmp-winning-region}
    \State \(\MDP \assign \subMDP{\MDP}{W_{\sd}}, \quad T \assign T \intersection
    W_{\sd}\)\label{alg-line:sdfpr-restrict-sure-dirfwmp}
    \State \(R \assign T, \quad E_g,\, E_b \assign \emptyset\)\label{alg-line:sdfpr-initialize-sets}
    \While{\(\{e = (s, t) \in \Edges \suchthat s \in \Vertices \setminus T, \, t \in R, \, \text{and } e \in E \setminus (E_g \union E_b)\} \neq \emptyset\)}\label{alg-line:sdfpr-while-loop}
      \State Choose such an edge \(e\) arbitrarily.
      \If{\Call{\(\IsGoodEdgeAlgo\)}{\(e, E_g\)}}\label{alg-line:sdfpr-isgoodedge-procedure-call}
        \State \(E_g \assign E_g \union \{e\}\)\label{alg-line:sdfpr-add-good-edge}, \quad \(R \assign R \union \{s\}\)\label{alg-line:sdfpr-add-good-start-vertex}, \quad \(E_b \assign \emptyset\)\label{alg-line:sdfpr-reset-bad-edges}
      \Else
        \State \(E_b \assign E_{b} \union \{e\}\)\label{alg-line:sdfpr-add-bad-edge}
      \EndIf
    \EndWhile
    \State \Return \(R\)\label{alg-line:sdfpr-return-winning-region}
    \Statex
    \Procedure{\(\IsGoodEdgeAlgo\)}{\(e = (s, t), E_g\)} \label{alg-line:sdfpr-isgoodedge-procedure-define}
      \State Construct MDP \(\MDP_{e}'\) in the following way:
      \State \quad \(V \assign\) the set of vertices in \(\MDP\) 
      \State \quad \(\widetilde{R} \assign \{\tilde{v} \suchthat v \in R \}\), where \(\tilde{v}\) belongs to \(\PlayerMain\) if and only if \(v\) belongs to \(\PlayerMain\).
      \State \quad The set of vertices of \(\MDP_{e}'\) is \(V \union \widetilde{R} \union \{\hat{s}\}\).
      \Statex \quad The vertex \(\hat{s}\) belongs to \(\PlayerMain\) if and only if \(s\) belongs to \(\PlayerMain\).
      \Statex
      \State \quad \(E \assign \) the set of edges in \(\MDP\)
      \State \quad\(\widetilde{E}_g \assign 
      \{(\tilde{u}, \tilde{v}) \suchthat u,v \in R, \, (u, v) \in E_g)\}\)
      \State \quad \(\widetilde{E}_{\Random} \assign \{(\tilde{u}, v) \suchthat u \in (R \setminus T) \intersection \VerticesRandom, \, v \in V, \, (u, v) \in E \setminus E_g)\} \) 
      \State \quad \(\widehat{E}_{\Random} \assign \textbf{if } s \in \VerticesMain \textbf{ then } \emptyset \textbf{ else } \{(\hat{s}, v) \suchthat (s, v) \in \Edges \ \text{and } v \ne t\} \) 
      \State \quad The set of edges of \(\MDP_{e}'\) is \(E \union \widetilde{E}_g \union \widetilde{E}_{\Random} \union \{(\hat{s}, \tilde{t})\} \union \widehat{E}_{\Random}\).
      \Statex \quad The probability and payoff of each edge is preserved from \(\MDP\).
      \Statex
      \State \(\MDP_{e} \assign \) For all \(v \in T\), identify \(\tilde{v}\) with \(v\)  in \(\MDP_{e}'\)
      \State \Return true iff \(\hat{s}\) is winning for sure-\(\DirFWMP(\WindowLength, \GuaranteeThreshold)\) in \(\MDP_{e}\) \label{alg-line:sdfpr-isgoodedge-return}
    \EndProcedure
  \end{algorithmic}
\end{algorithm}

\subparagraph*{Description of \Cref{alg:sure-dirfwmp-positive-reach}.}
We first compute in Line~\ref{alg-line:sdfpr-sure-dirfwmp-winning-region} the set \(W_{\sd}\) of vertices from which \(\PlayerMain\) can satisfy sure-\(\DirFWMP(\WindowLength, \GuaranteeThreshold)\).
Since \(\DirFWMP(\WindowLength, \GuaranteeThreshold)\) is closed under suffixes, from \Cref{prop:winning-region-sure-objective-closed-under-suffixes-induces-submdp} we have that \(W_{\sd}\) induces a subMDP \(\subMDP{\MDP}{W_{\sd}}\) of \(\MDP\).
We are not interested in vertices in \(V \setminus W_{\sd}\), and thus we restrict the MDP \(\MDP\) to \(\subMDP{\MDP}{W_{\sd}}\) (that is, we refer to the subMDP \(\subMDP{\MDP}{W_{\sd}}\) as \(\MDP\)) where all vertices are winning for sure-\(\DirFWMP(\WindowLength, \GuaranteeThreshold)\), and we refer to \(T \intersection W_{\sd}\) as \(T\) (Line~\ref{alg-line:sdfpr-restrict-sure-dirfwmp}).

We define a set \(R\) of vertices that at the end of the algorithm will consist of precisely all the vertices in \(\MDP\) that are winning for sure-\(\DirFWMP(\WindowLength, \GuaranteeThreshold)\)-positive-\(\ReachObj(T)\) in \(\MDP\).
Since all vertices in \(T\) satisfy sure-\(\DirFWMP(\WindowLength, \GuaranteeThreshold)\) and also trivially satisfy positive-\(\ReachObj(T)\) (as they are already in \(T\)), we initialize \(R\) with \(T\) in Line~\ref{alg-line:sdfpr-initialize-sets}.
As we discover more vertices that are winning for sure-\(\DirFWMP(\WindowLength, \GuaranteeThreshold)\)-positive-\(\ReachObj(T)\), we add them to \(R\).
We always only add vertices to \(R\) and never remove vertices from \(R\). 
  That is, we will show that a vertex \(v\) is winning for sure-\(\DirFWMP(\WindowLength, \GuaranteeThreshold)\)-positive-\(\ReachObj(T)\) if and only if \(v\) is added to \(R\) at some point in the execution of the algorithm.

We also define sets \(E_{g}\) and \(E_{b}\) of edges (\emph{good} and \emph{bad} edges respectively).
We initialize \(E_{g}\) and \(E_{b}\) with empty sets in Line~\ref{alg-line:sdfpr-initialize-sets}, and add edges to these sets over the course of the algorithm in the \textbf{while} loop (Lines~\ref{alg-line:sdfpr-while-loop} to~\ref{alg-line:sdfpr-add-bad-edge}) as we learn more about the MDP \(\MDP\).
The \textbf{while} loop runs as long as there exists an edge \(e = (s, t)\) in \(\MDP\) such that \(s\) is in \(V \setminus T\) and \(t\) is in \(R\), and \(e\) does not belong to \(E_g\) or \(E_b\).
If such an edge exists, then one such edge \(e = (s, t)\) is chosen arbitrarily.
For this edge \(e\), we run the procedure \(\IsGoodEdgeAlgo\) (Lines~\ref{alg-line:sdfpr-isgoodedge-procedure-define} to~\ref{alg-line:sdfpr-isgoodedge-return}) to determine if \(e\) is good or bad with respect to the set \(E_{g}\).
Intuitively, an edge \(e\) is good with respect to \(E_{g}\) if starting from \(s\), with positive probability it is the case that the token takes the edge \(e\) to reach \(R\) and eventually \(T\) using only edges in \(E_{g}\) while simultaneously satisfying sure-\(\DirFWMP(\WindowLength, \GuaranteeThreshold)\).
Formally, an edge \(e\) in \(\MDP\) is \emph{good} with respect to \(E_{g}\) if there exists a winning strategy \(\Strategy[e]\) for sure-\(\DirFWMP(\WindowLength, \GuaranteeThreshold)\)-positive-\(\ReachObj_{\{e\} \union E_{g}}(T)\) starting from edge \(e\) in \(\MDP\).
That is, all outcomes of \(\Strategy[e]\) with initial edge \(e\) satisfy \(\DirFWMP(\WindowLength, \GuaranteeThreshold)\), and with positive probability, an outcome of \(\Strategy[e]\) with initial edge \(e\) reaches \(T\) by subsequently taking only edges in \(E_{g}\).
The edge \(e\) is \emph{bad} with respect to \(E_{g}\) if no such winning strategy \(\Strategy[e]\) exists.

If \(\IsGoodEdgeAlgo\) returns false, then the edge is added to \(E_{b}\) (Line~\ref{alg-line:sdfpr-add-bad-edge}).
Otherwise, if \(\IsGoodEdgeAlgo\) returns true, then the edge \(e\) is added to \(E_{g}\) and the vertex \(s\) is added to the winning region \(R\) for sure-\(\DirFWMP(\WindowLength, \GuaranteeThreshold)\)-positive-\(\ReachObj(T)\) in \(\MDP\) if \(s\) is not already in \(R\) (Line~\ref{alg-line:sdfpr-add-good-start-vertex}).
This is because the winning strategy \(\Strategy[e]\) for sure-\(\DirFWMP(\WindowLength, \GuaranteeThreshold)\)-positive-\(\ReachObj_{\{e\} \union E_{g}}(T)\) from edge \((s, t)\) immediately gives a winning strategy for sure-\(\DirFWMP(\WindowLength, \GuaranteeThreshold)\)-positive-\(\ReachObj(T)\) from \(s\).
Note that given the set \(E_{g}\), we can reconstruct \(R\) as \(T \union \{ u \in V \suchthat (u, v) \in E_{g}\}\).

Finally, we note that if an edge \(e\) is good with respect to a set \(E_{g}\), then \(e\) is also good with respect to all \(E_{g}' \supsetneq E_{g}\), 
  since a strategy \(\Strategy[e]\) that is winning for sure-\(\DirFWMP(\WindowLength, \GuaranteeThreshold)\)-positive-\(\ReachObj_{\{e\} \union E_{g}}(T)\) from \(e\) is also winning for sure-\(\DirFWMP(\WindowLength, \GuaranteeThreshold)\)-positive-\(\ReachObj_{\{e\} \union E_{g}'}(T)\) from \(e\),
since satisfaction of \(\ReachObj_{\{e\} \union E_{g}}(T)\) implies satisfaction of \(\ReachObj_{\{e\} \union E_{g}'}(T)\).
Thus, once we add an edge to \(E_{g}\), it stays in \(E_{g}\) for the rest of the algorithm.
On the other hand, an edge \(e\) that is bad with respect to \(E_{g}\) may not be bad for \(E_{g}' \supsetneq E_{g}\), as non-existence of a winning strategy for sure-\(\DirFWMP(\WindowLength, \GuaranteeThreshold)\)-positive-\(\ReachObj_{\{e\} \union E_{g}}(T)\) does not imply the non-existence of a winning strategy for sure-\(\DirFWMP(\WindowLength, \GuaranteeThreshold)\)-positive-\(\ReachObj_{\{e\} \union E_{g}'}(T)\).
  Thus, when an edge \(e\) is added to \(E_{g}\), the edges present in \(E_{b}\) may or may not be bad with respect to \(E_{g} \union \{e\}\).
Therefore, we reset \(E_{b}\) to the empty set each time an edge is added to \(E_{g}\) (Line~\ref{alg-line:sdfpr-reset-bad-edges}), and we check whether these edges are good or bad with respect to the larger set \(E_{g} \union \{e\}\) in a later iteration of the \textbf{while} loop.

In each iteration of the \textbf{while} loop, either an edge is added to \(E_{b}\), or an edge is added to \(E_{g}\) and the set \(E_{b}\) is reset to empty.
We exit the \textbf{while} loop when every edge starting in \(\Vertices \setminus T\) and ending in \(R\) belongs to either \(E_{g}\) or \(E_{b}\).
Since there are finitely many edges in the MDP \(\MDP\), the \textbf{while} loop is eventually exited.
In fact, the \textbf{while} loop runs for at most \(\abs{E}^2\) iterations, since after at most every \(\abs{E}\) additions to \(E_b\), an edge is added to \(E_g\) and the set \(E_b\) is reset to empty.
After exiting the \textbf{while} loop, we return the set \(R\) of vertices (Line~\ref{alg-line:sdfpr-return-winning-region}).

\begin{example}
  \begin{figure}[t]
    \centering
    \scalebox{0.9}{
      \begin{tikzpicture}[node distance=1.75cm]
        \node[draw, random, initial above] (v1) {\(v_{1}\)};
        \node[state, left of=v1] (v0) {\(v_{0}\)};
        \node[state, right of=v1] (v2) {\(v_{2}\)};
        \node[state, right of=v2] (v3) {\(v_{3}\)};
        \node[state, right of=v3, accepting] (v4) {\(v_{4}\)};

        \draw 
        (v0) edge[loop left] node[auto]{\small \(\EdgeValues{2}{}\)} (v0)
        (v0) edge[bend right=23] node[below left, pos=0.2]{\small \(\EdgeValues{-2}{}\)} (v3)
        
        (v1) edge node[above]{\small \(\EdgeValues{-1}{.5}\)} (v0)
        (v1) edge node[above]{\small \(\EdgeValues{-2}{.5}\)} (v2)
        
        (v2) edge node[above]{\small \(\EdgeValues{2}{}\)} (v3)

        (v3) edge node[above]{\small \(\EdgeValues{1}{}\)} (v4)
        (v2) edge[bend left=35] node[above]{\small \(\EdgeValues{0}{}\)} (v4)
        (v4) edge[loop right] node[auto]{\small \(\EdgeValues{0}{}\)} (v4)      
        ;
      \end{tikzpicture}
    }
    \caption{%
      An MDP where we want to determine the winning region for sure-\(\DirFWMP(\WindowLength = 3, \GuaranteeThreshold = 0)\)-positive-\(\ReachObj(T = \{v_{4}\})\).
      The winning region is \(\{v_{1}, v_{2}, v_{3}, v_{4}\}\).
    }
    \label{fig:sdfpr-algorithm-example-execution}
  \end{figure}
  Consider the MDP in \Cref{fig:sdfpr-algorithm-example-execution}.
  We explain how \Cref{alg:sure-dirfwmp-positive-reach} works on this MDP.
  We consider the winning condition sure-\(\DirFWMP(\WindowLength = 3, \GuaranteeThreshold = 0)\)-positive-\(\ReachObj(T = \{v_{4}\})\).
  Initially, we have $T=R=\{v_4\}$ and $E_g=E_b=\emptyset$.
  Suppose the algorithm first considers the edge $(v_2,v_4)$.
  Note that $v_2 \in V \setminus T$ and $v_4 \in R$.
  The procedure \(\IsGoodEdgeAlgo\) returns that this is a good edge with respect to \(E_{g} = \emptyset\) since starting with this edge $(v_2, v_{4})$ ensures that \(T\) is reached with a positive probability (in fact, surely) and then looping at \(v_{4}\) ensures satisfaction of sure-\(\DirFWMP(\WindowLength, \GuaranteeThreshold)\).
  Hence, we add \(v_{2}\) to $R$ to get $R = \{v_2,v_4\}$ and add the edge $(v_2,v_4)$ to $E_g$.
  Suppose the edge $(v_1,v_2)$ is considered next.
  This edge is not good with respect to \(E_{g} = \{(v_{2}, v_{4})\}\) since starting with \((v_{1}, v_{2})\), it is not possible to satisfy sure-\(\DirFWMP(3, 0)\)-positive-\(\ReachObj_{\{(v_{1}, v_{2}), (v_{2}, v_{4})\}}(T)\).
  Indeed, in order to reach \(T\) starting from the edge \((v_{1}, v_{2})\) and then using only edges in \(E_{g}\), the path \(v_{1}v_{2}v_{4}\) must be taken, but this forces the token to then loop on \(v_{4}\) forever, which results in the \(0\)-window starting at \(v_{1}\) never closing and sure-\(\DirFWMP(3, 0)\) not being satisfied.
  Hence, the edge \((v_{1}, v_{2})\) is added to the set $E_b$.
  Suppose the edge $(v_3,v_4)$ is considered next.
  It is a good edge with respect to \(E_{g}\) since it leads to $T$ from $v_3$ and sure-\(\DirFWMP(3, 0)\) can be simultaneously from $v_3$, and thus, $E_g$ is updated to $\{(v_2,v_4), (v_3,v_4)\}$, the set $E_b$ is reset to $\emptyset$, and $R$ is updated to $\{v_2,v_3,v_4\}$.
  Next consider the edge $(v_2,v_3)$ which is a good edge with respect to $E_g$ and thus $E_g$ is updated to $\{(v_2,v_4),(v_3,v_4),(v_2,v_3)\}$ while $E_b$ remains $\emptyset$ and $R$ is still $\{v_2,v_3,v_4\}$.
  Now suppose $(v_1,v_2)$ is considered again.
  This time we see that this edge is good with respect to $E_g$ since taking the path \(v_{1}v_{2}v_{3}v_{4}\) and then looping on \(v_{4}\) ensures reaching \(T\) through only edges in \(\{(v_{1}, v_{2})\} \union E_{g}\) as well as satisfaction of sure-\(\DirFWMP(3, 0)\).
  Thus, $E_g$ is updated to $\{(v_2,v_4),(v_3,v_4),(v_2,v_3),(v_1,v_2)\}$ and $R$ is updated to $\{v_1,v_2,v_3,v_4\}$.
  Finally, the edge $(v_0,v_3)$ is considered and found to be not good with respect to $E_g$, and is added to \(E_{b}\).
  All edges starting in \(V \setminus T\) and ending in \(R\) belong to either \(E_{g}\) or \(E_{b}\).
  The edges $(v_1,v_0)$ and $(v_0,v_0)$ are not considered for checking if they are good since $v_0 \not\in R$ and the \textbf{while} loop terminates.
  The set $R = \{v_1,v_2,v_3,v_4\}$ returned by \Cref{alg:sure-dirfwmp-positive-reach} is exactly the set of vertices from which sure-\(\DirFWMP(\WindowLength = 3, \GuaranteeThreshold = 0)\)-positive-\(\ReachObj(T = \{v_{4}\})\) can be satisfied.
  \lipicsEnd
\end{example}

\subparagraph*{Description of \(\IsGoodEdgeAlgo\) procedure.}
Given an MDP \(\MDP\), the procedure \(\IsGoodEdgeAlgo(e = (s, t), E_{g})\) 
determines if the edge \(e\) is winning for sure-\(\DirFWMP(\WindowLength, \GuaranteeThreshold)\)-positive-\(\ReachObj_{\{e\} \union E_{g}}(T)\) in \(\MDP\).
In other words, \(\IsGoodEdgeAlgo\) checks if, starting from \(e\), it is possible with positive probability to reach \(T\) through only edges in~\(\{e\} \union E_{g}\) and also simultaneously satisfy sure-\(\DirFWMP(\WindowLength, \GuaranteeThreshold)\).
The procedure augments some vertices and edges to \(\MDP\) to construct an MDP~\(\MDP_{e}\) that has at most thrice as many vertices as \(\MDP\), and returns true if and only if a designated vertex \(\hat{s}\) is winning for sure-\(\DirFWMP(\WindowLength, \GuaranteeThreshold)\) in \(\MDP_{e}\).
This reduces the problem of satisfaction of
sure-\(\DirFWMP(\WindowLength, \GuaranteeThreshold)\)-positive-\(\ReachObj_{E_{g}}(T)\) in \(\MDP\) to that of satisfaction of sure-\(\DirFWMP(\WindowLength, \GuaranteeThreshold)\) in \(\MDP_{e}\), which we know how to solve from~\cite[Algorithm 2]{CDRR15}.
The correctness of this reduction follows from \Cref{lem:sdfpr-isgoodedge-correctness}.
We describe the construction of \(\MDP_{e}\) in two steps.

In the first step, we construct an intermediate MDP \(\MDP_{e}'\).
The set of vertices of the MDP \(\MDP_{e}'\) is \(V \union \widetilde{R} \union \{\hat{s}\}\), where
\begin{itemize}
\item the set \(V\) is the set of all the vertices of \(\MDP\);
\item the set \(\widetilde{R}\) contains a copy \(\tilde{v}\) of every vertex \(v\) in \(R\)\\(we denote by \(\widetilde{T}\) the set \(\{\tilde{v} \in \widetilde{R} \suchthat v \in T\}\), that is, the set of all vertices in \(\widetilde{R}\) in \(\MDP_{e}'\) that are copies of vertices in \(T\), and we have that \(\widetilde{T} \subseteq \widetilde{R}\));
\item the set \(\{\hat{s}\}\) contains a copy of the vertex \(s\).
\end{itemize}
We note that the copy of a vertex that is controlled by \(\PlayerMain\) is also controlled by \(\PlayerMain\), and the copy of a probabilistic vertex is a probabilistic vertex.

The set of edges of the MDP \(\MDP_{e}'\) is \(E \union \widetilde{E}_{g}  \union \widetilde{E}_{\Random} \union \{(\hat{s}, \tilde{t})\}\union \widehat{E}_{\Random}\), where
\begin{itemize}
\item the set \(E\) is the set of all the edges of \(\MDP\);
\item the set \(\widetilde{E}_{g}\) contains a copy \((\tilde{u}, \tilde{v})\) of every edge \((u, v)\) in \(E_{g}\)\\
  (recall that for every edge \((u, v)\) belonging to \(E_{g}\), we have that \(u \in R \setminus T\) and \(v \in R\), and thus, \((\tilde{u}, \tilde{v}) \in (\widetilde{R} \setminus \widetilde{T}) \times \widetilde{R}\));
\item the set \(\widetilde{E}_{\Random}\) contains, for every probabilistic vertex \(u\) in \(R \setminus T\), for every out-edge \((u, v)\) of \(u\) that is not in \(E_{g}\), a copy \((\tilde{u}, v)\) of \((u, v)\);
\item the set \(\{(\hat{s}, \tilde{t})\}\) contains a copy of the edge \(e = (s, t)\),
\item if \(s\) (and hence also \(\hat{s}\)) is a vertex belonging to \(\PlayerMain\), then \(\widehat{E}_{\Random}\) is empty, and
  if \(s\) is a probabilistic vertex, then \(\widehat{E}_{\Random}\) contains a copy \((\hat{s}, v)\) of every out-edge \((s, v)\) of \(s\) such that \(v \ne t\).
\end{itemize}
When we make a copy of an edge, we preserve its payoff and probability.
Briefly, the set \(E\) contains edges that start and end in \(V\).
The set \(\widetilde{E}_g \union \widetilde{E}_{\Random}\) contains edges starting in \(\widetilde{R} \setminus \widetilde{T}\).
The set \(\{(\hat{s}, \tilde{t})\} \union \widehat{E}_{\Random}\) contains edges starting at \(\hat{s}\).
The set \(\widetilde{E}_{\Random} \union \widehat{E}_{\Random}\) contains edges starting in probabilistic vertices in \((\widetilde{R} \setminus \widetilde{T}) \union \{\hat{s}\}\) and ending in \(V\).
We note that if \(s\) belongs to \(R\), then in addition to \(s\) and \(\hat{s}\), the MDP \(\MDP_{e}'\) also contains a copy  \(\tilde{s}\) of \(s\).

In the second step, we obtain \(\MDP_{e}\) from \(\MDP_{e}'\) by identifying, for all \(v \in T\) in \(\MDP_{e}'\), the vertex \(\tilde{v}\) with \(v\).
That is, we merge \(\tilde{v}\) and \(v\) into one vertex, and the set of in-edges (resp., out-edges) of \(v\) in \(\MDP_{e}\) is the union of the in-edges (resp., out-edges) of \(v\) and in-edges (resp, out-edges) of \(\tilde{v}\) in \(\MDP_{e}'\).
We have that \(\widetilde{T}\) and \(T\) become identical in \(\MDP_{e}\), and belong to both \(V\) and \(\widetilde{R}\) in \(\MDP_{e}\).
In other words, vertices in \(\widetilde{T}\) in \(\MDP_{e}'\) are merged into vertices in \(T\), and thus, each vertex in \(\MDP_{e}\) belongs to exactly one of the three sets: \(V\), \(\widetilde{R} \setminus T\), or \(\{\hat{s}\}\).
We show in \Cref{fig:sdfpr-isgoodedge-m} an MDP \(\MDP\), in \Cref{fig:sdfpr-isgoodedge-me-intermediate} the intermediate MDP \(\MDP_{e}'\), and in \Cref{fig:sdfpr-isgoodedge-me} the MDP \(\MDP_{e}\).
We describe in further detail the out-edges of vertices  in \(\MDP_{e}\).
\begin{itemize}
\item If \(s\) is a vertex belonging to \(\PlayerMain\), then \(\hat{s}\) has exactly one out-edge \((\hat{s}, \tilde{t})\), and \(\tilde{t}\) belongs to \(\widetilde{R}\).
  This simulates forcing \(\PlayerMain\) to choose the edge \((s, t)\) from \(s\) in \(\MDP\).
\item   If \(s\) is a probabilistic vertex, then in addition to \((\hat{s}, \tilde{t})\), the vertex \(\hat{s}\) also has the set \(\widehat{E}_{\Random}\) of out-edges, that is, an out-edge \((\hat{s}, v)\) for every out-edge \((s, v)\) of \(s\) such that \(v \neq t\).
  Thus, with positive probability, the token moves along the edge \((\hat{s}, \tilde{t})\). 
  If the token does not choose the edge \((\hat{s}, \tilde{t})\), then the behaviour of the token from \(\hat{s}\) is the same as the behaviour of the token from \(s\) if the edge \((s, t)\) is not chosen.
\item For each vertex \(\tilde{u}\) in \(\widetilde{R} \setminus T\) belonging to \(\PlayerMain\), the vertex only has out-edges from \(\widetilde{E}_g\), that is,  the set of out-edges of \(\tilde{u}\) consists of only copies \((\tilde{u}, \tilde{v})\) of edges \((u, v)\) that belong to \(E_{g}\).
  This simulates forcing \(\PlayerMain\) to choose an out-edge from \(E_{g}\) when the token reaches vertex \(u\) in \(\MDP\).
\item For each probabilistic vertex \(\tilde{u}\) in \(\widetilde{R} \setminus T\), in addition to the copy \((\tilde{u}, \tilde{v})\) of every out-edge \((u,v)\) of \(u\) belonging to \(E_g\), the vertex \(\tilde{u}\) also has the set \(\widetilde{E}_{\Random}\) of out-edges, that is an out-edge \((\tilde{u}, v)\) for every out-edge of \((u, v)\) of \(u\) that does not belong to \(E_g\).
  Since \(u \in R \setminus T\) and every vertex in \(R \setminus T\) has at least one out-edge in \(E_{g}\), it follows that \(\tilde{u}\) has at least one out-edge in \(\widetilde{E}_{g}\).
  Thus, with positive probability, the token chooses an edge in \(\widetilde{E}_{g}\) from \(\tilde{u}\).
  If the token does not choose an edge in \(\widetilde{E}_{g}\) from \(\tilde{u}\), then the behaviour of the token from \(\tilde{u}\) is the same as the behaviour of the token from \(u\) if an edge in \(E_{g}\) is not chosen from \(u\).
\item For all vertices \(v\) in \(V\) in \(\MDP_{e}\), the set of out-edges of \(v\) is the same as the set of out-edges of \(v\) in \(\MDP\).
\end{itemize}
We are interested in plays in \(\MDP_{e}\) that start from \(\hat{s}\).
The following observations give more insight on plays in \(\MDP_{e}\).
\begin{enumerate}
\item The vertex \(\hat{s}\) in \(\MDP_{e}\) has no in-edges.
  Thus, starting from \(\hat{s}\), the token moves to \((\widetilde{R} \setminus T) \union V\) in the next step and never returns to \(\hat{s}\).
  All edges that start from a vertex in \(V\) end in a vertex in \(V\).
  Therefore, if the token reaches \(V\), then the token remains in \(V\) forever.
  Thus, every play in \(\MDP_{e}\) starting from \(\hat{s}\) is of the form \(\{\hat{s}\}\ (\widetilde{R} \setminus T)^{\omega}\) or \(\{\hat{s}\} \ (\widetilde{R} \setminus T)^{*} \ V^{\omega}\).
\item 
  We can project vertices in \(\MDP_{e}\) to vertices in \(\MDP\) by ignoring the hats and the tildes on the vertex names.
  Formally, a projection \(\Proj \colon V \union \widetilde{R} \union \{\hat{s}\} \to V\) is defined for all \(v \in V\) as the identity \(\Proj(v) = v\), for all \(v \in \widetilde{R}\) as \(\Proj(\tilde{v}) = v\), and for \(\hat{s}\) as \(\Proj(\hat{s}) = s\).
  Projection of vertices in \(\MDP_{e}\) naturally extends to projection of prefixes  \(\Proj \colon \PrefixSet{\MDP_{e}} \to \PrefixSet{\MDP}\) and projection of plays \(\Proj \colon \PlaySet{\MDP_{e}} \to \PlaySet{\MDP}\) defined as \(\Proj(v_0 v_1 v_2 \cdots) = \Proj(v_0) \Proj(v_{1}) \Proj(v_{2}) \cdots\).
  It can be seen for all \(\Play_{e} \in \PlaySet{\MDP_{e}}\) that the projection \(\Proj(\Play_{e})\) is a valid play in \(\MDP\), that is, if \(v_{i+1}\) is an out-neighbour of \(v_{i}\) in \(\MDP_{e}\), then \(\Proj(v_{i+1})\) is an out-neighbour of \(\Proj(v_{i})\) in \(\MDP\), and further, the sequence of edge payoffs in \(\Play_{e}\) is the same as the sequence of edge payoffs in \(\Proj(\Play_{e})\).
\item For every strategy \(\Strategy[e]\) in \(\MDP_{e}\), we can construct a strategy in \(\Strategy\) in \(\MDP\) (using additional memory as described below) such that the behaviour of the two strategies is the same, that is, the same edge payoffs are seen with the same probabilities in outcomes of \(\Strategy[e]\) in \(\MDP_{e}\) and \(\Strategy\) in \(\MDP\).
  Formally, we have that for all prefixes \(\Prefix\) in \(\MDP_{e}\) ending in a vertex belonging to \(\PlayerMain\), we have \(\Proj(\Strategy[e](\Prefix)) = \Strategy(\Proj(\Prefix))\).
  Given a prefix \(\Prefix\) in \(\MDP_{e}\), the strategy \(\Strategy\) mimics on the projected prefix \(\Proj(\Prefix)\) the projection of the behaviour of \(\Strategy[e]\) on \(\Prefix\).
  Since \(\Strategy[e]\) may have different behaviour depending on whether the token is on \(\tilde{u}\) or \(u\) in \(\MDP_{e}\) for a vertex \(u \in R \setminus T\), or different behaviour for the token at \(s\), \(\tilde{s}\), and \(\hat{s}\) for the vertex \(s\), the strategy \(\Strategy\) requires at most three times as much memory as \(\Strategy[e]\) to distinguish among these cases.
\end{enumerate}

\begin{figure}[t]
  \centering
  \begin{tikzpicture}[node distance=1.75cm,
    inner sep=10pt]

    \tikzset{
      every state/.append style={minimum size=2pt, inner sep=2.25pt},
      random/.append style={diamond, thick, fill=gray!20, inner sep=1.75pt},
      notineg/.append style={thick, dotted},
      inediamond/.append style={dashed},
    }

    \draw[rounded corners=4pt, fill=green, fill opacity=0.1]
    (0,0) rectangle ++(6,3);
    
    \draw[fill=orange, fill opacity=0.1]
    (6,2) {[rounded corners=4pt] -- ++(-5,0) -- ++(0,-2) -- ++(+5,0)} -- cycle;

    \draw[fill=blue, fill opacity=0.1]
    (6,1.1) {[rounded corners=4pt] --
      ++(-1.5,0)  -- 
      ++(0,-1.1) --
      ++(+1.5,0)} --
    cycle;

    \node (s) at (4.5,2.5)   [random, draw] {\(s\)};

    \node (w) at (5.5,0.35)   [state] {\(w\)};
    
    \node (t) at (4,1.5)   [state] {\(t\)};

    \node (u) at (2,1)   [state] {\(u\)};

    \node (v) at (3.5,0.5)   [random, draw] {\(v\)};

    \node (x) at (0.5,1)   [state] {\(x\)};
    \node (y) at (5.5, 2.5)   [state] {\(y\)};

    \node (V) at (0.3,2.7)    {\(V\)};
    
    \node (R) at (1.25, 1.7)    {\(R\)};
    
    \node (T) at (4.75, 0.8)    {\(T\)};
    
    \draw
    (s) edge[double] (t)
    (s) edge[notineg] (y)
    (s) edge[loop left, notineg] (s)
    
    (x) edge[loop below, notineg] (x)
    
    (t) edge[bend left] (w)
    (t) edge (v)
    
    (u) edge[loop above, notineg] (u)
    (u) edge[notineg] (x)
    (u) edge[bend left] (v)

    (w) edge[bend right=10, notineg] (v)
    
    (v) edge[bend left, notineg] (u)
    (v) edge[bend right=10] (w)
    (v) edge[out=90, in=90, looseness=1.25, notineg] (x)
    ;

  \end{tikzpicture}
  \caption{%
    An MDP \(\MDP\) with \(T = \{w\}\), \(R = \{t, u, v, w\}\), and \(E_{g} = \{(u,v), (t, v), (t, w), (v, w)\}\). We want to determine if the edge \((s, t)\) is good with respect to \(E_{g}\). Edges in \(E_{g}\) are solid and edges not in \(E_{g}\) are dotted.
  }
  \label{fig:sdfpr-isgoodedge-m}
\end{figure}

\begin{figure}[t]
  \centering
  \begin{tikzpicture}[node distance=1.75cm,
    inner sep=10pt]

    \tikzset{
      every state/.append style={minimum size=2pt, inner sep=2.25pt},
      random/.append style={diamond, thick, fill=gray!20, inner sep=1.75pt},
      notineg/.append style={thick, dotted},
      inediamond/.append style={dashed},
    }

    \draw[rounded corners=4pt, fill=green, fill opacity=0.1]
    (0,0) rectangle ++(6,3);
    
    \draw[fill=orange, fill opacity=0.1]
    (6,2) {[rounded corners=4pt] --
      ++(-5,0)  -- 
      ++(0,-4) --
      ++(+5,0)} --
    cycle;

    \draw[fill=blue, fill opacity=0.1]
    (6,1.1) {[rounded corners=4pt] --
      ++(-1.5,0)  -- 
      ++(0,-2.2)} --
    ++(+1.5,0) --
    cycle;

    \node (s) at (4.5,2.5)   [random, draw] {\(s\)};
    \node (hats) at (4.5,-2.5)   [random, draw] {\(\hat{s}\)};

    \node (w) at (5.5,0.35)   [state] {\(w\)};
    \node (tildew) at (5.5,-0.35)   [state] {\(\tilde{w}\)};
    
    \node (t) at (4,1.5)   [state] {\(t\)};
    \node (tildet) at (4,-1.5)   [state] {\(\tilde{t}\)};

    \node (u) at (2,1)   [state] {\(u\)};
    \node (tildeu) at (2,-1)   [state] {\(\tilde{u}\)};  

    \node (v) at (3.5,0.5)   [random, draw] {\(v\)};
    \node (tildev) at (3.5,-0.5)   [random, draw] {\(\tilde{v}\)};

    \node (x) at (0.5,1)   [state] {\(x\)};
    \node (y) at (5.5, 2.5)   [state] {\(y\)};

    \node (V) at (0.3,2.7)    {\(V\)};
    
    \node (R) at (1.25, 1.7)    {\(R\)};
    \node (tildeR) at (1.25,-1.7)    {\(\widetilde{R}\)};  
    
    \node (T) at (4.75, 0.8)    {\(T\)};
    \node (tildeT) at (4.75,-0.8)    {\(\widetilde{T}\)};  
    
    \draw
    (s) edge[double] (t)
    (s) edge[notineg] (y)
    (s) edge[loop left, notineg] (s)
    
    (x) edge[loop below, notineg] (x)
    
    (t) edge[bend left] (w)
    (t) edge (v)
    (tildet) edge[bend right] (tildew)
    (tildet) edge (tildev)
    
    (hats) edge[double] (tildet)
    (hats) edge[bend right=75, inediamond] (s)
    (hats) edge[bend right=90, inediamond] (y)

    (u) edge[loop above, notineg] (u)
    (u) edge[notineg] (x)
    (u) edge[bend left] (v)

    (w) edge[bend right=10, notineg] (v)
    (tildeu) edge[bend right] (tildev)

    (v) edge[bend left, notineg] (u)
    (v) edge[bend right=10] (w)
    (v) edge[out=90, in=90, looseness=1.25, notineg] (x)

    (tildev) edge[bend left, inediamond] (u)
    (tildev) edge[bend left=15] (tildew)
    (tildev) edge[bend left=20, inediamond] (x)
    ;
  \end{tikzpicture}
  \caption{%
    Construction of \(\MDP_{e}'\).
    The vertices and edges present in \(\MDP\) remain unchanged.
    A copy \(\{\tilde{t}, \tilde{u}, \tilde{v}, \tilde{w}\}\) of every vertex in \(R\) is added. 
    A copy \(\hat{s}\) of \(s\) is also added.
    Edges in \(\widetilde{E}_{\Random} \union \widehat{E}_{\Random}\) are dashed.
  }
  \label{fig:sdfpr-isgoodedge-me-intermediate}
\end{figure}

\begin{figure}[t]
  \centering
  \begin{tikzpicture}[node distance=1.75cm,
    inner sep=10pt]

    \tikzset{
      every state/.append style={minimum size=2pt, inner sep=2.25pt},
      random/.append style={diamond, thick, fill=gray!20, inner sep=1.75pt},
      notineg/.append style={thick, dotted},
      inediamond/.append style={dashed},
    }

    \draw[rounded corners=4pt, fill=green, fill opacity=0.1]
    (0,0) rectangle ++(6,3);
    
    \draw[fill=orange, fill opacity=0.1]
    (6,2) {[rounded corners=4pt] --
      ++(-5,0)  -- 
      ++(0,-4) --
      ++(+5,0)} --
    cycle;

    \draw[fill=blue, fill opacity=0.1]
    (6,1.1) {[rounded corners=4pt] --
      ++(-1.5,0)  -- 
      ++(0,-2.2)} --
    ++(+1.5,0) --
    cycle;

    \node (s) at (4.5,2.5)   [random, draw] {\(s\)};
    \node (hats) at (4.5,-2.5)   [random, draw] {\(\hat{s}\)};

    \node (w) at (5.5,0)   [state] {\(w\)};
    
    \node (t) at (4,1.5)   [state] {\(t\)};
    \node (tildet) at (4,-1.5)   [state] {\(\tilde{t}\)};

    \node (u) at (2,1)   [state] {\(u\)};
    \node (tildeu) at (2,-1)   [state] {\(\tilde{u}\)};  

    \node (v) at (3.5,0.5)   [random, draw] {\(v\)};
    \node (tildev) at (3.5,-0.5)   [random, draw] {\(\tilde{v}\)};

    \node (x) at (0.5,1)   [state] {\(x\)};
    \node (y) at (5.5, 2.5)   [state] {\(y\)};

    \node (V) at (0.3,2.7)    {\(V\)};
    
    \node (R) at (1.25, 1.7)    {\(R\)};
    \node (tildeR) at (1.25,-1.7)    {\(\widetilde{R}\)};  
    
    \node (T) at (4.75, 0.8)    {\(T\)};
    
    \draw
    (s) edge[double] (t)
    (s) edge[notineg] (y)
    (s) edge[loop left, notineg] (s)
    
    (x) edge[loop below, notineg] (x)
    
    (t) edge[bend left] (w)
    (t) edge (v)
    (tildet) edge[bend right] (w)
    (tildet) edge (tildev)
    
    (hats) edge[double] (tildet)
    (hats) edge[bend right=75, inediamond] (s)
    (hats) edge[bend right=90, inediamond] (y)

    (u) edge[loop above, notineg] (u)
    (u) edge[notineg] (x)
    (u) edge[bend left] (v)

    (w) edge[bend right=10, notineg] (v)
    (tildeu) edge[bend right] (tildev)

    (v) edge[bend left, notineg] (u)
    (v) edge[bend right=10] (w)
    (v) edge[out=90, in=90, looseness=1.25, notineg] (x)

    (tildev) edge[bend left, inediamond] (u)
    (tildev) edge[bend right=15] (w)
    (tildev) edge[bend left=20, inediamond] (x)
    ;

  \end{tikzpicture}
  \caption{%
    Obtaining \(\MDP_{e}\) from \(\MDP_{e}'\) by merging vertex \(\tilde{w}\) in \(\widetilde{T}\) with \(w\) in \(T\).
  }
  \label{fig:sdfpr-isgoodedge-me}
\end{figure}

\subparagraph*{Proof of correctness of \Cref{alg:sure-dirfwmp-positive-reach}.}
We prove the correctness of \Cref{alg:sure-dirfwmp-positive-reach} in \Cref{thm:sdfpr-correctness}.
Before that, we prove the correctness of \(\IsGoodEdgeAlgo\) in \Cref{lem:sdfpr-isgoodedge-correctness}, and prove some intermediate results in \Cref{prop:sdfpr-isgoodvertex,prop:sdfpr-all-inedges-r-bad,lem:sdfpr-bad-edge-is-losing}.
\begin{lemma}\label{lem:sdfpr-isgoodedge-correctness}
  Given \(E_{g}\) and an edge \(e = (s, t)\) with \(s \in V \setminus T\) and \(t \in R\), the procedure \(\IsGoodEdgeAlgo(e, E_{g})\) correctly computes whether \(e\) is winning for sure-\(\DirFWMP(\WindowLength, \GuaranteeThreshold)\)-positive-\(\ReachObj_{\{e\} \union E_{g}}(T)\) in \(\MDP\).
\end{lemma}
\begin{proof}
  Recall that \(E_{g}\) is empty at the start of the algorithm and that edges are added to \(E_{g}\) over the course of the algorithm.
  We prove the statement of the lemma by induction on the size of \(E_{g}\).

\noindent  \textit{Base case:} 
  For the base case, we suppose that \(E_{g} = \emptyset\).
  The procedure \(\IsGoodEdgeAlgo\) is called on an edge \(e = (s, t)\) with \(s \in V \setminus T\) and \(t \in R\).
  Since \(E_{g}\) is empty, we have that \(R = T\), and thus \(t \in T\).
  We show that \(\IsGoodEdgeAlgo(e, \emptyset)\) returns true if and only if the edge \((s, t)\) is winning for sure-\(\DirFWMP(\WindowLength, \GuaranteeThreshold)\)-positive-\(\ReachObj_{\{e\}}(T)\) in \(\MDP\).

  \((\Rightarrow)\)
  Suppose \(\IsGoodEdgeAlgo(e, \emptyset)\) returns true, that is, the vertex \(\hat{s}\) is winning for sure-\(\DirFWMP(\WindowLength, \GuaranteeThreshold)\) in \(\MDP_{e}\).
If \(s\) is a vertex belonging to \(\PlayerMain\), then \(\hat{s}\) has only one out-edge \((\hat{s}, \tilde{t})\) where \(\tilde{t}\) is the same as \(t\) since \(t \in T\).
    Since \(\hat{s}\) is winning for sure-\(\DirFWMP(\WindowLength, \GuaranteeThreshold)\), it follows that there exists a strategy \(\Strategy[e]\) of \(\PlayerMain\) that moves the token along the edge \((\hat{s}, t)\) and then continues the play from \(t\) in a way that satisfies sure-\(\DirFWMP(\WindowLength, \GuaranteeThreshold)\).
    On the other hand, if \(s\) is a probabilistic vertex, then \(\hat{s}\) has an out-edge \((\hat{s}, t)\) and may have other out-edges.
    Since \(\hat{s}\) is winning for sure-\(\DirFWMP(\WindowLength, \GuaranteeThreshold)\), we have that no matter what edge is chosen from \(\hat{s}\), player \(\PlayerMain\) has a strategy that continues the play such that it satisfies sure-\(\DirFWMP(\WindowLength, \GuaranteeThreshold)\).
    In particular, there exists a strategy \(\Strategy[e]\) of \(\PlayerMain\)  that continues the play from \((\hat{s}, t)\) such that it satisfies sure-\(\DirFWMP(\WindowLength, \GuaranteeThreshold)\).
In both cases, the strategy \(\Strategy\)  that mimics the behaviour of \(\Strategy[e]\) (that is \(\Strategy(s \cdot t \cdot  \Prefix) = \Strategy[e](\hat{s} \cdot t \cdot \Prefix)\) for all prefixes \(\Prefix\)) is winning for sure-\(\DirFWMP(\WindowLength, \GuaranteeThreshold)\) from \(e\) in \(\MDP\).
    Since the token reaches \(T\) immediately after the initial edge \(e\), we have that \(\Strategy\) is also winning for sure-\(\DirFWMP(\WindowLength, \GuaranteeThreshold)\)-positive-\(\ReachObj_{\{e\}}(T)\) from \(e\) in \(\MDP\).

  \((\Leftarrow)\)
  Conversely, suppose that the edge \((s, t)\) is winning for sure-\(\DirFWMP(\WindowLength, \GuaranteeThreshold)\)-positive-\(\ReachObj_{\{e\}}(T)\) in \(\MDP\).
  Then, there exists a strategy \(\Strategy\) such that every outcome of \(\Strategy\) in \(\MDP\) starting from \(e\) satisfies \(\DirFWMP(\WindowLength, \GuaranteeThreshold)\).
  We use \(\Strategy\) to construct a strategy \(\Strategy[e]\) that is winning for sure-\(\DirFWMP(\WindowLength, \GuaranteeThreshold)\) from \(\hat{s}\) in \(\MDP_{e}\).
  If \(s\) belongs to \(\PlayerMain\), then so does \(\hat{s}\), and let \(\Strategy[e]\) be defined as \(\Strategy[e](\hat{s}) = t\), and \(\Strategy[e](\hat{s} \cdot t \cdot \Prefix) = \Strategy(s \cdot t \cdot \Prefix)\) for all prefixes \(\Prefix\).
  On the other hand, if \(s\) is a probabilistic vertex, in which case so is \(\hat{s}\), then let \(\Strategy[e]\) be defined as \(\Strategy[e](\hat{s}\cdot \Prefix) = \Strategy(s \cdot \Prefix)\) for all \(\Prefix\).
  In both cases, an outcome of \(\Strategy\) visits the same edge payoffs with the same probability as an outcome of \(\Strategy[e]\).
  Since \(\Strategy\) is winning for sure-\(\DirFWMP(\WindowLength, \GuaranteeThreshold)\) from \(s\), we also have that \(\Strategy[e]\) is winning for sure-\(\DirFWMP(\WindowLength, \GuaranteeThreshold)\) from \(\hat{s}\).

\noindent  \textit{Inductive step:}
  For the inductive step, suppose \(\abs{E_{g}} = k > 0\), and we label the edges of \(E_{g}\) as \(e_{1}, e_{2}, \ldots, e_{k}\) in the order they were added to \(E_{g}\).
  We want to verify if an edge \(e = (s, t)\) for \(s \in V \setminus T\) and \(t \in R\) is good with respect to \(E_{g}\).

  \((\Rightarrow)\)
  Suppose the procedure \(\IsGoodEdgeAlgo(e, E_{g})\) returns true, that is, the vertex \(\hat{s}\) is winning for sure-\(\DirFWMP(\WindowLength, \GuaranteeThreshold)\) in \(\MDP_{e}\), and thus, there exists a strategy \(\Strategy[e]\) of \(\PlayerMain\) that is winning for sure-\(\DirFWMP(\WindowLength, \GuaranteeThreshold)\) in \(\MDP_{e}\) from \(\hat{s}\).
  We construct a strategy from the edge \(e\) that is winning for sure-\(\DirFWMP(\WindowLength, \GuaranteeThreshold)\)-positive-\(\ReachObj_{\{e\} \union E_{g}}(T)\) in \(\MDP\).

  Suppose \(\PlayerMain\) follows the strategy \(\Strategy[e]\) from \(\hat{s}\) until the first time in the play where all \(\GuaranteeThreshold\)-windows are closed.
  For all outcomes of \(\Strategy[e]\), this happens in at most \(\WindowLength\) steps from the start of the play since \(\Strategy[e]\) is winning for sure-\(\DirFWMP(\WindowLength, \GuaranteeThreshold)\).
  Moreover, since every vertex in \(\widetilde{R} \setminus T\) belonging to \(\PlayerMain\) contains has all out-edges in \(\widetilde{E}_{g}\) and every probabilistic vertex in \(\widetilde{R} \setminus T\) has at least one out-edge in \(\widetilde{E}_{g}\), starting from \(\hat{s}\) we have that at least one of the following hold:
  \begin{enumerate*}[label=(\emph{\roman*})]
  \item with positive probability the token reaches \(T\) within at most \(\WindowLength\) steps while taking only edges in \(\{(\hat{s}, \tilde{t})\} \union \widetilde{E}_{g}\), or
  \item the token remains in \(\widetilde{R} \setminus T\) until all \(\GuaranteeThreshold\)-windows are closed while taking only edges in \(\{(\hat{s}, \tilde{t})\} \union \widetilde{E}_{g}\).
  \end{enumerate*}
  
  Suppose the token reaches \(T\) while visiting only edges in \(\{(\hat{s}, \tilde{t})\} \union \widetilde{E}_{g}\), that is, we have that \(\ReachObj_{\{(\hat{s}, \tilde{t})\} \union \widetilde{E}_{g}}(T)\) is satisfied and the token has entered \(V\).
  From Observation~1, the token now behaves the same as it would in \(\MDP\).
  If \(\PlayerMain\) continues to follow the strategy \(\Strategy[e]\) for the rest of the play, then since \(\Strategy[e]\) is winning for sure-\(\DirFWMP(\WindowLength, \GuaranteeThreshold)\), we have that all \(\GuaranteeThreshold\)-windows are closed in at most \(\WindowLength\) steps.
  Thus, in this case, sure-\(\DirFWMP(\WindowLength, \GuaranteeThreshold)\)-positive-\(\ReachObj_{\{(\hat{s}, \tilde{t})\} \union \widetilde{E}_{g}}(T)\) is satisfied from \(\hat{s}\) in \(\MDP_{e}\).
  If \(\PlayerMain\) follows the projection \(\Strategy\) (as described in Observation~3) in \(\MDP\) of the strategy \(\Strategy[e]\), then we see that \(\Strategy\) is winning for sure-\(\DirFWMP(\WindowLength, \GuaranteeThreshold)\)-positive-\(\ReachObj_{\{e\} \union E_{g}}(T)\) in \(\MDP\) from \((s, t)\).

  Otherwise, suppose that the token remains in \(\widetilde{R} \setminus T\) until all \(\GuaranteeThreshold\)-windows are closed for the first time while choosing only edges in \(\{(\hat{s}, \tilde{t})\} \union \widetilde{E}_{g}\).
  Suppose the token is on a vertex \(\tilde{s}_{1} \in \widetilde{R} \setminus T\) for some \(s_{1} \in R \setminus T\) when all \(\GuaranteeThreshold\)-windows are closed for the first time.
  Since \(s_{1}\) belongs to \(R \setminus T\), we have that \(s_{1}\) has at least one out-edge \(e_{i}\) in \(E_{g}\) for \(1 \le i \le k\), and let this edge be \(e_{i} = (s_1, t_1)\) for some \(t_1\) in \(\widetilde{R} \setminus T\).
  The edge \(e_i\) belongs to \(E_g\) because the procedure \(\IsGoodEdgeAlgo(e_i, E_g')\) returned true for \(E_{g}' = \{e_1, e_2, \ldots, e_{i - 1}\}\).
  By the induction hypothesis, we have that the edge \((s_1, t_1)\) is winning for sure-\(\DirFWMP(\WindowLength, \GuaranteeThreshold)\)-positive-\(\ReachObj_{E_{g}'}(T)\) in \(\MDP\).
  Using Observation~3, we have that  starting from \((s, t)\), player \(\PlayerMain\) has a strategy that reaches \(s_{1}\) with positive probability while only visiting edges in \(\{e\} \union E_{g}\) while also surely closing all \(\GuaranteeThreshold\)-windows in at most \(\WindowLength\) steps.
  Then, from the induction hypothesis, from the edge \((s_{1}, t_{1})\), player \(\PlayerMain\) has a strategy to satisfy sure-\(\DirFWMP(\WindowLength, \GuaranteeThreshold)\)-positive-\(\ReachObj_{\{(s_{1}, t_{1})\} \union E_{g}'}(T)\) in \(\MDP\).
  Thus, the edge \(e\) is winning for sure-\(\DirFWMP(\WindowLength, \GuaranteeThreshold)\)-positive-\(\ReachObj_{\{e\} \union E_{g}}(T)\) in \(\MDP\).

  \((\Leftarrow)\)
  Suppose \(\IsGoodEdgeAlgo(e, E_{g})\) returns false, that is, the vertex \(\hat{s}\) is not winning for sure-\(\DirFWMP(\WindowLength, \GuaranteeThreshold)\) in \(\MDP_{e}\).
  Since the \(\widetilde{R} \union \{\hat{s}\}\) part of \(\MDP_{e}\) simulates forcing \(\PlayerMain\) to only choose edges in \(\{e\} \union E_{g}\), we get that starting from \(s\) in \(\MDP\), if \(\PlayerMain\) only chooses edges in \(\{e\} \union E_{g}\) from vertices belonging to \(\PlayerMain\), then \(\PlayerMain\) cannot satisfy sure-\(\DirFWMP(\WindowLength, \GuaranteeThreshold)\),
  Hence, starting from \(s\) in \(\MDP\), it is not possible for \(\PlayerMain\) to satisfy sure-\(\DirFWMP(\WindowLength, \GuaranteeThreshold)\)-positive-\(\ReachObj_{\{e\} \union E_{g}}(T)\).
  Thus, the edge \((s, t)\) is not winning for sure-\(\DirFWMP(\WindowLength, \GuaranteeThreshold)\)-positive-\(\ReachObj_{\{e\} \union E_{g}}(T)\) in \(\MDP\).
\end{proof}

\begin{proposition}\label{prop:sdfpr-isgoodvertex}
  At the end of \Cref{alg:sure-dirfwmp-positive-reach}, if \(v\) is a vertex with an out-edge \((v, u)\) that is winning for sure-\(\DirFWMP(\WindowLength, \GuaranteeThreshold)\)-positive-\(\ReachObj_{E_{g}}(T)\) in \(\MDP\), then \(v\) is winning for sure-\(\DirFWMP(\WindowLength, \GuaranteeThreshold)\)-positive-\(\ReachObj_{E_{g}}(T)\) in \(\MDP\).
\end{proposition}
\begin{proof}
  Since \(v\) has an out-edge \((v, u)\) that is winning for  sure-\(\DirFWMP(\WindowLength, \GuaranteeThreshold)\)-positive-\(\ReachObj_{E_{g}}(T)\) in \(\MDP\), we have that after taking the edge \((v, u)\), there exists a strategy \(\Strategy\) continuing the play from \(u\) that is winning for sure-\(\DirFWMP(\WindowLength, \GuaranteeThreshold)\)-positive-\(\ReachObj_{E_{g}}(T)\) in \(\MDP\).
  There are two cases depending on whether the vertex \(v\) belongs to \(\VerticesMain\) or \(\VerticesRandom\).
  \begin{itemize}
  \item If \(v \in \VerticesMain\), then \(\PlayerMain\) can move the token to \(u\), and then follow the strategy \(\Strategy\) from \(u\) to satisfy sure-\(\DirFWMP(\WindowLength, \GuaranteeThreshold)\)-positive-\(\ReachObj_{E_{g}}(T)\) in \(\MDP\).
  \item If \(v \in \VerticesRandom\), then the token moves from \(v\) to \(u\) with positive probability, from which \(\PlayerMain\) can follow the strategy \(\Strategy\) from \(u\) to satisfy sure-\(\DirFWMP(\WindowLength, \GuaranteeThreshold)\)-positive-\(\ReachObj_{E_{g}}(T)\) in \(\MDP\).
    If token moves from \(v\) to a vertex other than \(u\), then since \(v\) is winning for sure-\(\DirFWMP(\WindowLength, \GuaranteeThreshold)\), we have that sure-\(\DirFWMP(\WindowLength, \GuaranteeThreshold)\) is still satisfied.
  \end{itemize}
  In both cases, we have that the vertex \(v\) is winning for sure-\(\DirFWMP(\WindowLength, \GuaranteeThreshold)\)-positive-\(\ReachObj_{E_{g}}(T)\) in \(\MDP\).
\end{proof}
\begin{proposition}\label{prop:sdfpr-all-inedges-r-bad}
  At the end of \Cref{alg:sure-dirfwmp-positive-reach}, for all edges \((u, v)\) such that \(u \not\in R\) and \(v \in R\), the edge \((u, v)\) is not winning for sure-\(\DirFWMP(\WindowLength, \GuaranteeThreshold)\)-positive-\(\ReachObj_{E_{g}}(T)\) in \(\MDP\).
\end{proposition}
\begin{proof}
  Suppose towards a contradiction that there exists an edge \((u, v)\) such that \(u \not\in R\) and \(v \in R\) and \((u, v)\) is winning for sure-\(\DirFWMP(\WindowLength, \GuaranteeThreshold)\)-positive-\(\ReachObj_{E_{g}}(T)\) in \(\MDP\).
  Then, from Line~\ref{alg-line:sdfpr-add-good-start-vertex} of the algorithm, we have that \(u\) must be included in \(R\), which is a contradiction.
\end{proof}

\begin{lemma}\label{lem:sdfpr-bad-edge-is-losing}
  At the end of \Cref{alg:sure-dirfwmp-positive-reach}, for all vertices \(v\) in MDP \(\MDP\), the vertex \(v\) is winning for sure-\(\DirFWMP(\WindowLength, \GuaranteeThreshold)\)-positive-\(\ReachObj_{E_{g}}(T)\) in \(\MDP\) if and only if \(v\) is winning for sure-\(\DirFWMP(\WindowLength, \GuaranteeThreshold)\)-positive-\(\ReachObj(T)\) in \(\MDP\).
\end{lemma}
\begin{proof}
  Since satisfaction of \(\ReachObj_{E_{g}}(T)\) implies satisfaction of \(\ReachObj(T)\), the forward direction is immediate.
  We prove the converse, that is, if  \(v\) is not winning for sure-\(\DirFWMP(\WindowLength, \GuaranteeThreshold)\)-positive-\(\ReachObj_{E_{g}}(T)\), then \(v\) is not winning for  sure-\(\DirFWMP(\WindowLength, \GuaranteeThreshold)\)-positive-\(\ReachObj(T)\).
  
  Suppose towards a contradiction that starting from \(v\), there exists a strategy \(\Strategy\) that is winning for sure-\(\DirFWMP(\WindowLength, \GuaranteeThreshold)\)-positive-\(\ReachObj(T)\) in \(\MDP\).
  Consider an outcome \(\Play\) of this strategy that reaches \(T\) and also satisfies \(\DirFWMP(\WindowLength,\GuaranteeThreshold)\).
  Since \(v\) is not winning for sure-\(\DirFWMP(\WindowLength, \GuaranteeThreshold)\)-positive-\(\ReachObj_{E_{g}}(T)\), we have that there exists at least one edge in this outcome before reaching \(T\) that does not belong to \(E_{g}\).
  Consider the last edge in this outcome before reaching \(T\) that does not belong to \(E_{g}\), say \((x,y)\).
  Every edge in the outcome after \((x, y)\) and before the token reaches \(T\) belongs to \(E_{g}\).
  From the proof of \Cref{lem:sdfpr-isgoodedge-correctness}, we have that since \((x, y)\) does not belong to \(E_{g}\), there does not exist a strategy that is winning for sure-\(\DirFWMP(\WindowLength, \GuaranteeThreshold)\)-positive-\(\ReachObj_{\{(x,y)\} \union E_{g}}(T)\) from \(x\).
  This is, however, a contradiction, since the path from \(x\) to \(T\) in \(\Play\) that does not visit any edge not in \(\{(x,y)\} \union E_{g}\) occurs with positive probability.
\end{proof}

\begin{theorem}\label{thm:sdfpr-correctness}
  \Cref{alg:sure-dirfwmp-positive-reach} computes the winning region for \(\PlayerMain\) for sure-\(\DirFWMP(\WindowLength, \GuaranteeThreshold)\)-positive-\(\ReachObj(T)\) in \(\MDP\).
\end{theorem}
\begin{proof}
  Recall that if a vertex \(v\) does not belong to \(W_{\sd}\), then \(v\) is not winning for   sure-\(\DirFWMP(\WindowLength, \GuaranteeThreshold)\) in \(\MDP\), and therefore also not winning for sure-\(\DirFWMP(\WindowLength, \GuaranteeThreshold)\)-positive-\(\ReachObj(T)\) in \(\MDP\).
  The algorithm correctly removes the vertices not in \(W_{\sd}\) in Line~\ref{alg-line:sdfpr-restrict-sure-dirfwmp}.
  In the rest of the algorithm and the rest of this proof, we work with the subMDP \(\subMDP{\MDP}{W_{\sd}}\) and refer to this subMDP as \(\MDP\), so we have that all vertices in \(\MDP\) are winning for sure-\(\DirFWMP(\WindowLength, \GuaranteeThreshold)\).

  Recall that the algorithm returns the set \(R \subseteq V\).
  We show for all \(v\) in \(\MDP\) that \(v\) belongs to the set \(R\) if and only if \(v\) is winning for sure-\(\DirFWMP(\WindowLength, \GuaranteeThreshold)\)-positive-\(\ReachObj(T)\) in \(\MDP\).
  We consider two cases separately.
  First, we show this result holds for all \(v \in T\), and then, for all \(v \in V \setminus T\).

  Let \(v \in T\).
  We have that every vertex in \(\MDP\) is winning for sure-\(\DirFWMP(\WindowLength, \GuaranteeThreshold)\).
  Moreover, every play \(\Play\) starting from \(v\) satisfies \(\ReachObj(T)\) since \(v\) is in \(T\), and thus, a winning strategy for sure-\(\DirFWMP(\WindowLength, \GuaranteeThreshold)\) from \(v\) is also a winning strategy for sure-\(\DirFWMP(\WindowLength, \GuaranteeThreshold)\)-positive-\(\ReachObj(T)\).
  Thus, \(v\) is winning for sure-\(\DirFWMP(\WindowLength, \GuaranteeThreshold)\)-positive-\(\ReachObj(T)\) in \(\MDP\).
  Since \(T \subseteq R\), we also have that \(v\) belongs to \(R\). 
  We have shown that for each vertex \(v\) in \(T\), the vertex \(v\) is in \(R\) if and only if \(v\) is winning for sure-\(\DirFWMP(\WindowLength, \GuaranteeThreshold)\)-positive-\(\ReachObj(T)\) in \(\MDP\).
  
  Now we see the interesting case, that is when \(v \in V \setminus T\).
  
  \((\Rightarrow)\)
  Suppose \(v\) belongs to \(R\).
  Since \(v\) also belongs to \(V \setminus T\), we have that \(v\) belongs to \(R \setminus T\).
  The vertex \(v\) belongs to \(R \setminus T\) because it is added to the set \(R\) in Line~\ref{alg-line:sdfpr-add-good-start-vertex} when an out-edge \(e\) of \(v\) is determined to be good by \(\IsGoodEdgeAlgo\) with respect to the set \(E_{g}\) at that point.
  From \Cref{lem:sdfpr-isgoodedge-correctness}, we have that this edge \(e\) is winning for sure-\(\DirFWMP(\WindowLength, \GuaranteeThreshold)\)-positive-\(\ReachObj_{\{e\} \union E_{g}}(T)\) in \(\MDP\).
  From \Cref{prop:sdfpr-isgoodvertex}, it follows that \(v\) is winning for sure-\(\DirFWMP(\WindowLength, \GuaranteeThreshold)\)-positive-\(\ReachObj_{E_g}(T)\) in \(\MDP\) for the set \(E_{g}\) at the end of   \Cref{alg:sure-dirfwmp-positive-reach}.
  By \Cref{lem:sdfpr-bad-edge-is-losing}, it implies that \(v\) is winning for sure-\(\DirFWMP(\WindowLength, \GuaranteeThreshold)\)-positive-\(\ReachObj(T)\) in \(\MDP\).

  \((\Leftarrow)\)
  Suppose \(v\) does not belong to \(R\).
  Since \(T \subseteq R\) and \(v \not\in R\), in order to satisfy positive-\(\ReachObj(T)\) from \(v\), the token needs to reach \(R\) with positive probability.
  From \Cref{prop:sdfpr-all-inedges-r-bad}, we have that starting from \(v\), if the token enters \(R\), it enters \(R\) through an edge in \(E_{b}\).
  However, from \Cref{lem:sdfpr-isgoodedge-correctness}, it follows that once an edge \(e\) in \(E_{b}\) is taken, then \(\PlayerMain\) does not have a strategy to satisfy sure-\(\DirFWMP(\WindowLength, \GuaranteeThreshold)\)-positive-\(\ReachObj_{\{e\} \union E_{g}}(T)\), and from \Cref{lem:sdfpr-bad-edge-is-losing}, player \(\PlayerMain\) does not have a strategy to satisfy sure-\(\DirFWMP(\WindowLength, \GuaranteeThreshold)\)-positive-\(\ReachObj(T)\) either.
  Thus, we have that \(v\) is not winning for sure-\(\DirFWMP(\WindowLength, \GuaranteeThreshold)\)-positive-\(\ReachObj(T)\) in \(\MDP\).
\end{proof}

\paragraph*{Memory requirement and complexity for sure-almost-sure \(\FWMPL\)}

We conclude this section by discussing bounds on the memory requirement and the complexity for the sure-almost-sure satisfaction of \(\FWMPL\) objectives.

\subparagraph*{Memory requirement: lower bound.}
It is known that memory of size \(\WindowLength\) suffices for satisfying sure-\(\FWMP(\WindowLength, \GuaranteeThreshold)\)~\cite[Theorem~4.4]{DGG25} and for almost-sure-\(\FWMP(\WindowLength, \GuaranteeThreshold)\)~\cite[Theorem~6.5]{DGG25}, and this is independent of the size of the MDP.
In \Cref{exa:fwmp-memory-lower-bound}, we show that $\Omega(\max\{\abs{V}, \WindowLength \})$ memory may be necessary for the sure-almost-sure satisfaction of the \(\FWMP(\WindowLength)\) objective.

\begin{example} \label{exa:fwmp-memory-lower-bound}
  Consider the MDP in Figure~\ref{fig:sdfpr-memory-lower-bound}.
  Note that the only vertex belonging to $\PlayerMain$ with more than one out-edge is $a$.
  Thus, a strategy of \(\PlayerMain\) can be fully specified by describing the behaviour of the token when it reaches $a$.
  Starting from vertex $v_3$, the memoryless strategy that always chooses $v_2$ from vertex $a$ satisfies sure-\(\FWMP(2,0)\) and the memoryless strategy that always chooses vertex $b$ from $a$ satisfies almost-sure-\(\FWMP(2,5)\).
  However, neither of these strategies are winning for sure-\(\FWMP(2,0)\)-almost-sure-\(\FWMP(2, 5)\). 

  In order to satisfy sure-\(\FWMP(2, 0)\), player \(\PlayerMain\) must eventually close all \(0\)-windows in at most \(2\) steps, and in order to satisfy almost-sure-\(\FWMP(2, 5)\) from $v_3$, the token must reach vertex $b$ repeatedly as long as it has not reached \(c\).
  However, each time the token moves to \(v_{3}\), the token then has to take the edge \((v_{3}, a)\) which opens a \(0\)-window.
  In order to reach \(b\) from \(v_{3}\) with all \(0\)-windows closed in at most \(2\) steps, the token must visit \(v_{2}\) and \(v_{1}\) on the way.
  This is because when the token reaches \(a\) from \(v_{3}\), the token must move to \(v_{2}\) to close the \(0\)-window starting at \(v_{3}\).
  The token then reaches \(a\), from where moving to either \(v_{2}\) or \(v_{1}\) results in closing the \(0\)-window opened at \(v_{2}\).
  If the token moves between \(v_{2}\) and \(a\) forever, then the token makes no progress towards reaching \(b\), and thus it must eventually move from \(a\) to \(v_{1}\).
  From \(v_{1}\), the token can move to  \(a\) and then to \(b\) while also ensuring that the \(0\)-window starting at \(v_{1}\) closes in \(2\) steps.
  Thus, the token needs to make visits to $v_2$, $v_1$, and $b$ from $a$ to close all \(0\)-windows in \(2\) steps, and this requires a strategy of memory size~\(3\).

  This example can be generalized to an MDP \(\MDP_{n}\) with $n+3$ vertices $\{v_n, v_{n-1}, \ldots,\) \(v_1, a, b, c\}$, and edges $(v_i, a)$ with payoff $-i$ and edges $(a, v_{i-1})$ with payoff $+i$ for $2 \leq i \leq n$, an edge \((a, b)\) with payoff \(1\), an edge \((b, v_{n})\) with payoff \(0\), an edge \((b, c)\) with payoff \(0\), and a self-loop on \(c\) with payoff \(5\).
  It can be seen that memory of size \(n\) is necessary to satisfy sure-\(\FWMP(2,0)\)-almost-sure-\(\FWMP(2, 5)\) in \(\MDP_{n}\) from \(v_{n}\).
  Since satisfying sure-$\FWMP(\WindowLength, \GuaranteeThreshold)$ requires at least $\WindowLength$ memory in general, at least $\Omega(\max\{\abs{V}, \WindowLength\})$ memory is necessary to satisfy sure-\(\FWMP(\WindowLength, \GuaranteeThreshold)\)-almost-sure-\(\FWMP(\WindowLength, \AlmostSureThreshold)\).
  \lipicsEnd
\end{example}
\begin{figure}[t]
  \centering
  \begin{tikzpicture}[node distance=1.5cm]
    \node[state, initial] (v3) {\(v_{3}\)};
    \node[state, below of=v3, yshift=3mm] (a) {\(a\)};
    \node[state, left of=a, xshift=-4mm] (v2) {\(v_{2}\)};
    \node[state, below of=a, yshift=3mm] (v1) {\(v_{1}\)};
    \node[draw, random, right of=a, xshift=3mm] (b) {\(b\)};
    \node[state, right of=b] (c) {\(c\)};

    \draw 
    (v3) edge node[left]{\small \(\EdgeValues{-3}{}\)} (a)
    
    (a) edge[bend right=15] node[above, pos=0.6]{\small \(\EdgeValues{3}{}\)} (v2)
    (v2) edge[bend right=15] node[below, pos=0.3]{\small \(\EdgeValues{-2}{}\)} (a)

    (a) edge[bend right=15] node[left]{\small \(\EdgeValues{2}{}\)} (v1)
    (v1) edge[bend right=15] node[right]{\small \(\EdgeValues{-1}{}\)} (a)
    
    (a) edge node[auto]{\small \(\EdgeValues{1}{}\)} (b)
    
    (b) edge[bend right] node[above right]{\small \(\EdgeValues{0}{.6}\)} (v3)
    (b) edge node[auto]{\small \(\EdgeValues{0}{.4}\)} (c)
    
    (c) edge[loop right] node[auto]{\small \(\EdgeValues{5}{}\)} (c)
    ;
  \end{tikzpicture}

  \caption{%
    There exist memoryless winning strategies for sure-\(\FWMP(2,0)\) and for almost-sure-\(\FWMP(2, 5)\) from all vertices in the MDP, but memory of size at least 3 is required for a strategy that is winning for sure-\(\FWMP(2,0)\)-almost-sure-\(\FWMP(2, 5)\) from \(v_{3}\).
  }
  \label{fig:sdfpr-memory-lower-bound}
\end{figure}

\subparagraph*{Memory requirement: upper bound.}
We give an upper bound for the memory requirement for winning strategies for sure-almost-sure satisfaction of \(\FWMPL\) by explicitly constructing a winning strategy \(\Strategy[\sas]\).
We first construct a winning strategy \(\Strategy[\sdpr]\) for sure-\(\DirFWMP(\WindowLength, \GuaranteeThreshold)\)-positive-\(\ReachObj(T)\) having memory size at most \(3 \cdot \abs{V} \cdot \WindowLength\).
Then, we use \(\Strategy[\sdpr]\) to construct a winning strategy \(\Strategy[\sdab]\) for sure-\(\DirFWMP(\WindowLength, \GuaranteeThreshold)\)-almost-sure-\(\BuchiObj(T)\) with memory size no greater than that of \(\Strategy[\sdpr]\).
Finally, we use \(\Strategy[\sdab]\) to describe a winning strategy \(\Strategy[\sas]\) for sure-almost-sure satisfaction of \(\FWMPL\) with memory size no greater than that of \(\Strategy[\sdab]\).

\emph{Winning strategy for sure-\(\DirFWMP(\WindowLength, \GuaranteeThreshold)\)-positive-\(\ReachObj(T)\).}
Recall that we have a set \(E_{g}\) of edges and a winning region \(R\) for sure-\(\DirFWMP(\WindowLength, \GuaranteeThreshold)\)-positive-\(\ReachObj(T)\) at the end of the execution of \Cref{alg:sure-dirfwmp-positive-reach}.
Let \(u_1, u_2, \ldots, u_{n}\) denote the set of vertices in \(R\setminus T\) at the end of the algorithm in the order in which they were added to \(R \setminus T\) where \(\abs{R \setminus T} = n\).
For \(1 \le i \le n\), let \(e_i\) denote the first out-edge of \(u_i\) added to \(E_g\) during the algorithm, and let \(E_{g}^{i}\) denote the set of all edges present in \(E_{g}\) immediately after the edge \(e_i\) is added to \(E_g\). 

For an arbitrary vertex \(v\) belonging to \(R\), we describe a winning strategy \(\Strategy[\sdpr]\) from \(v\). 
If \(v\) belongs to \(T\), then \(\ReachObj(T)\) is already satisfied at the start of the play, and it suffices for \(\Strategy[\sdpr]\) to mimic a winning strategy for sure-\(\DirFWMP(\WindowLength, \GuaranteeThreshold)\) from \(v\) throughout the rest of the play.
Otherwise, if \(v \not\in T\), then \(v \in R \setminus T\), and \(v\) is \(u_i\) for some \(1 \le i \le n\).
The construction of \(\MDP_{e_{i}}\)  ensures that there exists a winning strategy \(\Strategy[e_{i}]\) for sure-\(\DirFWMP(\WindowLength, \GuaranteeThreshold)\) from \(\widehat{u_{i}}\) in \(\MDP_{e_{i}}\), and from \Cref{lem:sdfpr-isgoodedge-correctness} it follows that projecting this strategy to \(\MDP\) gives a strategy \(\Strategy[i]\) that is winning for 
sure-\(\DirFWMP(\WindowLength, \GuaranteeThreshold)\)-positive-\(\ReachObj_{E_{g}^i}(T)\) from \(u_i\).
The strategy \(\Strategy[\sdpr]\) mimics \(\Strategy[i]\) from \(u_i\), as long as only edges in \(E_g^i\) are chosen, until either all \(\GuaranteeThreshold\)-windows are closed for the first time (which happens in at most \(\WindowLength\) steps) or the token reaches \(T\), whichever occurs first.
If the token ever takes an edge not in \(E_g^i\), then \(\Strategy[\sdpr]\) continues to follow a winning strategy for sure-\(\DirFWMP(\WindowLength, \GuaranteeThreshold)\) for the rest of the play, ensuring that all \(\GuaranteeThreshold\)-windows close in at most \(\WindowLength\) steps.
Otherwise, with probability at least \((\ProbabilityFunction_{\min})^\WindowLength\), the token only chooses edges in \(E_{g}^{i}\) until all \(\GuaranteeThreshold\)-windows are closed for the first time or the token reaches \(T\).

If the token reaches \(T\), then \(\Strategy[\sdpr]\) continues to follow a winning strategy for sure-\(\DirFWMP(\WindowLength, \GuaranteeThreshold)\) for the rest of the play.
Otherwise, the token reaches a vertex \(u_{j}\) with all \(\GuaranteeThreshold\)-windows closed for some \(1 \le j < i\), and \(\Strategy[\sdpr]\) switches to mimicking a winning strategy for sure-\(\DirFWMP(\WindowLength, \GuaranteeThreshold)\)-positive-\(\ReachObj_{E_{g}^j}(T)\) from \(u_j\).
Since \(n\) is finite and \(E_{g}^{j}\subsetneq E_{g}^{i}\) for \(1 \le j < i \le n\), we can repeat the same procedure from \(u_{j}\) to get that, with positive probability, the token reaches \(T\) eventually (in fact, in at most \(n \cdot \WindowLength\) steps).

Each strategy \(\Strategy[e_{i}]\) is a winning strategy for sure-\(\DirFWMP(\WindowLength, \GuaranteeThreshold)\) and requires at most \(\WindowLength\) memory~\cite[Lemma~4.1]{DGG25}.
The strategy \(\Strategy[i]\) obtained after projecting \(\Strategy[e_{i}]\) to \(\MDP\) requires at most three times the memory, that is, at most \(3 \cdot \WindowLength\).
For each \(1 \le i \le n\), there is a potentially different strategy \(\Strategy[i]\) that \(\Strategy[\sdpr]\) mimics, and thus, the memory requirement \(\Strategy[\sdpr]\) is at most \(3 \cdot n \cdot \WindowLength\), which is at most \(3 \cdot \abs{V} \cdot \WindowLength\).

\emph{Winning strategy for sure-\(\DirFWMP(\WindowLength, \GuaranteeThreshold)\)-almost-sure-\(\BuchiObj(T)\).}
For an arbitrary vertex \(v\) belonging to the winning region for sure-\(\DirFWMP(\WindowLength, \GuaranteeThreshold)\)-almost-sure-\(\BuchiObj(T)\), we define a winning strategy \(\Strategy[\sdab]\).
The strategy \(\Strategy[\sdab]\) follows the winning strategy \(\Strategy[\sdpr]\) for sure-\(\DirFWMP(\WindowLength, \GuaranteeThreshold)\)-positive-\(\ReachObj(T)\) as described above as long as there are no deviations from \(E_{g}\) and no vertex in \(T\) is reached.
If there is a deviation or a vertex in \(T\) is reached, then the strategy \(\Strategy[\sdab]\) continues to follow the strategy \(\Strategy[\sdpr]\) for at most \(\WindowLength - 1\) more steps until all open \(\GuaranteeThreshold\)-windows are closed, following which the strategy \(\Strategy[\sdab]\) resets,
that is, it forgets the history and starts following \(\Strategy[\sdpr]\) afresh.
The memory requirement of \(\Strategy[\sdab]\) is no greater than that of \(\Strategy[\sdpr]\), that is at most \(3 \cdot \abs{V} \cdot \WindowLength\).

\emph{Winning strategy for sure-\(\FWMP(\WindowLength, \GuaranteeThreshold)\)-almost-sure-\(\FWMP(\WindowLength, \AlmostSureThreshold)\).}
Recall that the winning region for sure-almost-sure satisfaction of \(\FWMPL\) has three types of regions: \(A^{i}\) (attractor region) for \(1 \le i \le k\), \(W_{\sdab}^{i}\) (sure-\(\DirFWMP(\WindowLength, \GuaranteeThreshold)\)-almost-sure-\(\BuchiObj(P^{i})\) winning region) for \(1 \le i \le k\), and \(W^{0}\) (sure-\(\FWMP(\WindowLength, \AlmostSureThreshold)\) winning region).
We describe a winning strategy \(\Strategy[\sas]\) for sure-\(\FWMP(\WindowLength, \GuaranteeThreshold)\)-almost-sure-\(\FWMP(\WindowLength, \AlmostSureThreshold)\):
\begin{itemize}
\item  If the token is in \(A^{i}\) for some \(1 \le i \le k\), then \(\Strategy[\sas]\) follows a sure-attractor strategy to reach \( W_{\sdab}^{i} \union W^{i-1} \).
\item If the token is in \(W_{\sdab}^{i}\) for some \(1 \le i \le k\), then \(\Strategy[\sas]\) follows the winning strategy \(\Strategy[\sdab]\) for sure-\(\DirFWMP(\WindowLength, \GuaranteeThreshold)\)-almost-sure-\(\BuchiObj(P^{i})\).
  This ensures that \(\DirFWMP(\WindowLength, \GuaranteeThreshold)\) is satisfied in all outcomes, and with probability \(1\), the token reaches \(W^{i-1}\).
\item If the token is in \(W^{0}\), then \(\Strategy[\sas]\) follows a winning strategy for sure-\(\FWMP(\WindowLength, \AlmostSureThreshold)\).
\end{itemize}
The memory requirement of \(\Strategy[\sas]\) is no greater than that of \(\Strategy[\sdab]\), that is at most \(3 \cdot \abs{V} \cdot \WindowLength\).

\subparagraph*{Complexity.}
We recall that \(\SureDirFWMPAlgo\) can be computed in time $\bigO(|V| \cdot \abs{E}  \cdot \WindowLength \cdot \log(\wmax))$ \cite[Lemma~5]{CDRR15} and \(\SureFWMPAlgo\) can be computed in time $\bigO(|V|^2 \cdot \abs{E}  \cdot \WindowLength \cdot \log(\wmax))$~\cite[Lemma~6]{CDRR15}, which are polynomial in the size of the input when $\WindowLength$ is given in unary.
We denote the running time of \(\SureDirFWMPAlgo\) by $\mathbb{C}$.

We show that when \(\WindowLength\) is given in unary, \Cref{alg:sure-fwmp-almostsure-fwmp} (\(\SASFWMPAlgo\)) has a polynomial running time, and thus, the sure-almost-sure problem for the \(\FWMPL\) objectives is in \(\PTime\).
The \textbf{repeat} loop in \Cref{alg:sure-fwmp-almostsure-fwmp} is dominated by the call to \Cref{alg:sure-dirfwmp-almostsure-buchi}, which in turn is a recursive algorithm that is dominated by the call to \Cref{alg:sure-dirfwmp-positive-reach}.
In \Cref{alg:sure-dirfwmp-positive-reach}, we have a \textbf{repeat} loop that runs at most \(\abs{E}^2\) times.
In each iteration of this loop, we make a call to \(\IsGoodEdgeAlgo\) which involves calling \(\SureDirFWMPAlgo\) on the augmented MDP \(\MDP_{e}\) that has at most thrice as many vertices as the original MDP \(\MDP\).
Thus, the running time of \Cref{alg:sure-dirfwmp-positive-reach} is $\bigO(\abs{E}^2 \cdot \mathbb{C})$.
Algorithm~\ref{alg:sure-dirfwmp-almostsure-buchi} has a recursion depth of at most $|V|$ since it is called on a smaller \(\MDP\) each time, and hence its time complexity is $\bigO(\abs{V} \cdot \abs{E}^2 \cdot \mathbb{C})$.
Finally, the \textbf{repeat} loop of Algorithm~\ref{alg:sure-fwmp-almostsure-fwmp} runs for at most $\abs{V}$ iterations, thus leading to a total time complexity of $\bigO(\abs{V}^2 \cdot \abs{E}^2 \cdot \mathbb{C})$, which is $\bigO(\abs{V}^3 \cdot \abs{E}^3  \cdot \WindowLength \cdot \log(\wmax))$.
This is polynomial in the size of the input when $\WindowLength$ is given in unary.

\begin{theorem} \label{thm:fwmp-sas-complexity}
  The sure-almost-sure problem for the \(\FWMP(\WindowLength)\) objective is in \(\PTime\) when \(\WindowLength\) is given in unary.
  The memory required is at least \(\Omega(\max\{\abs{V}, \WindowLength\})\) and at most \(\bigO(\abs{\Vertices} \cdot \WindowLength)\).
\end{theorem}

\subsection{Sure-almost-sure satisfaction of BWMP}
In this section, we present an algorithm to solve the sure-almost-sure satisfaction problem for \(\BWMP\), the bounded window mean-payoff objective.
Recall that a play \(\Play\) satisfies \(\BWMP(\GuaranteeThreshold)\) if there exists a window length \(\WindowLength \ge 1\) such that \(\Play\) satisfies \(\FWMP(\WindowLength, \GuaranteeThreshold)\).

The algorithm is similar to \Cref{alg:sure-fwmp-almostsure-fwmp}.
To solve the sure-almost-sure satisfaction problem for \(\BWMP\), we consider \Cref{alg:sure-fwmp-almostsure-fwmp} for the \(\FWMPL\) objectives and replace every occurrence of \(\FWMPL\) with \(\BWMP\).
Thus, the algorithm for sure-almost-sure \(\BWMP\) has the same structure as \Cref{alg:sure-fwmp-almostsure-fwmp}, that is, it has subalgorithms for sure-\(\DirBWMP(\GuaranteeThreshold)\)-almost-sure-\(\BuchiObj(T)\) and sure-\(\DirBWMP(\GuaranteeThreshold)\)-positive-\(\ReachObj(T)\)
(analogous to sure-\(\DirFWMP(\WindowLength, \GuaranteeThreshold)\)-almost-sure-\(\BuchiObj(T)\) and sure-\(\DirFWMP(\WindowLength, \GuaranteeThreshold)\)-positive-\(\ReachObj(T)\) respectively, obtained again by replacing all occurrences of \(\FWMPL\) with \(\BWMP\)).
  We state the analogous algorithms for sure-almost-sure \(\BWMP\), sure-\(\DirBWMP(\GuaranteeThreshold)\)-almost-sure-\(\BuchiObj(T)\), and sure-\(\DirBWMP(\GuaranteeThreshold)\)-positive-\(\ReachObj(T)\) (\Cref{alg:sure-bwmp-almostsure-bwmp}, \Cref{alg:sure-dirbwmp-almostsure-buchi}, and \Cref{alg:sure-dirbwmp-positive-reach}, respectively) in the appendix for completeness.
  The proof of correctness for these algorithms also follows directly from the proof of correctness of the algorithms for the \(\FWMPL\) objectives (\Cref{thm:sasf-correctness}, \Cref{thm:sdfab-correctness}, \Cref{thm:sdfpr-correctness}) by replacing \(\FWMPL\) with \(\BWMP\).
  We state in \Cref{lem:dirbwmp-suptotal-goodwin-dirfwmp-equivalence} and \Cref{lem:dirbwmp-equivalent-to-dirfwmpl} results that relate the winning region for sure-\(\DirBWMP(\GuaranteeThreshold)\) with the winning region for sure-\(\DirFWMP(\WindowLength', \GuaranteeThreshold)\) where \(\WindowLength' = \abs{V} \cdot (\abs{V} \cdot \PayoffFunction_{\max} + 1)\).
  \begin{lemma}\label{lem:dirbwmp-suptotal-goodwin-dirfwmp-equivalence}
    All vertices in \(\MDP\) are winning for sure-\(\DirBWMP(0)\) if and only if all vertices in \(\MDP\) are winning for sure-\(\DirFWMP(\WindowLength', 0)\), where \(\WindowLength' = \abs{V} \cdot (\abs{V} \cdot \PayoffFunction_{\max} + 1)\).
  \end{lemma}
  \begin{proof}
    We list some observations that are useful in the proof.
    \begin{enumerate}
    \item If a play \(\Play\) satisfies \(\DirBWMP(0)\), then \(\Play\) also satisfies the total-payoff objective \(\TP(0)\) in \(\MDP\).
      This is because it happens infinitely often in \(\Play\) that all \(0\)-windows are closed, and each time all \(0\)-windows are closed, it is the case that the total payoff is non-negative.
      Thus, if a vertex \(v\) is winning for sure-\(\DirBWMP(0)\) in \(\MDP\), then \(v\) is also winning for sure-\(\TP(0)\) in \(\MDP\).
    \item We recall \cite[Lemma 10]{CDRR15} which states that if a vertex \(v\) is winning for sure-\(\TP(0)\) in \(\MDP\) then \(v\) is winning for sure-\(\GW(\WindowLength', 0)\) in \(\MDP\) for \(\WindowLength' = \abs{V} \cdot (\abs{V} \cdot \PayoffFunction_{\max} + 1)\), that is, \(\PlayerMain\) can ensure for all outcomes that the \(0\)-window starting at the \(v\) closes in at most \(\WindowLength'\) steps.
    \item We recall from \cite[Lemma 5]{CDRR15} that if all vertices in \(\MDP\) are winning for sure-\(\GW(\WindowLength', 0)\), then all vertices in \(\MDP\) are winning for sure-\(\DirFWMP(\WindowLength', 0)\) since each time a \(0\)-window closes, satisfaction of sure-\(\GW(\WindowLength', 0)\) ensures that the next \(0\)-window gets closed in at most \(\WindowLength'\) steps.
    \item If a vertex \(v\) is winning for sure-\(\DirFWMP(\WindowLength', 0)\) in \(\MDP\), then \(v\) is also winning for sure-\(\DirBWMP(0)\) in \(\MDP\), since \(\WindowLength'\) is a witness window length for the satisfaction of \(\DirBWMP(0)\).
    \end{enumerate}
    Now, suppose all vertices in \(\MDP\) are winning for sure-\(\DirBWMP(0)\).
    Then, we have that all vertices in \(\MDP\) are winning for sure-\(\TP(0)\) (Observation 1), which implies that all vertices in \(\MDP\) are winning for sure-\(\GW(\WindowLength', 0)\) (Observation 2), which implies that all vertices in \(\MDP\) are winning for  sure-\(\DirFWMP(\WindowLength', 0)\) (Observation 3), which in turn implies that  all vertices are winning for sure-\(\DirBWMP(0)\) (Observation 4).
    Thus, all four of the preceding statements are equivalent, and in particular, we have that all vertices in \(\MDP\) are winning for sure-\(\DirBWMP(0)\) if and only if all vertices in \(\MDP\) are winning for sure-\(\DirFWMP(\WindowLength', 0)\).
  \end{proof}
  \Cref{lem:dirbwmp-equivalent-to-dirfwmpl} lets us use \(\DirFWMP(\WindowLength', \GuaranteeThreshold)\) as a drop-in replacement for \(\DirBWMP(\GuaranteeThreshold)\) whenever necessary.
  \begin{lemma}\label{lem:dirbwmp-equivalent-to-dirfwmpl}
    For every vertex \(v\) in an MDP \(\MDP\), the vertex \(v\) is winning for sure-\(\DirBWMP(\GuaranteeThreshold)\) if and only if \(v\) is winning for sure-\(\DirFWMP(\WindowLength', \GuaranteeThreshold)\), where \(\WindowLength' = \abs{V} \cdot (\abs{V} \cdot \PayoffFunction_{\max} + 1)\).
  \end{lemma}
  \begin{proof}
    Let \(W_{\textsf{B}}\) denote the winning region for sure-\(\DirBWMP(\GuaranteeThreshold)\) in \(\MDP\), and let \(v\) be a vertex in \(W_{\textsf{B}}\).
    Since \(\DirBWMP(\GuaranteeThreshold)\) is an objective that is closed under suffixes, from \Cref{prop:winning-region-sure-objective-closed-under-suffixes-induces-submdp} we have that \(W_{\textsf{B}}\) induces a subMDP \(\subMDP{\MDP}{W_{\textsf{B}}}\) of \(\MDP\).

    Every vertex in \(\subMDP{\MDP}{W_{\textsf{B}}}\) is winning for sure-\(\DirBWMP(\GuaranteeThreshold)\).
    If we subtract \(\GuaranteeThreshold\) from the payoff of every edge in \(\subMDP{\MDP}{W_{\textsf{B}}}\), then we obtain an MDP \((\subMDP{\MDP}{W_{\textsf{B}}})_{-\GuaranteeThreshold}\)  in which every vertex is winning for sure-\(\DirBWMP(0)\).
    Using \Cref{lem:dirbwmp-suptotal-goodwin-dirfwmp-equivalence}, we get that every vertex in \((\subMDP{\MDP}{W_{\textsf{B}}})_{-\GuaranteeThreshold}\) is winning for sure-\(\DirFWMP(\WindowLength', 0)\).
    Adding back \(\GuaranteeThreshold\) to every edge payoff in \((\subMDP{\MDP}{W_{\textsf{B}}})_{-\GuaranteeThreshold}\) gives us that every vertex in \(\subMDP{\MDP}{W_{\textsf{B}}}\) is winning for sure-\(\DirFWMP(\WindowLength', \GuaranteeThreshold)\).
    Since \(\subMDP{\MDP}{W_{\textsf{B}}}\) is a subMDP of \(\MDP\), we get that every vertex in \(W_{\textsf{B}}\) is also winning for sure-\(\DirFWMP(\WindowLength', \GuaranteeThreshold)\) in the original MDP \(\MDP\).

    For the converse direction, since \(\DirFWMP(\WindowLength', \GuaranteeThreshold)\) is also an objective that is closed under suffixes, we can follow the same proof as above starting with the winning region \(W_{\textsf{F}}\) for sure-\(\DirFWMP(\WindowLength', \GuaranteeThreshold)\) in \(\MDP\) to show that every vertex in \(\MDP\) that is winning for sure-\(\DirFWMP(\WindowLength', \GuaranteeThreshold)\) is also winning for sure-\(\DirBWMP(\GuaranteeThreshold)\).
  \end{proof}

  \begin{lemma}\label{lem:sdbpr-sdfpr-equivalence}
    For every vertex \(v\) in an MDP \(\MDP\), the vertex \(v\) is winning for sure-\(\DirBWMP(\GuaranteeThreshold)\)-positive-\(\ReachObj(T)\) if and only if \(v\) is winning for sure-\(\DirFWMP(\WindowLength'', \GuaranteeThreshold)\)-positive-\(\ReachObj(T)\) for \(\WindowLength'' = 3 \cdot \abs{V} \cdot (3 \cdot \abs{V} \cdot \PayoffFunction_{\max} + 1) = 9 \cdot \WindowLength'\).
  \end{lemma}
  \begin{proof}
    We recall the proof of correctness of \Cref{alg:sure-dirfwmp-positive-reach} using window length \(\WindowLength'' = 3 \cdot \abs{V} \cdot (3 \cdot \abs{V} \cdot \PayoffFunction_{\max} + 1)\), giving us an algorithm for sure-\(\DirFWMP(\WindowLength'', \GuaranteeThreshold)\)-positive-\(\ReachObj(T)\).
    Observing that \Cref{alg:sure-dirfwmp-positive-reach} works with MDPs \(\MDP_{e}\) that have at most thrice as many vertices as \(\MDP\), we use \Cref{lem:dirbwmp-equivalent-to-dirfwmpl} in the proof of correctness of \Cref{alg:sure-dirfwmp-positive-reach} to replace every occurrence of
    ``winning for sure-\(\DirFWMP(\WindowLength'', \GuaranteeThreshold)\)'' with 
    ``winning for sure-\(\DirBWMP(\GuaranteeThreshold)\)''.
    This gives that the winning region for sure-\(\DirBWMP(\GuaranteeThreshold)\)-positive-\(\ReachObj(T)\) coincides with the winning region for sure-\(\DirFWMP(\WindowLength'', \GuaranteeThreshold)\)-positive-\(\ReachObj(T)\).
  \end{proof}

  We can follow the process described in the proof of \Cref{lem:sdbpr-sdfpr-equivalence} for the proofs of correctness of \Cref{alg:sure-dirfwmp-almostsure-buchi} and \Cref{alg:sure-fwmp-almostsure-fwmp} to obtain \Cref{lem:sdbab-sdfab-equivalence} and \Cref{lem:sasb-sasf-equivalence} for sure-\(\DirBWMP(\GuaranteeThreshold)\)-almost-sure-\(\BuchiObj(T)\) and sure-\(\BWMP(\GuaranteeThreshold)\)-almost-sure-\(\BWMP(\AlmostSureThreshold)\) respectively.
  In order to show the correctness of the algorithm for sure-\(\DirBWMP(\GuaranteeThreshold)\)-almost-sure-\(\BuchiObj(T)\), we need to prove an analogous variant of \Cref{lem:sdfab-v-l-bound} for \(\BWMP\), which we do in \Cref{lem:sdbab-bound}.

  \begin{lemma}\label{lem:sdbab-bound}
    For every vertex \(v\) in an MDP \(\MDP\), if \(\PlayerMain\) has a strategy to reach a target \(T \subseteq \Vertices\) from \(v\) with positive probability while simultaneously satisfying sure-\(\DirBWMP(\GuaranteeThreshold)\), then \(\PlayerMain\) also has a strategy to reach \(T\) from \(v\) with positive probability in at most \(\abs{V} \cdot \WindowLength''\) (that is \(3 \cdot \abs{\Vertices}^{2} \cdot (3 \cdot \abs{\Vertices} \cdot \PayoffFunction_{\max} + 1)\)) steps while simultaneously satisfying  sure-\(\DirBWMP(\GuaranteeThreshold)\).
  \end{lemma}
  \begin{proof}
    From \Cref{lem:dirbwmp-equivalent-to-dirfwmpl}, we have that if \(\PlayerMain\) can ensure that there exists a window length \(\WindowLength\) such that all \(\GuaranteeThreshold\)-windows in every outcome of the MDP close in at most \(\WindowLength\) steps, then for \(\WindowLength'' = 3 \cdot \abs{V} \cdot (3 \cdot \abs{V} \cdot \PayoffFunction_{\max} + 1)\) we have that \(\PlayerMain\) has a strategy such that all \(\GuaranteeThreshold\)-windows close in at most \(\WindowLength''\) steps in every outcome of the MDP.
    Then, we use \Cref{lem:sdfab-v-l-bound} with window length \(\WindowLength''\)  to get the desired result.
  \end{proof}

  \begin{lemma}\label{lem:sdbab-sdfab-equivalence}
    For every vertex \(v\) in an MDP \(\MDP\), the vertex \(v\) is winning for sure-\(\DirBWMP(\GuaranteeThreshold)\)-almost-sure-\(\BuchiObj(T)\) if and only if \(v\) is winning for sure-\(\DirFWMP(\WindowLength'', \GuaranteeThreshold)\)-almost-sure-\(\BuchiObj(T)\) for \(\WindowLength'' = 3 \cdot \abs{V} \cdot (3 \cdot \abs{V} \cdot \PayoffFunction_{\max} + 1)\).
  \end{lemma}

\begin{lemma}\label{lem:sasb-sasf-equivalence}
  For every vertex \(v\) in an MDP \(\MDP\), the vertex \(v\) is winning for sure-\(\BWMP(\GuaranteeThreshold)\)-almost-sure-\(\BWMP(\AlmostSureThreshold)\) if and only if \(v\) is winning for sure-\(\FWMP(\WindowLength'', \GuaranteeThreshold)\)-almost-sure-\(\FWMP(\WindowLength'', \AlmostSureThreshold)\) for \(\WindowLength'' = 3 \cdot \abs{V} \cdot (3 \cdot \abs{V} \cdot \PayoffFunction_{\max} + 1)\).
\end{lemma}

\paragraph*{Memory requirement and complexity for sure-almost-sure BWMP}

We conclude this section by discussing bounds on the memory requirement and the complexity for the sure-almost-sure satisfaction of \(\BWMP\) objectives.

\subparagraph*{Memory requirement: lower bound.}
It is known that memoryless strategies suffice for satisfying sure-\(\BWMP(\GuaranteeThreshold)\)~\cite[Theorem~19]{CDRR15} and for satisfying almost-sure-\(\BWMP(\GuaranteeThreshold)\)~\cite[Theorem~6.2(b)]{BDOR20}.
In \Cref{exa:bwmp-memory-lower-bound}, we show that memory of size at least  $\Omega(\abs{V} \cdot \PayoffFunction_{\max})$ may be necessary for the sure-almost-sure satisfaction of \(\BWMP\).

\begin{example} \label{exa:bwmp-memory-lower-bound}
  \begin{figure}[t]
    \centering
    \begin{tikzpicture}[node distance=1.5cm]
      \node[state, initial] (v1) {\(v_{1}\)};
      \node[state, right of=v1] (a) {\(a\)};
      \node[draw, random, right of=a] (b) {\(b\)};
      \node[state, right of=b, xshift=3mm] (c) {\(c\)};

      \draw 
      (v1) edge node[above]{\small \(\EdgeValues{-10}{}\)} (a)

      (a) edge[loop above] node[auto]{\small \(\EdgeValues{1}{}\)} (a)
      (a) edge node[above]{\small \(\EdgeValues{-10}{}\)} (b)

      (b) edge node[above]{\small \(\EdgeValues{-10}{.1}\)} (c)
      (b) edge[bend left=35] node[below right, pos=0.4]{\small \(\EdgeValues{-10}{.9}\)} (v1)

      (c) edge[loop right] node[auto]{\small \(\EdgeValues{5}{}\)} (c)
      ;
    \end{tikzpicture}
    \caption{%
      There exist memoryless winning strategies for sure-\(\BWMP(0)\) and for almost-sure-\(\BWMP(5)\) from all vertices in the MDP, but memory of size at least 31 is required for a strategy that is winning for sure-\(\BWMP(0)\)-almost-sure-\(\BWMP(5)\) from \(v_{1}\).
    }
    \label{fig:sdbpr-memory-lower-bound}
  \end{figure}
  
  Consider the MDP in Figure~\ref{fig:sdbpr-memory-lower-bound}.
  Starting from vertex $v_1$, a memoryless strategy that always chooses the self-loop on vertex $a$ satisfies sure-\(\BWMP(0)\) while a memoryless strategy that always chooses vertex $b$ from $a$ satisfies almost-sure-\(\BWMP(5)\).
  However, neither of these strategies are winning for sure-\(\BWMP(0)\)-almost-sure-\(\BWMP(5)\).
  
  In order to satisfy almost-sure-\(\BWMP(5)\) from $v_1$, the token needs to repeatedly reach vertex $b$ until it reaches vertex \(c\).
  However, each time the token moves from \(b\) to \(v_{1}\), the token has to go through the edges \((b, v_{1})\), \((v_{1}, a)\), and \((a, b)\) before it can reach \(b\) again.
  This accrues a payoff of \(-30\) and is thus an open \(0\)-window.
  In order to also ensure satisfaction of sure-\(\BWMP(0)\), player \(\PlayerMain\) needs to eventually always be closing \(0\)-windows.
  If the token takes the self-loop on \(a\) fewer than \(30\) times each time it reaches \(a\) before moving on to \(b\), then this results in an outcome that contains infinitely many open \(0\)-windows that never close.
  On the other hand, if each time the token reaches \(a\) the token takes the self-loop on $a$ at least $30$ times before forwarding the token to \(b\), then as long as the token does not reach \(c\) all \(0\)-windows get closed in at most \(33\) steps, and once the token reaches \(c\), all \(5\)-windows get closed in \(1\) step.
  We see that this memoryful strategy which uses a counter of size \(31\) is winning for sure-\(\BWMP(0)\)-almost-sure-\(\BWMP(5)\).

  This example can be generalized to an MDP \(\MDP_{n,\wmax}\) with $n + 3$ vertices \(\{v_n, v_{n-1}, \dots\), \(v_1, a, b, c\}\) and edges $(v_i, v_{i-1})$ for $2 \leq i \leq n$,  \((v_{1}, a)\),  $(a,b)$, $(b, v_n)$, and $(b,c)$ each with payoff $-\PayoffFunction_{\max}$, a self-loop on \(a\) with payoff $1$, and a self-loop on $c$ with payoff $5$.
  A winning strategy for sure-\(\BWMP(0)\)-almost-sure-\(\BWMP(5)\) from $v_n$ in \(\MDP_{n,\wmax}\) requires memory of size at least $(n + 2) \cdot \PayoffFunction_{\max} + 1$, which is \(\Omega(\abs{V} \cdot \wmax)\).
  \lipicsEnd
\end{example}

\subparagraph*{Memory requirement: upper bound.}

By \Cref{lem:sasb-sasf-equivalence}, if there exists a winning strategy \(\Strategy[\textsf{B}]\) from a vertex \(v\) for sure-almost-sure \(\BWMP\), then there also exists a winning strategy \(\Strategy[\textsf{F}]\) from \(v\) for sure-almost-sure \(\FWMP(\WindowLength'')\) for \(\WindowLength'' = 3 \cdot \abs{V} \cdot (3 \cdot \abs{V} \cdot \PayoffFunction_{\max} + 1)\), and moreover, this strategy \(\Strategy[\textsf{F}]\) is also winning strategy for sure-almost-sure \(\BWMP\) from \(v\).
From \Cref{thm:sasf-correctness}, we have that the memory of size \(\abs{V} \cdot \WindowLength''\) (which is \(\bigO(\abs{\Vertices}^{3} \cdot \abs{\PayoffFunction_{\max}} )\)) is sufficient for sure-almost-sure \(\FWMP(\WindowLength'')\) and thus is also sufficient for sure-almost-sure \(\BWMP\).

\subparagraph*{Complexity.}
From \cite[Theorem~19]{CDRR15}, we have that satisfaction of sure-\(\BWMP(\GuaranteeThreshold)\) and of sure-\(\DirBWMP(\GuaranteeThreshold)\) are both in \(\NP \intersection \coNP\).
Similar to the case of the sure-almost-sure satisfaction of \(\FWMPL\), the algorithm for the sure-almost-sure satisfaction of \(\BWMP\) makes a polynomial number of calls to algorithms for satisfaction of sure-\(\BWMP(\GuaranteeThreshold)\), sure-\(\BWMP(\AlmostSureThreshold)\),
and sure-\(\DirBWMP(\GuaranteeThreshold)\).
Thus, we have that sure-almost-sure satisfaction of \(\BWMP\) is in \(\PTime^{\NP \intersection \coNP}\), which is the same as \(\NP \intersection \coNP\)~\cite{Sch83}.

\begin{theorem}\label{thm:bwmp-sas-complexity}
  The sure-almost-sure problem for the \(\BWMP\) objective is in \(\NP \intersection \coNP\).
  The memory required is at least \(\Omega(\abs{V} \cdot \PayoffFunction_{\max})\) and at most \(\bigO(\abs{\Vertices}^{3} \cdot \abs{\PayoffFunction_{\max}} )\).
\end{theorem}

\section{Sure-limit-sure for FWMP and BWMP} \label{sec:sls}
In this section, we see an algorithm to solve the sure-limit-sure satisfaction of \(\FWMPL\) and \(\BWMP\) objectives.
\Cref{alg:sure-fwmp-limitsure-fwmp} (\(\SLSFWMPAlgo\)) returns the set of all vertices \(v\) that are winning for sure-\(\FWMP(\WindowLength, \GuaranteeThreshold)\)-limit-sure-\(\FWMP(\WindowLength, \LimitSureThreshold)\) in an MDP \(\MDP\).
An analogous algorithm to compute the winning region for sure-limit-sure satisfaction of \(\BWMP\) objectives is given in the appendix.
It is obtained by simply replacing \(\SureFWMPAlgo\) and \(\AlmostSureFWMPAlgo\) with \(\SureBWMPAlgo\) and \(\AlmostSureBWMPAlgo\) respectively.

\begin{algorithm}[t]
  \caption{\(\SLSFWMPAlgo(\MDP, \WindowLength, \GuaranteeThreshold, \AlmostSureThreshold)\)}%
  \label{alg:sure-fwmp-limitsure-fwmp}
  \begin{algorithmic}[1]
    \Require An MDP \(\MDP = ((\Vertices, \Edges), (\VerticesMain, \VerticesRandom), \ProbabilityFunction, \PayoffFunction)\), window length \(\WindowLength\), sure threshold \(\GuaranteeThreshold\), and limit-sure threshold \(\LimitSureThreshold\)
    \Ensure The winning region for sure-\(\FWMP(\WindowLength, \GuaranteeThreshold)\)-limit-sure-\(\FWMP(\WindowLength, \AlmostSureThreshold)\) in \(\MDP\)
    \State \(W_{\Sure\GuaranteeThreshold} \assign \SureFWMPAlgo(\MDP, \WindowLength, \GuaranteeThreshold)\)
    \State \(W_{\AlmostSure\AlmostSureThreshold} \assign \AlmostSureFWMPAlgo(\subMDP{\MDP}{ W_{\Sure\GuaranteeThreshold}}, \WindowLength, \LimitSureThreshold)\)
    \State \Return \(W_{\AlmostSure\AlmostSureThreshold}\)
  \end{algorithmic}
\end{algorithm}

\subparagraph{Description of \Cref{alg:sure-fwmp-limitsure-fwmp}.}
We first compute the set \(W_{\Sure\GuaranteeThreshold}\) of vertices that are winning for sure-\(\FWMP(\WindowLength, \GuaranteeThreshold)\) in \(\MDP\) using \cite[Algorithm~1]{CDRR15}.
Then, in the subMDP \(\subMDP{\MDP}{W_{\Sure\GuaranteeThreshold}}\) induced by \(W_{\Sure\GuaranteeThreshold}\), we compute using \cite[Algorithm~1]{BDOR20} the set \(W_{\AlmostSure\AlmostSureThreshold}\) of vertices that are winning for almost-sure-\(\FWMP(\WindowLength, \LimitSureThreshold)\), and we return this set \(W_{\AlmostSure\AlmostSureThreshold}\).
We briefly describe \cite[Algorithm~1]{BDOR20}: it first computes the set \(W_{\Sure\AlmostSureThreshold}\) of all vertices in \(\subMDP{\MDP}{ W_{\Sure\GuaranteeThreshold}}\) that are winning for sure-\(\FWMP(\WindowLength, \AlmostSureThreshold)\) in \(\subMDP{\MDP}{ W_{\Sure\GuaranteeThreshold}}\) (once again using \cite[Algorithm~1]{CDRR15}).
Then, the algorithm performs the MEC decomposition of \(\subMDP{\MDP}{ W_{\Sure\GuaranteeThreshold}}\), and labels all the MECs that have a non-empty intersection with \(W_{\Sure\AlmostSureThreshold}\) as \emph{good MECs}.
The algorithm finally returns all vertices in \(\subMDP{\MDP}{ W_{\Sure\GuaranteeThreshold}}\) from which \(\PlayerMain\) can reach the set of good MECs with probability~\(1\).

\subparagraph{Correctness of \Cref{alg:sure-fwmp-limitsure-fwmp}.}
It is useful to note that since \(\FWMP(\WindowLength, \LimitSureThreshold)\) is prefix-independent, from \Cref{prop:winning-region-almost-sure-objective-closed-under-suffixes-induces-submdp}, the set \(W_{\AlmostSure\AlmostSureThreshold}\) that is the winning region for almost-sure-\(\FWMP(\WindowLength, \LimitSureThreshold)\) in \(\MDP\) induces a subMDP of \(\MDP\).
Thus, if the token reaches \(W_{\AlmostSure\AlmostSureThreshold}\), then \(\PlayerMain\) has a strategy that keeps the token in \(W_{\AlmostSure\AlmostSureThreshold}\).
We also recall that we have \(\GuaranteeThreshold < \LimitSureThreshold\).

\begin{theorem}\label{thm:sls-algorithm-correctness}
  \Cref{alg:sure-fwmp-limitsure-fwmp} computes the winning region for sure-\(\FWMP(\WindowLength, \GuaranteeThreshold)\)-limit-sure-\(\FWMP(\WindowLength, \LimitSureThreshold)\) in \(\MDP\).
\end{theorem}
\begin{proof}
  \((\Rightarrow)\)
  We show that if a vertex \(v\) belongs to the set \(W_{\AlmostSure\AlmostSureThreshold}\) returned by \Cref{alg:sure-fwmp-limitsure-fwmp}, then \(v\) is winning for sure-\(\FWMP(\WindowLength, \GuaranteeThreshold)\)-limit-sure-\(\FWMP(\WindowLength, \LimitSureThreshold)\) in \(\MDP\).
  That is, we show that for all \(\epsilon > 0\), player \(\PlayerMain\) has a strategy \(\Strategy[\epsilon]\) from \(v\) that satisfies both sure-\(\FWMP(\WindowLength, \GuaranteeThreshold)\) and probability(\(1 - \epsilon\))-\(\FWMP(\WindowLength, \LimitSureThreshold)\).
  
  Recall that the set \(W_{\AlmostSure\AlmostSureThreshold}\) is the winning region for almost-sure-\(\FWMP(\WindowLength, \LimitSureThreshold)\) in \(\subMDP{\MDP}{ W_{\Sure\GuaranteeThreshold}}\).
  Before we describe \(\Strategy[\epsilon]\), we describe a strategy \(\Strategy[\AlmostSure]\) that is winning for almost-sure-\(\FWMP(\WindowLength, \LimitSureThreshold)\) in \(\subMDP{\MDP}{ W_{\Sure\GuaranteeThreshold}}\) from all vertices in \(W_{\AlmostSure\AlmostSureThreshold}\).
  We recall the following from \cite[Lemma 6.1]{BDOR20}:
  if the winning region for almost-sure-\(\FWMP(\WindowLength, \LimitSureThreshold)\) in an MDP is non-empty, then the winning region for sure-\(\FWMP(\WindowLength, \LimitSureThreshold)\) in the MDP is also non-empty.
  Since we started with \(W_{\AlmostSure\AlmostSureThreshold}\)  non-empty, the winning region for sure-\(\FWMP(\WindowLength, \LimitSureThreshold)\) in \(\subMDP{\MDP}{ W_{\Sure\GuaranteeThreshold}}\) is also non-empty (we call this region \(W_{\Sure\AlmostSureThreshold}\)).
  Moreover, the set \(W_{\AlmostSure\AlmostSureThreshold}\) is precisely the set of all vertices in \(\subMDP{\MDP}{ W_{\Sure\GuaranteeThreshold}}\) from where \(\PlayerMain\) has a strategy to reach \(W_{\Sure\AlmostSureThreshold}\) with probability~\(1\).
  The strategy \(\Strategy[\AlmostSure]\) is defined as follows:
  if the token is in \(W_{\AlmostSure\AlmostSureThreshold} \setminus W_{\Sure\AlmostSureThreshold}\), then follow an almost-sure attractor strategy to \(W_{\Sure\AlmostSureThreshold}\), and if the token is in \(W_{\Sure\AlmostSureThreshold}\), then follow a winning strategy for sure-\(\FWMP(\WindowLength, \LimitSureThreshold)\) in \(\subMDP{\MDP}{W_{\Sure\AlmostSureThreshold}}\).

  If \(\PlayerMain\) follows this strategy \(\Strategy[\AlmostSure]\), then with probability \(1\), the token eventually reaches \(W_{\Sure\AlmostSureThreshold}\) and stays there forever, and eventually all \(\LimitSureThreshold\)-windows close in at most \(\WindowLength\) steps.
  Thus, this strategy is winning for almost-sure-\(\FWMP(\WindowLength, \LimitSureThreshold)\), and since \(\GuaranteeThreshold < \LimitSureThreshold\), we have that this strategy is also winning for almost-sure-\(\FWMP(\WindowLength, \GuaranteeThreshold)\).
  However, this strategy may not be winning for sure-\(\FWMP(\WindowLength, \GuaranteeThreshold)\) as there may exist an outcome of \(\Strategy[\AlmostSure]\) where the token gets stuck in \(W_{\AlmostSure\AlmostSureThreshold} \setminus W_{\Sure\AlmostSureThreshold}\) forever and infinitely many open \(\GuaranteeThreshold\)-windows of length \(\WindowLength\) occur in the outcome.
  We modify the strategy \(\Strategy[\AlmostSure]\) to construct a strategy \(\Strategy[\epsilon]\) that is winning for both sure-\(\FWMP(\WindowLength, \GuaranteeThreshold)\) and probability(\(1 - \epsilon\))-\(\FWMP(\WindowLength, \LimitSureThreshold)\). 
  
  Let \(N\) be a sufficiently large integer that depends on \(\abs{W_{\AlmostSure\AlmostSureThreshold}}\) (the number of vertices in \(W_{\AlmostSure\AlmostSureThreshold}\)), \(\ProbabilityFunction_{\min}\) (the minimum edge probability in \(W_{\AlmostSure\AlmostSureThreshold}\)), and \(\epsilon\).
  The strategy \(\Strategy[\epsilon]\) is as follows:
  for the first \(N\) steps of the play, the strategy \(\Strategy[\epsilon]\) follows \(\Strategy[\AlmostSure]\), that is, if the token is in \(W_{\AlmostSure\AlmostSureThreshold} \setminus W_{\Sure\AlmostSureThreshold}\), then \(\Strategy[\epsilon]\) follows an almost-sure-attractor strategy to \(W_{\Sure\AlmostSureThreshold}\), and if the token is in \(W_{\Sure\AlmostSureThreshold}\), then \(\Strategy[\epsilon]\) follows a winning strategy for sure-\(\FWMP(\WindowLength, \LimitSureThreshold)\) in \(\subMDP{\MDP}{W_{\Sure\AlmostSureThreshold}}\).
  After \(N\) steps of the play have elapsed, if the token is in \(W_{\Sure\AlmostSureThreshold}\), then \(\Strategy[\epsilon]\) continues to follow the winning strategy for sure-\(\FWMP(\WindowLength, \LimitSureThreshold)\) in the \(\subMDP{\MDP}{W_{\Sure\AlmostSureThreshold}}\) for the rest of the play, and otherwise the token is in \(W_{\AlmostSure\AlmostSureThreshold} \setminus W_{\Sure\AlmostSureThreshold}\), in which case \(\Strategy[\epsilon]\) switches to a strategy that is winning for sure-\(\FWMP(\WindowLength, \GuaranteeThreshold)\) in \(\subMDP{\MDP}{ W_{\Sure\GuaranteeThreshold}}\).
  Such a winning strategy exists from every vertex in \(W_{\Sure\GuaranteeThreshold}\), and thus such a winning strategy also exists from all vertices in \(W_{\AlmostSure\AlmostSureThreshold} \setminus W_{\Sure\AlmostSureThreshold}\).
  
  If a large enough \(N\) is chosen, then with probability at least \(1 - \epsilon\), the token reaches \(W_{\Sure\AlmostSureThreshold}\) in the first \(N\) steps following which satisfaction of \(\FWMP(\WindowLength, \LimitSureThreshold)\) is guaranteed, and even if the token does not reach \(W_{\Sure\AlmostSureThreshold}\), all outcomes of this strategy satisfy \(\FWMP(\WindowLength, \GuaranteeThreshold)\).
  Thus, we have that this strategy \(\Strategy[\epsilon]\) is winning for sure-\(\FWMP(\WindowLength, \GuaranteeThreshold)\) and probability(\(1 - \epsilon\))-\(\FWMP(\WindowLength, \LimitSureThreshold)\).
  Since we can construct such a strategy \(\Strategy[\epsilon]\) for every \(\epsilon > 0\), we have that \(v\) is winning for sure-\(\FWMP(\WindowLength, \GuaranteeThreshold)\)-limit-sure-\(\FWMP(\WindowLength, \LimitSureThreshold)\) in \(\MDP\). 

  \((\Leftarrow)\)
  We show that if \(v\) does not belong to the set \(W_{\AlmostSure\AlmostSureThreshold}\) returned by \Cref{alg:sure-fwmp-limitsure-fwmp}, then \(v\) is not winning for sure-\(\FWMP(\WindowLength, \GuaranteeThreshold)\)-limit-sure-\(\FWMP(\WindowLength, \LimitSureThreshold)\) in \(\MDP\).
  If \(v \not\in W_{\AlmostSure\AlmostSureThreshold}\), then exactly one of the following hold: either \(v \in V \setminus W_{\Sure\GuaranteeThreshold}\) or \(v \in W_{\Sure\GuaranteeThreshold} \setminus W_{\AlmostSure\AlmostSureThreshold}\).
  \begin{itemize}
  \item Suppose \(v \in V \setminus W_{\Sure\GuaranteeThreshold}\).
    Then, the vertex \(v\) is not winning for sure-\(\FWMP(\WindowLength, \GuaranteeThreshold)\) in \(\MDP\), and hence it is also not winning for     sure-\(\FWMP(\WindowLength, \GuaranteeThreshold)\)-limit-sure-\(\FWMP(\WindowLength, \LimitSureThreshold)\) in \(\MDP\).
  \item Suppose \(v \in W_{\Sure\GuaranteeThreshold} \setminus W_{\AlmostSure\AlmostSureThreshold}\).
    Starting from \(v\), if \(\PlayerMain\) ever moves the token out of \(W_{\Sure\GuaranteeThreshold}\) into \(V \setminus W_{\Sure\GuaranteeThreshold}\), then \(\PlayerMain\) cannot ensure satisfaction of sure-\(\FWMP(\WindowLength, \GuaranteeThreshold)\).
    On the other hand, if \(\PlayerMain\) never moves the token out of \(W_{\Sure\GuaranteeThreshold}\), then since \(v\) is not in the winning region \(W_{\AlmostSure\AlmostSureThreshold}\) for almost-sure-\(\FWMP(\WindowLength, \LimitSureThreshold)\) in \(\subMDP{\MDP}{ W_{\Sure\GuaranteeThreshold}}\), there exists some \(0 < \epsilon' \le 1\) such that \(\PlayerMain\) cannot satisfy probability(\(1 - \epsilon'\))-\(\FWMP(\WindowLength, \LimitSureThreshold)\) from \(v\).
    For all \(\epsilon < \epsilon'\), player \(\PlayerMain\) does not have a strategy to keep the token in \(W_{\Sure\GuaranteeThreshold}\) and also satisfy probability(\(1 - \epsilon\))-\(\FWMP(\WindowLength, \LimitSureThreshold)\), in which case the limit-sure-\(\FWMP(\WindowLength, \LimitSureThreshold)\) is not satisfied from \(v\).
    Thus, we have that \(v\) is not winning for sure-\(\FWMP(\WindowLength, \GuaranteeThreshold)\)-limit-sure-\(\FWMP(\WindowLength, \LimitSureThreshold)\) in \(\MDP\).
  \end{itemize}
  This shows the correctness of \Cref{alg:sure-fwmp-limitsure-fwmp}, and also shows how to construct strategies for the satisfaction of sure-\(\FWMP(\WindowLength, \GuaranteeThreshold)\)-limit-sure-\(\FWMP(\WindowLength, \LimitSureThreshold)\) for \(\PlayerMain\) from vertices in \(W_{\AlmostSure\AlmostSureThreshold}\).
\end{proof}

\paragraph*{Memory requirement and complexity for sure-limit-sure satisfaction}
In the rest of this section, we discuss the memory requirement and complexity for the sure-limit-sure satisfaction problems for the \(\FWMPL\) and \(\BWMP\) objectives.

\subparagraph*{Memory requirement of SLS-\(\FWMPL\).}
For every vertex \(v\) in the set \(W_{\AlmostSure\LimitSureThreshold}\) returned by \Cref{alg:sure-fwmp-limitsure-fwmp}, for all \(\epsilon > 0\), we consider the winning strategy \(\Strategy[\epsilon]\) for sure-\(\FWMP(\WindowLength, \GuaranteeThreshold)\)-probability(\(1 - \epsilon\))-\(\FWMP(\WindowLength, \LimitSureThreshold)\) from \(v\) as described in the proof of \Cref{thm:sls-algorithm-correctness}.
The strategy \(\Strategy[\epsilon]\) is defined in two phases:
in the first phase \(\Strategy[\epsilon]\) follows a memoryless almost-sure attractor strategy to \(W_{\Sure\AlmostSureThreshold}\) until either the token reaches \(W_{\Sure\AlmostSureThreshold}\) or \(N\) steps of the play elapse (whichever occurs first), after which the second phase begins where \(\Strategy[\epsilon]\) follows a winning strategy for either sure-\(\FWMP(\WindowLength, \LimitSureThreshold)\) or sure-\(\FWMP(\WindowLength, \GuaranteeThreshold)\) (depending on whether the token is in \(W_{\Sure\AlmostSureThreshold}\) or \(W_{\AlmostSure\AlmostSureThreshold} \setminus W_{\Sure\AlmostSureThreshold}\) respectively at the end of the first phase).
A counter of size \(N\) counts the length of the first phase of the strategy, and memory of size \(\WindowLength\) is sufficient for a winning strategy for sure-\(\FWMP(\WindowLength, \LimitSureThreshold)\) or sure-\(\FWMP(\WindowLength, \GuaranteeThreshold)\) \cite[Theorem 4.4]{DGG25} in the second phase.
Thus, memory size of \(N + \WindowLength \) suffices for \(\Strategy[\epsilon]\).

We show lower and upper bounds for \(N\), respectively 
\(\Omega\left(\frac{1 - \epsilon}{(\ProbabilityFunction_{\min})^{\abs{V}}}\right)\) and \(\bigO\left(\frac{\abs{V} \cdot \log(1/\epsilon)}{(\ProbabilityFunction_{\min})^{\abs{V}}}\right)\) for the minimum number of steps that the almost-sure attractor strategy to \(W_{\Sure\AlmostSureThreshold}\) must be followed for to ensure that \(W_{\Sure\AlmostSureThreshold}\) is reached with probability at least \(1 - \epsilon\).
We show the bounds for a family of conservative MDPs, i.e.,
given \(\abs{V}\) and \(\ProbabilityFunction_{\min}\), the family consists of the MDP in which the probability \(x_{N}\) of reaching \(W_{\Sure\AlmostSureThreshold}\) from the initial vertex in at most \(N\) visits to probabilistic vertices is the minimum amongst all MDPs with \(\abs{V}\) vertices and minimum probability \(\ProbabilityFunction_{\min}\).
If the bounds hold for this family of MDPs, then they hold for all MDPs.
An MDP \(\MDP_{m,p}\) of this family is shown in \Cref{fig:N-epsilon-bound} with \(\abs{V} = 2m + 1\) vertices and \(\ProbabilityFunction_{\min} = p\).
The MDP \(\MDP_{m,p}\) consists of vertices \(u_{1}, u_{2}, \ldots, u_{m+1}\) belonging to \(\PlayerMain\) and vertices \(v_{1}, v_{2}, \ldots, v_{m}\) that are probabilistic.
For all \(1 \le i \le m\), we have that vertex \(u_{i}\) has an out-neighbour of \(v_{i}\), and vertex \(v_{i}\) has out-neighbours \(u_{1}\) and \(u_{i+1}\) with probabilities \(1 - p\) and \(p\) respectively.
Vertices \(u_{1}\) and \(u_{m+1}\) have self-loops with payoffs \(\GuaranteeThreshold\) and \(\AlmostSureThreshold\) respectively.
All other edges in \(\MDP_{m,p}\) have payoff \(\GuaranteeThreshold - 1\).
\begin{figure}[t]
  \centering
  \scalebox{0.8}{
    \begin{tikzpicture}[node distance=1.7cm]
      \node[state] (u0) {\(u_{1}\)};
      \node[random, draw, right of=u0, xshift=1mm] (v0) {\(v_{1}\)};
      
      \node[state, right of=v0] (u1) {\(u_{2}\)};
      \node[random, draw, right of=u1, xshift=-3mm] (v1) {\(v_{2}\)};

      \node[state, right of=v1] (u2) {\(u_{3}\)};
      \node[random, draw, right of=u2, xshift=-3mm] (v2) {\(v_{3}\)};
      
      \node[right of=v2, xshift=-3mm] (dots) {\(\cdots\)};

      \node[random, draw, right of=dots, xshift=-3mm] (vk) {\(v_{m}\)};
      \node[state, right of=vk] (uV) {\(u_{m+1}\)};
      
      \draw 
      (u0) edge (v0)
      (v0) edge[bend right=20] node[above, pos=0.2, yshift=-0.7mm, xshift=-0.1mm]{\small \(\EdgeProbability{1 - p}\)} (u0)
      (v0) edge node[above, pos=0.3]{\small\(\EdgeProbability{p}\)} (u1)
      
      (u1) edge (v1)
      (v1) edge[bend right=35] node[above, pos=0.3]{\small\(\EdgeProbability{1 - p}\)} (u0)
      (v1) edge node[above, pos=0.3]{\small\(\EdgeProbability{p}\)} (u2)
      
      (u2) edge (v2)
      (v2) edge[bend right=35] node[above, pos=0.3]{\small\(\EdgeProbability{1 - p}\)} (u0)
      (v2) edge node[above, pos=0.3]{\small\(\EdgeProbability{p}\)} (dots)
      
      (dots) edge (vk)
      
      (vk) edge node[above, pos=0.3]{\small\(\EdgeProbability{p}\)} (uV)
      (vk) edge[bend right=35] node[above, pos=0.3]{\small \(\EdgeProbability{1 - p}\)} (u0)
      
      (uV) edge[loop right] node[above, pos=0.3]{\(\EdgeValues{\AlmostSureThreshold}{}\)} (uV)
      
      (u0) edge[loop left]node[below, pos=0.3]{\(\EdgeValues{\GuaranteeThreshold}{}\)} (u0)
      ;
    \end{tikzpicture}
  }
  \caption{%
    The self-loop on \(u_{1}\) has payoff \(\GuaranteeThreshold\) and the self-loop on edge \(u_{m+1}\) has payoff \(\AlmostSureThreshold\).
    Every other edge has payoff \(\GuaranteeThreshold - 1\).}
  \label{fig:N-epsilon-bound}
\end{figure}

For all \(\WindowLength \ge 1\), the set \(W_{\Sure\AlmostSureThreshold}\) consists of only \(u_{m+1}\).
In order to reach \(u_{m+1}\) starting from \(u_{1}\), the ``right'' edge \((v_{i}, v_{i+1})\) must be seen at \(m\) consecutive visits to the probabilistic vertices \(v_{1}, v_{2}, \ldots, v_{m}\).
If at any probabilistic vertex along the way the token moves to \(u_{1}\), then the token loses all progress towards reaching \(u_{m+1}\) and has to start over again.
We want to find bounds on the number of times \(N\) the MDP \(\MDP_{m,p}\) visits probabilistic vertices for to ensure that \(u_{m+1}\) is reached with probability at least \(1 - \epsilon\).
Formally, we would like to find \(N\) in terms of \(\abs{V}\), \(\ProbabilityFunction_{\min}\), and \(\epsilon\) such that \(x_{N} \ge 1 - \epsilon\) holds.
It can be seen that the probability \(x_{N}\) of reaching \(u_{m+1}\) starting from \(u_{1}\) in \(\MDP_{m,p}\) after visiting probabilistic vertices at most \(N\) times is equal to the following:
given a coin that lands heads with probability \(p\) and tails with probability \(1 - p\), the probability that a sequence of \(N\) tosses of this coin contains a streak of at least \(m\) consecutive heads.
We can compute \(x_{N}\) using the recurrence relation
\begin{align*}
  x_{N} &= p^{m} + \sum_{i=1}^{m} p^{i-1} \cdot (1 - p) \cdot  x_{N-i}\\
        &= p^{m} + (1 - p) \cdot x_{N-1}
          + p \cdot (1 - p) \cdot  x_{N-2}
          + \cdots +
          p^{m-1} \cdot (1 - p) \cdot  x_{N-m}
\end{align*}
for \(N \ge m\), and with initial values \(x_{0} = x_{1} = \cdots = x_{m - 1} = 0\).
Intuitively, a sequence of \(N\) coin tosses that contains a streak of \(m\) heads either does not contains a tail in the first \(m\) tosses (which happens with probability \(p^{m}\)), or the first tail appears at the \(i^{\text{th}}\) position (this happens with probability \(p^{i-1} \cdot (1-p)\), following which there is a streak of \(m\) heads in the remaining \(N - i\) tosses with probability \(x_{N-i}\)) for \(1 \le i \le m\).
This recurrence relation, however, does not have a nice closed form, so we find approximate lower and upper bounds for \(N\) such that \(x_{N} \ge 1 - \epsilon\).
\begin{itemize}
\item 
  \emph{Upper bound:}
  We divide the sequence of \(N\) coin tosses into blocks of length \(m\) each (we assume that \(m\) divides \(N\)).
  We define \(q\) to be probability that at least one of the \(N/m\) blocks consists of all heads.
  If one of the blocks is all heads, then this is a witness for the streak of at least \(m\) consecutive heads in the sequence of \(N\) coin tosses, and thus, we have that \(x_{N} \ge q\).
  For each block, the probability that the block contains at least one tail is \((1 - p^{m})\), independent of the other blocks.
  The probability \(q\) that at least one of the blocks is all heads is \(1 - (1 - p^{m})^{N / m}\).
  We get that \(x_{N} \ge 1 - (1 - p^{m})^{N / m}\).
  Further, we have the following equivalences for \(0 < \epsilon < 1\):
  {
    \allowdisplaybreaks
    \begin{align*}
      &&  1 - (1 - p^{m})^{N / m} &\ge 1 - \epsilon \\
      \iff &&    (1 - p^{m})^{N / m} &\le \epsilon \\
      \iff &&    \frac{1}{(1 - p^{m})^{N / m}} &\ge \frac{1}{\epsilon} \\
      \iff &&   \frac{N}{m} \log(1/(1 - p^{m})) &\ge \log(1/\epsilon) \\
      \iff &&    N &\ge \frac{m \cdot \log(1/\epsilon)}{\log(1/(1 - p^{m}))}
    \end{align*}
  }%
  Moreover, for values of \(p\) not close to \(1\) (say \(0 < p \le 0.5\)), we have that \(\frac{1}{\log(1 / (1 - p^m))}\) is closely bounded above by \(\frac{2.4}{p^{m}}\).
  
    Suppose \(N\) is large enough (i.e. \(N \ge \frac{2.4 \cdot m \cdot \log(1/\epsilon)}{p^{m}}\)).
    Then we have that \(N \ge \frac{m \cdot \log(1/\epsilon)}{\log(1/(1 - p^{m}))}\), and from the chain of equivalences, we have that \(1 - (1 - p^{m})^{N / m} \ge 1 - \epsilon\).
    Since we also have that \(x_{N} \ge 1 - (1 - p^{m})^{N / m} \), it follows that the desired inequality \(x_{N} \ge 1 - \epsilon\) holds.
  Since \(m = \frac{\abs{V} - 1}{2}\) and \(p = \ProbabilityFunction_{\min}\) for \(\MDP_{m,p}\), we substitute these values in \(\frac{2.4 \cdot m \cdot \log(1/\epsilon)}{p^{m}}\) and get that for all MDPs with \(\abs{V}\) vertices and minimum probability \(\ProbabilityFunction_{\min}\), choosing a value of \(N\) that is at least \(\bigO\left(\frac{\abs{V} \cdot \log(1/\epsilon)}{(\ProbabilityFunction_{\min})^{\abs{V}}}\right)\) is sufficient to ensure that \(W_{\Sure\AlmostSureThreshold}\) is reached with probability at least \(1 - \epsilon\) after \(N\) visits to probabilistic vertices, and this gives us an upper bound for \(N\).
\item
  \emph{Lower bound:}
  If there is a streak of at least \(m\) consecutive heads in the sequence of \(N\) coin tosses, then the streak must begin at the \(i^{\text{th}}\) coin toss for some \(1 \le i \le N - m + 1 \).
  Moreover, the probability of starting a streak of \(m\) heads at the \(i^{\text{th}}\) coin toss is \(p^m\), independent of \(i\). 
  Using the union bound, we get that \(x_{N}\) is bounded above by \((N - m + 1) \cdot p^{m}\), which is bounded above by \(N \cdot p^{m}\).
  That is, we have \(N \cdot p^{m} \ge x_{N}\).
  Thus, in order to satisfy \(x_{N} \ge 1 - \epsilon\), it is necessary that \(N \cdot p^m  \ge 1 - \epsilon\) holds, that is     \(N \ge \frac{1 - \epsilon}{p^{m}}\).
  Substituting \(m = \frac{\abs{V} - 1}{2}\) and \(p = \ProbabilityFunction_{\min}\) gives us a lower bound \(\Omega\left(\frac{1 - \epsilon}{(\ProbabilityFunction_{\min})^{\abs{V}}}\right)\) for \(N\).
\end{itemize}

\subparagraph*{Complexity of SLS-\(\FWMPL\).}
Since sure satisfaction and almost-sure satisfaction of the \(\FWMPL\) objectives can both be checked in polynomial time~\cite{CDRR15, BDOR20}, we have that  \Cref{alg:sure-fwmp-limitsure-fwmp} runs in polynomial time, and we have the following result.

\begin{theorem} \label{thm:fwmp-sls-complexity}
  Sure-limit-sure satisfaction of \(\FWMPL\) objectives is in \(\PTime\).
  The memory required is at least \(\Omega\left(\frac{1 - \epsilon}{(\ProbabilityFunction_{\min})^{\abs{V}}} + \WindowLength\right)\) and at most \(\bigO\left(\frac{\abs{V} \cdot \log(1/\epsilon)}{(\ProbabilityFunction_{\min})^{\abs{V}}} + \WindowLength\right)\).
\end{theorem}

\subparagraph*{Memory requirement of SLS-\(\BWMP\).}
A winning strategy for sure-\(\BWMP(\GuaranteeThreshold)\)-probability(\(1 - \epsilon\))-\(\BWMP(\LimitSureThreshold)\) follows a similar two-phase structure as that for \(\FWMPL\).
The first phase requires a counter of size \(N\) (with the same bounds as above), and memoryless strategies exist for the satisfaction of sure-\(\BWMP(\GuaranteeThreshold)\) or sure-\(\BWMP(\LimitSureThreshold)\)~\cite[Theorem 19]{CDRR15} in the second phase.
Thus, we have that memory of size \(N + 1\) suffices for the sure-limit-sure satisfaction of \(\BWMP\).

\subparagraph*{Complexity of SLS-\(\BWMP\).}
Since sure satisfaction and almost-sure satisfaction of the \(\BWMP\) objectives are in \(\NP \intersection \coNP\)~\cite{CDRR15, BDOR20}, and the algorithm for sure-limit-sure satisfaction for \(\BWMP\) makes one call each to the two algorithms, we have that sure-limit-sure satisfaction of \(\BWMP\) objectives is in \(\PTime^{\NP \intersection \coNP}\), which is the same as \(\NP \intersection \coNP\)~\cite{Sch83}.

\begin{theorem}\label{thm:bwmp-sls-complexity}
  Sure-limit-sure satisfaction of \(\BWMP\) objectives is in \(\NP \intersection \coNP\).
  The memory required is at least \(\Omega\left(\frac{1 - \epsilon}{(\ProbabilityFunction_{\min})^{\abs{V}}}\right)\) and at most 
  \(\bigO\left(\frac{\abs{V} \cdot \log(1/\epsilon)}{(\ProbabilityFunction_{\min})^{\abs{V}}}\right)\).
\end{theorem}

\section{Discussion} \label{sec:conc}

In the context of reactive synthesis, beyond worst-case synthesis has received considerable attention since the last decade.
In this paper, we continued this study for window mean-payoff objectives which strengthen the classical mean-payoff objective by removing the drawback that there may exist arbitrarily long windows with suboptimal mean payoff.
Our algorithms use techniques that are significantly different from existing studies on beyond worst-case synthesis and our results are positive in the sense that the complexities of solving these problems are no more than that of solving the corresponding sure satisfaction problem and the almost-sure satisfaction problem separately.

\subparagraph{Novelty.}
The results presented in this paper are novel in the following ways:
\begin{itemize}
\item In \cite{CDRR15}, it is shown that the winning region for sure-\(\FWMP(\WindowLength, \GuaranteeThreshold)\) can be computed by repeatedly finding winning regions for sure-\(\DirFWMP(\WindowLength, \GuaranteeThreshold)\) and sure-attractors to sets.
  We carefully extend this approach in \Cref{alg:sure-fwmp-almostsure-fwmp}, where we find the winning region for sure-\(\FWMP(\WindowLength, \GuaranteeThreshold)\)-almost-sure-\(\FWMP(\WindowLength, \AlmostSureThreshold)\) by computing the winning regions for sure-\(\DirFWMP(\WindowLength, \AlmostSureThreshold)\), for sure-\(\DirFWMP(\WindowLength, \GuaranteeThreshold)\)-almost-sure-\(\BuchiObj(P)\), and sure-attractors to various sets.
\item In \Cref{alg:sure-dirfwmp-almostsure-buchi}, we show that the techniques that allow us to reduce satisfaction of almost-sure-\(\BuchiObj(T)\) to satisfaction of positive-\(\ReachObj(T)\) can be lifted to give a reduction of
sure-\(\DirFWMP(\WindowLength, \GuaranteeThreshold)\)-almost-sure-\(\BuchiObj(T)\) to sure-\(\DirFWMP(\WindowLength, \GuaranteeThreshold)\)-positive-\(\ReachObj(T)\).
\item \Cref{alg:sure-dirfwmp-positive-reach} uses the \(\IsGoodEdgeAlgo\) procedure to reduce the problem of finding the winning region for sure-\(\DirFWMP(\WindowLength, \GuaranteeThreshold)\)-positive-\(\ReachObj(T)\) in an MDP \(\MDP\) to finding the winning region for sure-\(\DirFWMP(\WindowLength, \GuaranteeThreshold)\) in MDPs \(\MDP_{e}\) intricately created from \(\MDP\) for edges \(e\) in \(\MDP\).
  The MDPs \(\MDP_{e}\) are designed such that \(\PlayerMain\) is forced to choose edges that with positive probability take the token closer to the target \(T\), allowing us to drop the reachability requirement.
  This technique has not appeared elsewhere in the literature to the best of our knowledge. 
\end{itemize}

\subparagraph{Randomized strategies.}
We have only considered deterministic strategies for \(\PlayerMain\) in this work since for both sure-almost-sure satisfaction and sure-limit-sure satisfaction, randomized strategies neither give \(\PlayerMain\) any additional power nor do they allow for winning strategies of smaller memory size.

For the sure-almost-sure satisfaction problem, \Cref{exa:fwmp-memory-lower-bound,exa:bwmp-memory-lower-bound} show that, in general, the memory requirement does not decrease even if we allow randomization for \(\FWMPL\) and \(\BWMP\) objectives respectively.
This is because the token must eventually always take the correct sequence of \(\abs{V}\) vertices for sure-almost-sure satisfaction of \(\FWMPL\) and the correct sequence of \(\abs{V} \cdot \wmax\) edges for the sure-almost-sure satisfaction of \(\BWMP\).

Randomization does not help in the case of sure-limit-sure satisfaction either.
The MDP in \Cref{fig:N-epsilon-bound} shows that the same lower bounds on the memory also continue to hold for randomized strategies. 
A winning strategy of \(\PlayerMain\) must satisfy sure-\(\FWMP(\WindowLength, \GuaranteeThreshold)\) by either reaching and looping at \(u_{m+1}\) or by looping at \(u_{1}\), and must reach \(u_{m+1}\) (where \(\FWMP(\WindowLength, \LimitSureThreshold)\) is satisfied) with probability at least \(1-\epsilon\). 
Such a strategy must surely eventually switch from the memoryless almost-sure attractor strategy to \(u_{m+1}\) to a strategy that loops at \(u_{1}\) if \(u_{m+1}\) has not been reached after several steps. 
The winning strategy must follow the almost-sure attractor strategy to \(u_{m+1}\) for sufficiently many steps to ensure that \(\FWMP(\WindowLength, \LimitSureThreshold)\) is satisfied with a high probability, and this requires a counter of the given size.

\subparagraph{Future work.}
Some interesting problems to study include the combination of satisfying a window mean payoff $\alpha$ surely (resp. almost surely) and a window mean payoff $\beta$ with some threshold probability $p$, or a combination of sure, almost sure, and threshold probability satisfaction, or even a combination of satisfaction with different threshold probabilities, that is, $\GuaranteeThreshold_1$ with probability $p_1$ and $\GuaranteeThreshold_2$ with probability $p_2$.
Another avenue to extend this work is study combinations of winning conditions with different window lengths, for instance, satisfaction of sure-\(\FWMP(\WindowLength_{1}, \GuaranteeThreshold)\)-almost-sure-\(\FWMP(\WindowLength_{2}, \AlmostSureThreshold)\) for \(\WindowLength_{1} \ne \WindowLength_{2}\).
It would also be interesting to see if some of these techniques can be used for the study of other finitary objectives like finitary parity or Streett conditions.
Finally, various combinations of sure, almost-sure, and threshold probability satisfaction can also be studied for classical parity (some of which exist in~\cite{BRR17} and \cite{BGR20}) and classical mean-payoff objectives.

\bibliography{mybib}

\appendix

\clearpage
\section{Appendix: Additional algorithms and results}

\begin{algorithm}
  \caption{\(\SureAttrAlgo(\MDP, T)\)}%
  \label{alg:sure-attr}
  \begin{algorithmic}[1]
    \Require An MDP \(\MDP = ((\Vertices, \Edges), (\VerticesMain, \VerticesRandom), \ProbabilityFunction, \PayoffFunction)\) and a target set \(T \subseteq \Vertices\)
    \Ensure \(\SureAttr{}{T}\), the sure-attractor of \(T\)
    \State \(X \assign T\)
    \ForAll{\(v \in \VerticesMain\)}
    \If{\(\OutNeighbours{v} \intersection T \ne \emptyset\)}
    \Comment{at least one out-neighbour in \(T\)}
    \State \(X \assign X \union \{v\}\)
    \EndIf
    \EndFor
    \ForAll{\(v \in  \VerticesRandom\)}
    \If{\(\OutNeighbours{v} \subseteq T\)}
    \Comment{all out-neighbours in \(T\)}
    \State \(X \assign X \union \{v\}\)
    \EndIf
    \EndFor
    \If{\(X \setminus T = \emptyset\)}
    \State \Return \(T\)
    \Else
    \State \Return \(\SureAttrAlgo(\MDP, X)\)
    \EndIf
  \end{algorithmic}
\end{algorithm}

\begin{algorithm}
  \caption{\(\SureSafeAlgo(\MDP, T)\)}%
  \label{alg:sure-safety}
  \begin{algorithmic}[1]
    \Require An MDP \(\MDP = ((\Vertices, \Edges), (\VerticesMain, \VerticesRandom), \ProbabilityFunction, \PayoffFunction)\) and a safe set \(T \subseteq \Vertices\)
    \Ensure The winning region for sure-\(\SafeObj(T)\) in \(\MDP\)
    \State \(X \assign \emptyset\)
    \ForAll{\(v \in \VerticesRandom\)}
    \If{\(\OutNeighbours{v} \intersection (V \setminus T) \ne \emptyset\)}
    \Comment{at least one out-neighbour in \(\Vertices \setminus T\)}
    \State \(X \assign X \union \{v\}\)
    \EndIf
    \EndFor
    \ForAll{\(v \in \VerticesMain\)}
    \If{\(\OutNeighbours{v} \subseteq \Vertices \setminus T\)}
    \Comment{all out-neighbours in \(\Vertices \setminus T\)}
    \State \(X \assign X \union \{v\}\)
    \EndIf
    \EndFor
    \If{\(X = \emptyset\)}
    \State \Return \(T\)
    \Else
    \State \Return \(\SureSafeAlgo(\MDP, T \setminus X)\)
    \EndIf
  \end{algorithmic}
\end{algorithm}

\begin{algorithm}
  \caption{\(\SASBWMPAlgo(\MDP, \GuaranteeThreshold, \AlmostSureThreshold)\)}%
  \label{alg:sure-bwmp-almostsure-bwmp}
  \begin{algorithmic}[1]
    \Require An MDP \(\MDP = ((\Vertices, \Edges), (\VerticesMain, \VerticesRandom), \ProbabilityFunction, \PayoffFunction)\), sure threshold \(\GuaranteeThreshold\), and almost-sure threshold \(\AlmostSureThreshold\)
    \Ensure The winning region for sure-\(\BWMP(\GuaranteeThreshold)\)-almost-sure-\(\BWMP(\AlmostSureThreshold)\) in \(\MDP\)
    \State \(W_{\Sure\GuaranteeThreshold} \assign \SureBWMPAlgo(\MDP, \GuaranteeThreshold)\)\label{alg-line:sasb-sure-alpha-win}
    \State \(W_{\Sure\AlmostSureThreshold} \assign \SureBWMPAlgo(\subMDP{\MDP}{ W_{\Sure\GuaranteeThreshold}}, \AlmostSureThreshold)\)\label{alg-line:sasb-sure-beta-win}

    \State \(W \assign W_{\Sure\AlmostSureThreshold}\)\label{alg-line:sasb-initialize-w-to-wbeta}
    \Repeat \label{alg-line:sasb-repeat}
    \State \(P \assign \PosCPreAlgo(\subMDP{\MDP}{ W_{\Sure\GuaranteeThreshold}}, W)\)\label{alg-line:sasb-poscpre}
    \State \(W_{\sdab} \assign \SureDirBWMPASBuchiAlgo(\subMDPClosure{\MDP}{W_{\Sure\GuaranteeThreshold} \setminus W}, \GuaranteeThreshold, P)\)\label{alg-line:sasb-sadb}
    \State \(W \assign W \union W_{\sdab}\)\label{alg-line:sasb-include-wsdab}
    \State \(W \assign \SureAttrAlgo(\subMDP{\MDP}{ W_{\Sure\GuaranteeThreshold}}, W)\)\label{alg-line:sasb-sure-attr}

    \Until{\(W_{\sdab} = \emptyset\)}\label{alg-line:sasb-repeat-block-end}
    \State \Return \(W\)\label{alg-line:sasb-return-w}
  \end{algorithmic}
\end{algorithm}

\begin{algorithm}[t]
  \caption{\(\SureDirBWMPASBuchiAlgo(\MDP, \GuaranteeThreshold, T)\)}%
  \label{alg:sure-dirbwmp-almostsure-buchi}
  \begin{algorithmic}[1]
    \Require An MDP \(\MDP = ((\Vertices, \Edges), (\VerticesMain, \VerticesRandom), \ProbabilityFunction, \PayoffFunction)\), threshold \(\GuaranteeThreshold\), and target set \(T\)
    \Ensure The winning region for sure-\(\DirBWMP(\GuaranteeThreshold)\)-almost-sure-\(\BuchiObj(T)\) in \(\MDP\)
    \State \(W_{\sdpr} \assign \SureDirFWMPPosReachAlgo(\MDP, \GuaranteeThreshold, T)\)\label{alg-line:sdbab-sdpr}
    \If{\(V = W_{\sdpr}\)}\label{alg-line:sdbab-until}
    \State \Return \(V\)\label{alg-line:sdbab-return-winning-region}
    \Else
    \State \(S \assign \SureSafeAlgo(\MDP, W_{\sdpr})\)\label{alg-line:sdbab-almost-sure-attr}
    \State \Return \(\SureDirBWMPASBuchiAlgo(\subMDP{\MDP}{S}, \GuaranteeThreshold, T \intersection S)\)\label{alg-line:sdbab-restrict-to-almost-sure-attr}
    \EndIf
  \end{algorithmic}
\end{algorithm}

\begin{algorithm}[t]
  \caption{\(\SureDirBWMPPosReachAlgo(\MDP, \GuaranteeThreshold, T)\)}%
  \label{alg:sure-dirbwmp-positive-reach}
  \begin{algorithmic}[1]
    \Require An MDP \(\MDP = ((\Vertices, \Edges), (\VerticesMain, \VerticesRandom), \ProbabilityFunction, \PayoffFunction)\), threshold \(\GuaranteeThreshold\), and target set \(T\)
    \Ensure The winning region for sure-\(\DirBWMP(\GuaranteeThreshold)\)-positive-\(\ReachObj(T)\) in \(\MDP\)
    \State \(W_{\sd} \assign \SureDirBWMPAlgo(\MDP, \GuaranteeThreshold)\)\label{alg-line:sdbpr-sure-dirbwmp-winning-region}
    \State \(\MDP \assign \subMDP{\MDP}{W_{\sd}}, \quad T \assign T \intersection
    W_{\sd}\)\label{alg-line:sdbpr-restrict-sure-dirbwmp}
    \State \(R \assign T, \quad E_g,\, E_b \assign \emptyset\)\label{alg-line:sdbpr-initialize-sets}
    \While{\(\{e = (s, t) \in \Edges \suchthat s \in \Vertices \setminus T, \, t \in R, \, \text{and } e \in E \setminus (E_g \union E_b)\} \neq \emptyset\)}\label{alg-line:sdbpr-while-loop}
      \State Choose such an edge \(e\) arbitrarily.
      \If{\Call{\(\IsGoodEdgeAlgo\)}{\(e, E_g\)}}\label{alg-line:sdbpr-isgoodedge-procedure-call}
        \State \(E_g \assign E_g \union \{e\}\)\label{alg-line:sdbpr-add-good-edge}, \quad \(R \assign R \union \{s\}\)\label{alg-line:sdbpr-add-good-start-vertex}, \quad \(E_b \assign \emptyset\)\label{alg-line:sdbpr-reset-bad-edges}
      \Else
        \State \(E_b \assign E_{b} \union \{e\}\)\label{alg-line:sdbpr-add-bad-edge}
      \EndIf
    \EndWhile
    \State \Return \(R\)\label{alg-line:sdbpr-return-winning-region}
    \Statex
    \Procedure{\(\IsGoodEdgeAlgo\)}{\(e = (s, t), E_g\)} \label{alg-line:sdbpr-isgoodedge-procedure-define}
      \State Construct MDP \(\MDP_{e}'\) in the following way:
      \State \quad \(V \assign\) the set of vertices in \(\MDP\) 
      \State \quad \(\widetilde{R} \assign \{\tilde{v} \suchthat v \in R \}\), where \(\tilde{v}\) belongs to \(\PlayerMain\) if and only if \(v\) belongs to \(\PlayerMain\).
      \State \quad The set of vertices of \(\MDP_{e}'\) is \(V \union \widetilde{R} \union \{\hat{s}\}\).
      \Statex \quad The vertex \(\hat{s}\) belongs to \(\PlayerMain\) if and only if \(s\) belongs to \(\PlayerMain\).
      \Statex
      \State \quad \(E \assign \) the set of edges in \(\MDP\)
      \State \quad\(\widetilde{E}_g \assign 
      \{(\tilde{u}, \tilde{v}) \suchthat u,v \in R, \, (u, v) \in E_g)\}\)
      \State \quad \(\widetilde{E}_{\Random} \assign \{(\tilde{u}, v) \suchthat u \in (R \setminus T) \intersection \VerticesRandom, \, v \in V, \, (u, v) \in E \setminus E_g)\} \) 
      \State \quad \(\widehat{E}_{\Random} \assign \textbf{if } s \in \VerticesMain \textbf{ then } \emptyset \textbf{ else } \{(\hat{s}, v) \suchthat (s, v) \in \Edges \ \text{and } v \ne t\} \) 
      \State \quad The set of edges of \(\MDP_{e}'\) is \(E \union \widetilde{E}_g \union \widetilde{E}_{\Random} \union \{(\hat{s}, \tilde{t})\} \union \widehat{E}_{\Random}\).
      \Statex \quad The probability and payoff of each edge is preserved from \(\MDP\).
      \Statex
      \State \(\MDP_{e} \assign \) For all \(v \in T\), identify \(\tilde{v}\) with \(v\)  in \(\MDP_{e}'\)
      \State \Return true iff \(\hat{s}\) is winning for sure-\(\DirBWMP(\GuaranteeThreshold)\) in \(\MDP_{e}\) \label{alg-line:sdbpr-isgoodedge-return}
    \EndProcedure
  \end{algorithmic}
\end{algorithm}

\begin{algorithm}[t]
  \caption{\(\SLSBWMPAlgo(\MDP, \GuaranteeThreshold, \AlmostSureThreshold)\)}%
  \label{alg:sure-bwmp-limitsure-bwmp}
  \begin{algorithmic}[1]
    \Require An MDP \(\MDP = ((\Vertices, \Edges), (\VerticesMain, \VerticesRandom), \ProbabilityFunction, \PayoffFunction)\), sure threshold \(\GuaranteeThreshold\), and limit-sure threshold \(\LimitSureThreshold\)
    \Ensure The winning region for sure-\(\BWMP(\GuaranteeThreshold)\)-limit-sure-\(\BWMP(\AlmostSureThreshold)\) in \(\MDP\)
    \State \(W_{\Sure\GuaranteeThreshold} \assign \SureBWMPAlgo(\MDP, \GuaranteeThreshold)\)
    \State \(W_{\AlmostSure\AlmostSureThreshold} \assign \AlmostSureBWMPAlgo(\subMDP{\MDP}{ W_{\Sure\GuaranteeThreshold}}, \LimitSureThreshold)\)
    \State \Return \(W_{\AlmostSure\AlmostSureThreshold}\)
  \end{algorithmic}
\end{algorithm}

\begin{theorem}\label{thm:sasb-correctness}
  \Cref{alg:sure-bwmp-almostsure-bwmp} computes the winning region for sure-\(\BWMP(\GuaranteeThreshold)\)-almost-sure-\(\BWMP(\AlmostSureThreshold)\) in \(\MDP\).
\end{theorem}

\begin{theorem}\label{thm:sdbab-correctness}
  \Cref{alg:sure-dirbwmp-almostsure-buchi} computes the winning region for sure-\(\DirBWMP(\GuaranteeThreshold)\)-almost-sure-\(\BuchiObj(T)\).
\end{theorem}

\begin{theorem}
  \Cref{alg:sure-dirbwmp-positive-reach} computes the winning region for sure-\(\DirBWMP(\GuaranteeThreshold)\)-positive-\(\ReachObj(T)\) in \(\MDP\).
\end{theorem}

\end{document}